\documentclass[a4paper,UKenglish,cleveref, autoref, thm-restate, nolipics]{lipics-v2021}
\usepackage{amsmath}
\usepackage{amssymb}
\usepackage{amsthm}
\usepackage{mathpartir}
\usepackage[T1]{fontenc}
\usepackage[utf8]{inputenc}
\usepackage[mathscr]{eucal}
\usepackage{bm}
\usepackage[all,cmtip,2cell]{xy}
\usepackage{enumerate}
\usepackage{graphicx}
\usepackage{stmaryrd}
\usepackage{rotating}
\usepackage{enumitem}
\usepackage{algorithm}
\usepackage{algorithmic}
\usepackage{hyperref}

\allowdisplaybreaks

\usepackage{makecell}

\newcounter{para}[section]

\usepackage{tikz}
\usetikzlibrary{arrows,snakes,automata,backgrounds,petri, positioning,calc,cd, automata}
\usetikzlibrary{shapes.geometric,shapes.misc,matrix,arrows.meta,decorations.markings}
\tikzset{x=1em, y=1.5ex, baseline=-0.5ex}
\tikzset{ihbase/.style={inner sep=0,circle,draw,fill=lightgray,minimum size=0.4em,node contents={}}}
\tikzset{ihblack/.style={ihbase,fill=black}}
\tikzset{ihwhite/.style={ihbase,fill=white}}
\tikzset{mat/.style={draw,fill=white,rectangle,node font=\scriptsize}}
\tikzset{ha/.style={mat,rounded rectangle,rounded rectangle left arc=none}}
\tikzset{haop/.style={mat,rounded rectangle,rounded rectangle right arc=none}}
\tikzset{blackha/.style={mat,rounded rectangle,rounded rectangle left arc=none,font=\color{white},fill=black}}
\tikzset{blackhaop/.style={mat,rounded rectangle,rounded rectangle right arc=none,font=\color{white},fill=black}}
\tikzset{anti/.style={inner sep=0,isosceles triangle,fill=black,draw=black, minimum width=0.75em, node contents={}}}
\tikzset{antiop/.style={anti,shape border rotate=180}}
\tikzset{antisq/.style={inner sep=0,rectangle,fill=black, minimum height=1em, minimum width=0.6em, node contents={}}}
\tikzset{count/.style={above,inner ysep=0.15em,font=\scriptsize}}
\tikzset{axiom/.style={above,font=\small}}
\tikzset{dir/.style={-Latex}}
\tikzset{st/.style={decoration={markings,
    mark={at position 0.5 with {\draw (0, 2pt) to (0, -2pt);}}},
    postaction=decorate}}

\newcommand{\objr}{\blacktriangleright}
\newcommand{\objl}{\blacktriangleleft}
\newcommand{\PMet}{\mathsf{PMet}}

\newcommand{\idright}{
\tikzset{x=1em, y=2.1ex}
\begin{tikzpicture}
	\begin{pgfonlayer}{nodelayer}
		\node [style=none] (0) at (-1.5, 0) {};
		\node [style=none] (1) at (0, 0) {};
		\node [style=none] (2) at (1.25, 0) {};
	\end{pgfonlayer}
	\begin{pgfonlayer}{edgelayer}
		\draw [->] (0.center) to (1.center);
		\draw (1.center) to (2.center);
	\end{pgfonlayer}
\end{tikzpicture}
}
\tikzset{x=1em, y=1.5ex}
}
\newcommand{\alphabet}{\Sigma}

\newcommand{\sem}[1]{\left\llbracket{#1}\right\rrbracket}

\newcommand{\disteq}[1]{\equiv_{#1}}

\newcommand{\arrowright}{
\tikzset{x=1em, y=2.1ex}
\begin{tikzpicture}
	\begin{pgfonlayer}{nodelayer}
		\node [style=none] (0) at (-1.5, 0) {};
		\node [style=none] (1) at (0, 0) {};
		\node [style=none] (2) at (1.25, 0) {};
	\end{pgfonlayer}
	\begin{pgfonlayer}{edgelayer}
		\draw [->] (0.center) to (1.center);
		\draw (1.center) to (2.center);
	\end{pgfonlayer}
\end{tikzpicture}
}
\tikzset{x=1em, y=1.5ex}
}
\newcommand{\arrowleft}{
\tikzset{x=1em, y=2.1ex}
\begin{tikzpicture}
	\begin{pgfonlayer}{nodelayer}
		\node [style=none] (1) at (-1.25, 0) {};
		\node [style=none] (2) at (0, 0) {};
		\node [style=none] (3) at (1.25, 0) {};
	\end{pgfonlayer}
	\begin{pgfonlayer}{edgelayer}
		\draw (1.center) to (2.center);
		\draw [->] (3.center) to (2.center);
	\end{pgfonlayer}
\end{tikzpicture}
}
\tikzset{x=1em, y=1.5ex}
}
\newcommand{\fv}{\mathsf{fv}}

\newcommand{\Qp}{\mathbb{Q}_{\geq 0}}

\newcommand{\genericcounitn}[2]{
  \tikz \draw (0, 0) -- node[count] {#2} (1, 0) node[ihbase,#1];
}

\newcommand{\diagbox}[3]{
\begin{tikzpicture}
	\begin{pgfonlayer}{nodelayer}
		\node [style=none] (6) at (1.75, 0) {};
		\node [style=none] (7) at (0.75, 0) {};
		\node [style=basic rounded box] (10) at (0, 0) {$#1$};
		\node [style=none] (11) at (1.5, 0.5) {\scriptsize $#3$};
		\node [style=none] (12) at (-1.5, 0.5) {\scriptsize $#2$};
		\node [style=none] (13) at (-1.75, 0) {};
		\node [style=none] (14) at (-0.75, 0) {};
	\end{pgfonlayer}
	\begin{pgfonlayer}{edgelayer}
		\draw [in=180, out=0, looseness=1.25] (7.center) to (6.center);
		\draw [in=180, out=0, looseness=1.25] (13.center) to (14.center);
	\end{pgfonlayer}
\end{tikzpicture}
}

\newcommand{\Bcomult}{
\tikzset{x=1em, y=2.1ex}
\InputIfFileExists{lr-copy.tikz}{}{\input{./tikz/lr-copy.tikz}}
\tikzset{x=1em, y=1.5ex}
}

\newcommand{\Bcounit}{
\tikzset{x=1em, y=2.1ex}
\begin{tikzpicture}
	\begin{pgfonlayer}{nodelayer}
		\node [style=black] (37) at (0.75, 0) {};
		\node [style=none] (43) at (0.25, 0) {};
		\node [style=none] (44) at (-0.5, 0) {};
	\end{pgfonlayer}
	\begin{pgfonlayer}{edgelayer}
		\draw (43.center) to (37);
		\draw [->] (44.center) to (43.center);
	\end{pgfonlayer}
\end{tikzpicture}
}
\tikzset{x=1em, y=1.5ex}
}

\newcommand{\Bmult}{
\tikzset{x=1em, y=2.1ex}
\InputIfFileExists{lr-merge.tikz}{}{\input{./tikz/lr-merge.tikz}}
\tikzset{x=1em, y=1.5ex}
}

\newcommand{\Bunit}{
\tikzset{x=1em, y=2.1ex}
\begin{tikzpicture}
	\begin{pgfonlayer}{nodelayer}
		\node [style=black] (37) at (-0.5, 0) {};
		\node [style=none] (43) at (0.25, 0) {};
		\node [style=none] (44) at (0.75, 0) {};
	\end{pgfonlayer}
	\begin{pgfonlayer}{edgelayer}
		\draw [->] (37) to (43.center);
		\draw (44.center) to (43.center);
	\end{pgfonlayer}
\end{tikzpicture}
}
\tikzset{x=1em, y=1.5ex}
}
\newcommand{\Bunitn}[1]{\Bunit^{\!\!\!\! #1}}

\tikzset{x=1em, y=1.5ex}
}

\tikzset{x=1em, y=1.5ex}
}
\newcommand{\Wcounitn}[1]{\genericcounitn{ihwhite}}

\newcommand\idzero{
\tikzset{x=1em, y=2.1ex}
\InputIfFileExists{empty-diag.tikz}{}{\input{./tikz/empty-diag.tikz}}
\tikzset{x=1em, y=1.5ex}
} 

\tikzset{x=1em, y=1.5ex}
}

\newcommand\scalar[1]{
  \tikz {
    \node[ha] (ha) {$#1$};
    \draw (ha.west) -- ++(-0.75, 0);
    \draw (ha.east) -- ++(0.75, 0);
  }
}

\tikzset{x=1em, y=1.5ex}
}

\pgfdeclarelayer{edgelayer}
\pgfdeclarelayer{nodelayer}
\pgfsetlayers{edgelayer,nodelayer,main}

\definecolor{light-gray}{gray}{.5}
\tikzstyle{right arrow} = [->]
\tikzstyle{none}=[inner sep=0pt]
\tikzstyle{plain}=[inner sep=0pt]
\tikzstyle{black}=[circle, draw=black, fill=black, inner sep=0pt, minimum size=4pt]
\tikzstyle{black-faded}=[circle, draw=light-gray, fill=light-gray, inner sep=0pt, minimum size=4pt]
\tikzstyle{white}=[circle, draw=black, fill=white, inner sep=0pt, minimum size=4.5pt]
\tikzstyle{white-faded}=[circle, draw=light-gray, fill=white, inner sep=0pt, minimum size=4.5pt]
\tikzstyle{delay}=[fill=black, regular polygon, regular polygon sides=3,rotate=-90, scale=.55]
\tikzstyle{delay-op}=[fill=black, regular polygon, regular polygon sides=3,rotate=90, scale=.55]
\tikzstyle{reg}=[draw, fill=white, rounded rectangle, rounded rectangle left arc=none, minimum height=1.2em, minimum width=1.4em, node font={\scriptsize}]
\tikzstyle{coreg}=[draw, fill=white, rounded rectangle, rounded rectangle right arc=none, minimum height=1.2em, minimum width=1.4em, node font={\scriptsize}]
\tikzstyle{basicb}=[draw, fill=white, rectangle, rounded corners, minimum height=1.6em, minimum width=1.4em]
\tikzstyle{basic box}=[draw, fill=white, rectangle, rounded corners, minimum height=1.6em, minimum width=1.4em]
\tikzstyle{smallb}=[draw, fill=white, rectangle, rounded corners, minimum height=1.2em, minimum width=1.4em, node font={\scriptsize}]
\tikzstyle{rcoreg}=[draw=red, fill=white, rounded rectangle, rounded rectangle right arc=none, minimum height=1.2em, minimum width=1.4em, node font={\scriptsize}]
\tikzstyle{regb}=[draw, fill=black, rounded rectangle, rounded rectangle left arc=none, minimum height=1.2em, minimum width=1.4em, node font={\scriptsize}]
\tikzstyle{regbw}=[draw, left color=black, right color=white, middle color=white, rounded rectangle, rounded rectangle left arc=none, minimum height=1.2em, minimum width=1.4em, node font={\scriptsize}]
\tikzstyle{regwb}=[draw, left color=white, right color=black, middle color=white, rounded rectangle, rounded rectangle left arc=none, minimum height=1.2em, minimum width=1.4em, node font={\scriptsize}]
\tikzstyle{coregb}=[draw, fill=black, rounded rectangle, rounded rectangle right arc=none, minimum height=1.2em, minimum width=1.4em, node font={\scriptsize}]
\tikzstyle{coregbw}=[draw, left color=black, right color=white, middle color=white, rounded rectangle, rounded rectangle right arc=none, minimum height=1.2em, minimum width=1.4em, node font={\scriptsize}]
\tikzstyle{coregwb}=[draw, left color=white, right color=black, middle color=white, rounded rectangle, rounded rectangle right arc=none, minimum height=1.2em, minimum width=1.4em, node font={\scriptsize}]
\tikzstyle{rn}=[circle, draw=red, fill=red, inner sep=0pt, minimum size=4pt]
\tikzstyle{wrn}=[circle, draw=red, fill=white, inner sep=0pt, minimum size=4pt]
\tikzstyle{place}=[circle, draw=black, fill=white, inner sep=0pt, minimum size=9pt]
\tikzstyle{act}=[circle, draw=black, fill=white, inner sep=0pt, minimum size=4.5pt]
\tikzstyle{coact}=[draw, fill=white, rounded rectangle, rounded rectangle right arc=none, minimum height=.7em, minimum width=.9em, node font={\scriptsize}]
\tikzstyle{basic rounded box}=[draw, fill=white, rectangle, rounded corners, minimum height=1.6em, minimum width=1.4em]
\tikzstyle{small rounded box}=[draw, fill=white, rectangle, rounded corners, minimum height=1.2em, minimum width=1.4em, node font={\scriptsize}]
\tikzset{
BWmatrix/.pic={
    \coordinate (center) at (0,0);
    \filldraw[fill=white, draw=black, line width=1pt] (.5,0) 
        [rounded corners=14pt] -- (1,0) 
        [rounded corners=14pt] -- (1,1)
        [rounded corners=0pt] -- (.5,1) 
        [rounded corners=0pt] -- cycle;
    \filldraw[fill=black, draw=black, line width=1pt] (0,0) 
        -- (.5,0) 
        -- (.5,1)
        -- (0,1) 
        -- cycle;
   },
pics/BWmatrix/.default=0.2
}

\tikzstyle{pl}=[circle,thick,draw=black!75,fill=white,minimum size=17pt]
\tikzstyle{port}=[circle, fill,inner sep=1.2pt]
\tikzstyle{transition}=[rectangle,thick,draw=black!75,
  			  fill=black!20,minimum size=7pt]

\tikzstyle{arrow}=[->]

\newcommand{\dbox}[3]{
\begin{tikzpicture}
	\begin{pgfonlayer}{nodelayer}
		\node [style=none] (6) at (1.5, 0) {};
		\node [style=basicb] (10) at (0, 0) {$#1$};
		\node [style=none] (11) at (-1.25, 0) {};
		\node [style=none] (13) at (-2, 0) {};
		\node [style=none] (14) at (2, 0) {};
		\node [style=none] (15) at (1.75, 0.75) {\scriptsize $#3$};
		\node [style=none] (16) at (-1.75, 0.75) {\scriptsize $#2$};
	\end{pgfonlayer}
	\begin{pgfonlayer}{edgelayer}
		\draw [->] (13.center) to (11.center);
		\draw (6.center) to (14.center);
		\draw (11.center) to (10);
		\draw [->] (10) to (6.center);
	\end{pgfonlayer}
\end{tikzpicture}
}

\newcommand{\traceaction}[5]{
\begin{tikzpicture}
	\begin{pgfonlayer}{nodelayer}
		\node [style=none] (0) at (-0.75, 1.25) {};
		\node [style=none] (1) at (0.75, 1) {};
		\node [style=none] (2) at (-0.25, 1.75) {};
		\node [style=none] (3) at (-0.75, -0.25) {};
		\node [style=none] (4) at (0.75, -0.25) {};
		\node [style=none] (5) at (-0.25, -0.75) {};
		\node [style=none] (6) at (-0.75, 1) {};
		\node [style=none] (7) at (0.75, 1.25) {};
		\node [style=none] (8) at (2.5, 0) {};
		\node [style=none] (9) at (0.75, 0) {};
		\node [style=none] (10) at (0.25, -0.75) {};
		\node [style=none] (11) at (0.25, 1.75) {};
		\node [style=none] (12) at (0, 0.5) {$#1$};
		\node [style=none] (13) at (4, 0.5) {\scriptsize $#3$};
		\node [style=none] (14) at (2.5, 2.5) {};
		\node [style=none] (15) at (-1.5, 2.5) {};
		\node [style=none] (16) at (-2.75, 0.5) {\scriptsize $#2$};
		\node [style=none] (17) at (-1.75, 0) {};
		\node [style=none] (18) at (-0.75, 0) {};
		\node [style=none] (20) at (2.5, 1) {};
		\node [style=none] (21) at (3.5, 2.75) {\scriptsize $#4$};
		\node [style=none] (22) at (-1.5, 1) {};
		\node [style=reg] (23) at (1.75, 1) {$#5$};
		\node [style=none] (24) at (4.25, 0) {};
		\node [style=none] (25) at (-3, 0) {};
	\end{pgfonlayer}
	\begin{pgfonlayer}{edgelayer}
		\draw [semithick, in=0, out=-90] (4.center) to (10.center);
		\draw [semithick, in=-90, out=180] (5.center) to (3.center);
		\draw [semithick, in=180, out=90] (0.center) to (2.center);
		\draw [semithick, in=90, out=0] (11.center) to (7.center);
		\draw [semithick] (2.center) to (11.center);
		\draw [semithick] (7.center) to (4.center);
		\draw [semithick] (10.center) to (5.center);
		\draw [semithick] (3.center) to (0.center);
		\draw (15.center) to (14.center);
		\draw (6.center) to (22.center);
		\draw (1.center) to (20.center);
		\draw [->, bend right=90, looseness=2.25] (20.center) to (14.center);
		\draw [->, bend right=90, looseness=2.25] (15.center) to (22.center);
		\draw (8.center) to (24.center);
		\draw [->] (9.center) to (8.center);
		\draw (17.center) to (18.center);
		\draw [->] (25.center) to (17.center);
	\end{pgfonlayer}
\end{tikzpicture}}

\newcommand{\myeq}[1]{\mathrel{\overset{\makebox[0pt]{\mbox{\normalfont\tiny\sffamily (#1)}}}{=}}}
\newcommand{\myleq}[1]{\mathrel{\overset{\makebox[0pt]{\mbox{\normalfont\tiny\sffamily (#1)}}}{\leq}}}

\newcommand{\cat}[1]{\mathcal{#1}}

\newcommand{\N}{\mathbb{N}}

\newcommand{\R}{\mathbb{R}}

\newcommand{\freeP}[1]{\mathsf{T}_{\scriptscriptstyle #1}} 

\newcommand{\Obj}{\mathcal{O}}
\newcommand{\Morph}{\mathcal{M}}

\newcommand{\id}{\mathrm{id}}

\newcommand{\from}{\mathrel{:}\,}

\newcommand{\proptimes}{\oplus}
\newcommand{\adjto}{\,\lower1pt\hbox{$\dashv$}\,}

\newcommand{\Int}[1]{\mathbf{Int}(#1)}

\newcommand{\Pfin}{P_{\text{fin}}}
\newcommand{\Tr}{\mathsf{Tr}}

\newcommand{\Syn}{\mathsf{Syn}}

\newcommand{\SMCeq}{=}

\newcommand{\Signature}{\mathcal{S}}

\makeatletter
\def\moverlay{\mathpalette\mov@rlay}
\def\mov@rlay#1#2{\leavevmode\vtop{%
\baselineskip\z@skip \lineskiplimit-\maxdimen
\ialign{\hfil$#1##$\hfil\cr#2\crcr}}}
\makeatother

\newcommand\twarr[2]{%
\mathrel{\mathop{\moverlay{\scriptstyle\xrightarrow{\,#1\,}\cr{\lower.2em\hbox{$\scriptstyle{}_{#2}$}}}}}}
\newcommand\twarrw[2]{%
\mathrel{\mathop{\moverlay{\scriptstyle\Longrightarrow\cr{\lower-.6em\hbox{$\scriptstyle{}_{#1}$}}
\cr{\lower.3em\hbox{$\scriptstyle{}_{#2}$}}}}}}

\newcommand{\dtransw}[2]{\raise1pt\hbox{$\;\twarrw{#1}{#2}\;$}}

\newcommand{\diagregexp}[1]{
\begin{tikzpicture}
	\begin{pgfonlayer}{nodelayer}
		\node [style=none] (0) at (1.5, 0) {};
		\node [style=rcoreg] (1) at (0, 0) {{\color{red} $e$}};
	\end{pgfonlayer}
	\begin{pgfonlayer}{edgelayer}
		\draw [red] (1) to (0.center);
	\end{pgfonlayer}
\end{tikzpicture}}

\usepackage{todonotes}
\usepackage{caption}
\usepackage{interval}
\newcommand{\inl}{\mathsf{inl}}
\newcommand{\inr}{\mathsf{inr}}
\newcommand{\lc}{\langle}
\newcommand{\rc}{\rangle}
\newcommand{\Expr}{\mathsf{Exp}}
\newcommand{\RegBeh}{\mathsf{RegBeh}}
\newcommand*\circlednum[1]{\tikz[baseline=(char.base)]{
            \node[shape=circle,draw,inner sep=1pt] (char) {#1};}}
\newcommand{\tr}[1]{\mathrel{\!
    \raisebox{-2pt}{
        \(\xrightarrow{#1}\)
    }\!
}}

\hideLIPIcs  

\title{A Diagrammatic Axiomatisation of Behavioural Distance of Nondeterministic Processes} 

\titlerunning{Behavioural Distance of Nondeterministic Processes, Diagrammatically} 

\author{Wojciech Różowski}{Lean FRO, United Kingdom\and \url{https://wkrozowski.github.io}}{wojciech@lean-fro.org}{https://orcid.org/0000-0002-8241-7277}{}
\author{Robin Piedeleu}{Department of Computer Science, University College London, United Kingdom\and \url{https://piedeleu.com}}{r.piedeleu@ucl.ac.uk}{https://orcid.org/0000-0002-3945-2704}{}
\author{Alexandra Silva}{Department of Computer Science, Cornell University, United States of America\and \url{https://alexandrasilva.org}}{alexandra.silva@cornell.edu}{https://orcid.org/0000-0001-5014-9784}{}
\author{Fabio Zanasi}{Department of Computer Science, University College London, United Kingdom\and \url{https://fzanasi.github.io}}{f.zanasi@ucl.ac.uk}{https://orcid.org/0000-0001-6457-1345}{}
\authorrunning{Różowski et al.} 

\Copyright{Wojciech Różowski, Robin Piedleu, Alexandra Silva, and Fabio Zanasi} 


\category{} 

\relatedversion{} 

\funding{ This work was partially supported by ERC grant Autoprobe (no. 101002697; Różowski, Piedeleu and Silva) and ARIA Safeguarded AI programme (Zanasi).}

\acknowledgements{The authors are grateful to Filippo Bonchi, Florence Clerc, Clemens Kupke, and Ralph Sarkis for valuable discussions that refined the paper’s presentation.}

\nolinenumbers 

\EventEditors{John Q. Open and Joan R. Access}
\EventNoEds{2}
\EventLongTitle{42nd Conference on Very Important Topics (CVIT 2016)}
\EventShortTitle{CVIT 2016}
\EventAcronym{CVIT}
\EventYear{2016}
\EventDate{December 24--27, 2016}
\EventLocation{Little Whinging, United Kingdom}
\EventLogo{}
\SeriesVolume{42}
\ArticleNo{23}

\begin{document}

\maketitle

\begin{abstract}
Behavioural distances provide a quantitative approach to comparing the states of transition systems, moving beyond traditional Boolean notions of equivalence. In this paper, we develop a sound and complete axiomatisation of behavioural distance for nondeterministic processes using Milner's charts, a model that generalises finite-state automata by incorporating variable outputs. Charts provide a compelling setting for studying behavioural distances because they shift the focus from language equivalence to bisimilarity. Their axiomatic study lays the groundwork for quantitative analysis of more expressive models, such as weighted transition systems. 

To formalise this approach, we adopt string diagrams as our syntax of choice. String diagrams closely mirror the graphical structure of charts, while providing a rigorous formalism that supports inductive reasoning and compositional semantics. Unlike traditional algebraic syntaxes, which require additional mechanisms such as binders and substitution, string diagrams offer a variable-free representation where recursion naturally decomposes into simpler components. This makes them well-suited for reasoning about behavioural distances and aligns with broader efforts to axiomatise automata-theoretic equivalences through a unified diagrammatic framework.
\end{abstract}
\section{Introduction}
In Theoretical Computer Science, it is customary to model computations as transition systems. To facilitate formal analysis of such models, considerable effort has been devoted to developing expressive syntaxes 
and compositional reasoning techniques. Notable examples include Kleene algebra~\cite{Kozen:1994:Completeness} and its extensions~\cite{katb, netkat,cnetkat}, as well as a vast body of work on process calculi~\cite{Milner:1984:Complete,Bergstra:2001:Handbook,spectrum}. A particularly important feature of such approaches is the presence of an \emph{axiomatisation}---a set of equations that relate syntactic terms that represent semantically equivalent behaviours. When an axiomatisation is available, one may reason about model behaviour via syntactic manipulation of terms, which is particularly well-suited for implementation and automation. 

In many contexts, especially when dealing with probabilistic or quantitative models, focussing on exact equivalence of behaviours is  too restrictive. Instead, it is often more meaningful to measure how far apart the behaviours of two terms are. This has motivated the development of \emph{behavioural distances}, which endow the state-spaces of transition systems with (pseudo)metric structures quantifying the dissimilarity of states~\cite{Breugel:2001:Towards,Desharnais:2004:Metrics,Baldan:2018:Coalgebraic}, and  \emph{quantitative equational theories}~\cite{Mardare:2016:Quantitative,Mio:2024:Universal}, which replace equational judgements $s = t$ between terms with quantitative ones of the form $s =_{\varepsilon} t$ capturing ``the distance between $s$ and $t$ is at most $\varepsilon$''. 

Behavioural distances have mostly been studied for probabilistic systems~\cite{Desharnais:2004:Metrics,Breugel:2001:Towards}. More recently, there has been growing interest in understanding distances in a general categorical framework and how this would yield coarser notions of equivalence for a variety of systems~\cite{Baldan:2018:Coalgebraic}. The instantiations of that framework, and in particular axiomatisations of distances, are largely unexplored, with the exception of one for deterministic automata~\cite{Rozowski:2024:Complete}. 

In this paper, we take a step further and investigate axiomatisations of behavioural distance for a {\em nondeterministic} model of computation, known as \emph{charts}~\cite{Milner:1984:Complete}. Originally introduced by Milner, charts extend \emph{finite-state nondeterministic automata} (NFA) by replacing the notion of acceptance with variable outputs. Intuitively, the distance between two charts can be quantified by, roughly, the number of steps after which their behaviours disagree---\emph{i.e.} are no longer bisimilar. This seemingly small generalisation from deterministic finite automata provides a range of challenges, stemming from the fact that the presence of non-determinism moves the semantics from language to bisimilarity, while at the same time representing a crucial step towards weighted transition systems~\cite{Larsen:2011:Metrics}, where the general theory of behavioral distances and axiomatisations thereof is relatively underexplored. 


The central contribution of this paper is an inference system for reasoning about behavioural distances of behaviours of Milner's charts. We demonstrate its \emph{soundness}~(\Cref{thm:soundness}) and \emph{completeness}~(\Cref{thm:completeness}). On the way, we gather several contributions of independent interest. First, we instantiate the abstract framework of behavioural distances in the concrete case of charts. We organise such behaviours as a symmetric monoidal category, in which they may be composed \emph{sequentially} and \emph{in parallel}. We do so relying on rich structures associated with charts, such as Conway Theories~\cite{Bloom:1993:Iteration,Esik:1999:Group}. Second, as one of the steps in the soundness argument, we give a concrete characterisation of behavioural distance between charts via Hennessy and Milner's stratification of bisimilarity~\cite{hennessy:1985:algebraic}. Finally, the completeness argument makes use of tools from fixpoint theory to simplify the calculation of behavioural distance to the point it can be mimicked via syntactic manipulation.

The syntax and equations of our complete axiomatic theory are given in terms of \emph{string diagrams}, the two-dimensional language of monoidal categories~\cite{Selinger_2010,piedeleu2023introduction}. The pictorial representation of string diagrams provides an intuitive understanding of how information flows and is exchanged between components within a system. For this reason, they have been increasingly popular as a formal language for computations and processes in areas such as quantum theory~\cite{Coecke:2008:Interacting}, concurrency~\cite{Bonchi:2019:Diagrammatic}, probabilistic programming~\cite{Piedeleu:2024:Complete}, and digital circuits~\cite{Ghica:2022:Full}.  
\begin{figure}
\begin{align*}
\scalebox{0.8}{

\tikzset{x=1em, y=2.1ex}
\InputIfFileExists{ex-intro.tikz}{}{\input{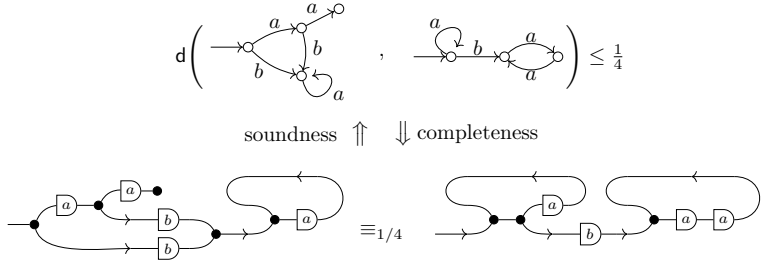}}
\tikzset{x=1em, y=1.5ex}

}
\end{align*}
\caption{Two charts at distance $\frac{1}{4}$ and their corresponding representations as string diagrams. In Milner's syntax, the chart on the left can be presented as $a.(a.0 + b.\mu v_1.a.v_1)+b.\mu v_1.a.v_1$, while the chart on the right corresponds to $\mu v_2. (a.v_2 + b.\mu v_1. a.a.v_1)$. }
\label{fig:intro}
\end{figure}
There are several reasons to favour string diagrams as our syntax of choice. 
First, they closely resemble the usual graphical representation of the transition structure of charts, while constituting a formal syntax that supports inductive reasoning and to which we can assign semantics formally.
Moreover, as Milner observed~\cite{Milner:1984:Complete}, the standard algebraic syntax of regular expressions is not expressive enough to capture all chart behaviour~\cite{Grabmayer:2022:Milner}. His solution introduced a more complex syntax with binders and names,  later studied in the process algebra community for various models, including probabilistic~\cite{Stark:2000:Complete} and quantitative~\cite{Jensen:2020:Complete} ones. In contrast, string diagrams offer a variable-free approach, eliminating the need to define substitution and recursion as a primitive operation (the latter is decomposed into simpler components). 
Finally, using string diagrams aligns our work with a broader research programme aimed at axiomatising various notions of equivalence in automata theory through a unified diagrammatic syntax~\cite{piedeleu2023finite,antoinecsl2025}. 



\noindent
\textbf{Outline. } In \Cref{sec:operations}, we introduce charts, as well as their associated notions of behavioural equivalence and distances. Then, in \Cref{sec:syntax}, we introduce the syntax of our diagrammatic calculus, for which we construct the semantics in \Cref{sec:semantics}. Next, in \Cref{sec:axioms} we present a (quantitative) equational inference system for reasoning about distances of the denotations of the terms of our calculus; we also prove its soundness and study one example in more detail. \Cref{sec:completeness} contains the main technical result of the paper, namely completeness for the proposed behavioural distance between charts. We wrap up in \Cref{sec:discussion} where we review related literature, and sketch directions for future work.  
 
\section{Preliminaries}\label{sec:operations}
In this section, we briefly review Milner's charts~\cite{Milner:1984:Complete}, whose behaviours are the central semantic object of this paper. Then, we instantiate the abstract framework of coalgebraic behavioural metrics~\cite{Baldan:2018:Coalgebraic} to the concrete case of charts. 

\noindent
\textbf{(Pre)charts and their algebraic operations. } Fix a set $V=\{v_1, v_2, \dots\}$ of \emph{variables} and $\alphabet$ of \emph{letters} respectively. A prechart is a triple $(Q,E,D)$, where $Q$ is a set of states, $D \subseteq Q \times \Sigma \times Q$ a finite labelled transition relation and $E \subseteq Q \times V$ is a finite output relation. Precharts can be thought as a generalisation of nondeterministic automata, where instead of acceptance, we deal with the notion of outputs.
 Given a $q \in Q$, we will often write $D(q) = \{(a, q') \mid (q, a, q') \in D\}$ and $E(q) = \{v \mid (q,v) \in E\}$. Moreover, when $D$ and $E$ are clear from the context, we will write $q \tr{a} q' \iff (q,a,q') \in D$ and $q \rhd v \iff (q,v) \in E$. A chart $C$ is a quadruple $(Q, s, D, E)$, where $(Q,D,E)$ is a prechart and $s \in Q$ is a distinguished start node. We call a chart finite if $Q$ is finite. There are several operations on charts that will be of interest in this paper.

\begin{itemize}[topsep=0pt]
	\item \textbf{Empty chart.} We set $0 := (\{s\},s, \emptyset, \emptyset)$.
	\item \textbf{Variable.} Given $v \in V$, we define $v := (\{s\},s,\emptyset, \{(s,v)\})$.
	\item \textbf{Prefix.} Let $C=(Q,s,D,E)$ be a chart and $a \in \Sigma$. We define $a.C := (Q \cup \{s'\}, s', D \cup \{(s',a,s)\}, E)$, where $s' \notin Q$.
	\item \textbf{Nondeterministic choice.} Let $C_i = (Q_i,s_i,D_i,E_i)$, $i \in \{1,2\}$. For simplicity of presentation, we assume $Q_1, Q_2$ to be disjoint and $s \notin Q_1 \cup Q_2$. We define $$C_1 + C_2 := (Q_1 \cup Q_2 \cup \{s\}, s, D_1 \cup D_2 \cup D', E_1 \cup E_2 \cup E')$$ where
		$D' := \{s\} \times (D_1(s_1) \cup D_2(s_2))$ and $E' := \{s\} \times (E_1(s_1) \cup E_2(s_2))$.
	\item \textbf{Substitution.} Let $\vec{C} = (C_1, \dots, C_n)$ and $C$ be disjoint charts, and let $\vec{v} = (v_1, \dots, v_n)$ be distinct variables. We define
		$$
			C[\vec{C}/\vec{v}]=\left(Q \cup \bigcup_i Q_i, \{s\}, D' \cup \bigcup_i D_i, E' \cup \bigcup_i E_i\right)
		$$
		where for $q \in Q_i$, $D'(q)=E'(q)=\emptyset$, while for $q \in Q$, 
		$D'(q)=D(q) \cup \bigcup_i \left\{ D_i(s_i) \mid v_i \in E(q)\right\}$
		and 
		$E'(q)=(E(q) \setminus \vec{v}) \cup \bigcup_i \left\{ E_i(s_i) \mid v_i \in E(q)\right\}$.
	\item \textbf{Recursion.} If $C=(Q,s,D,E)$, then we define $\mu v.C := (Q,s,D^+,E^+)$, where
	\begin{gather*} 
		D^+(q) = \begin{cases} D(q) \cup D(s) & \text{if } v \in E(q) \\ D(q) & \text{otherwise} \end{cases}\qquad E^+(q) = \begin{cases}\left(E(q) \cup E(s)\right) \setminus \{v\} & \text{if } v \in E(q) \\ E(q) & \text{otherwise} \end{cases}
	\end{gather*}
\end{itemize}
When applied to finite charts, the operations above preserve finiteness. Given a chart $C=(Q,s,D,E)$ we say that a variable $v \in V$ is \emph{live} in $C$ if there exists a path of transitions $s \tr{a_1} \dots \tr{a_n} s' \rhd v$ or call it \emph{dead} otherwise.
The canonical notion of (pre)charts having the same observational behaviour is given through the following concept.
\begin{definition}[Strong Bisimulation]
	Let $C_i = (Q_i,D_i,E_i)$, $i \in \{1,2\}$ be precharts. A bisimulation between $C_1$ and $C_2$ is a relation ${R} \subseteq Q_1 \times Q_2$, such that \circlednum{1} if $(q_1,q_2)\in R$, then $E(q_1)=E(q_2)$, \circlednum{2} if $(q_1,q_2) \in R$ and $q_1 \tr{a} q'_1$, then there exists $q'_2 \in Q_2$, such that $q_2 \tr{a} q'_2$ and $(q'_1, q'_2) \in R$ and symmetrically. If $C_1$ and $C_2$ are charts, we say that they are bisimilar (denoted $C_1 \sim C_2$) if there exists a bisimulation between their underlying precharts that relates their start nodes.
\end{definition}
It immediately follows that $\sim$ is an equivalence relation on the set of finite charts. We call this relation \emph{bisimilarity} and we will write $\Omega$ for the set of all finite charts modulo bisimilarity. We will refer to the elements of this set as \emph{regular behaviours}. Conveniently, $\sim$ is a congruence with respect to all operations described above and hence one can unambiguously extend all the mentioned operations to elements of $\Omega$~\cite[Proposition~3.2]{Milner:1984:Complete}. Let $[C_1], [C_2] \in \Omega$, such that $C_1 = (Q,s,D,E)$ and $C_2=(Q,s',D,E)$. One can make $\Omega$ into prechart itself by setting $[C_1] \tr{a} [C_2] \iff s \tr{a}_{C_1} s'$ and $[C_1] \rhd v \iff s \rhd_{C_1} v$.
\begin{definition}[Prechart homomorphism]
	 Let $C_i = (Q_i,D_i,E_i)$, $i \in \{1,2\}$ be precharts. We call a function $f \colon Q_1 \to Q_2$ a prechart homomorphism if the graph of $f$, given by $G(f) = \{(q,f(q)) \mid q \in Q_1\}$ is a bisimulation between $C_1$ and $C_2$.
\end{definition}
Fix a prechart $C=(Q,D,E)$. A map taking each state $s \in Q$ to the equivalence class $[(Q,s,D,E)]$ in $\Omega$ is a homomorphism from $C$ into prechart on $\Omega$. It can be easily observed that bisimulations and homomorphisms preserve the liveness of variables. From now on, we will abuse the notation and omit the quotient brackets when talking about elements of $\Omega$.

\noindent
\textbf{Pseudometric spaces. } We will formally give the notion of distance between regular behaviours by equipping them with a \emph{pseudometric} structure, a mild generalisation of metric spaces, where we drop the requirement of points in zero distance having to be strictly equal. This stems from the fact that in precharts we might have two states with identical behaviour (and hence zero distance), which are not strictly equal (but rather bisimilar). 

\begin{definition}
	A 1-bounded pseudometric space is a pair $(X, d_X)$, where $X$ is a set and $d_X \colon X \times X \to [0,1]$ is a function satisfying \circlednum{1} $d_X(x,x)=0$ (Reflexivity), \circlednum{2} $d_X(x,y) = d_X(y,x)$ (Symmetry) and \circlednum{3} $d_X(x,z) \leq d_X(x,y) + d_X(y,z)$ (Transitivity) for all $x,y,z \in X$. 
\end{definition}
In this paper, all pseudometric spaces are 1-bounded, and hence we will abuse the terminology and simply call them pseudometric spaces. We call a function $f \colon X \to Y$ between pseudometric spaces $(X,d_X)$ and $(Y,d_Y)$ nonexpansive, if $d_Y(f(x),f(y))\leq d_X(x,y)$ for all $x,y \in X$. It is called an isometry if it satisfies $d_Y(f(x),f(y))=d_X(x,y)$. Given two pseudometric spaces $(X,d_X)$ and $(Y, d_Y)$ one can define their product to be $(X, d_X) \times (Y, d_Y) = (X \times Y, d_{X \times Y})$, where $d_{X \times Y}((x,y),(x',y')) = \max \{d(x,x'), d(y,y')\}$ for all $x,x' \in X$ and $y,y' \in Y$. This can be easily extended to any $n$-tuple. We define $0$-tuples to be given by $1_{\bullet} = (\{\bullet\}, d_{\bullet})$, the unique single point pseudometric space, where $d_{\bullet}(\bullet,\bullet)=0$. Given a function of multiple arguments, i.e. $X_1 \to X_2 \to Y$, we will call it nonexpansive, if it is nonexpansive as a function $f \colon (X_1, d_{X_1}) \times (X_2, d_{X_2}) \to (Y, d_Y)$. Given a set $X$, we write $D_X$ for the set of all pseudometrics on the set $X$. This set carries a partial order structure, given by $d \sqsubseteq d' \iff \forall_{x,x' \in X}~d(x,x') \leq d'(x,x')$. For any $X$, $(D_X, \sqsubseteq)$ is a complete lattice~\cite[Lemma~3.2]{Baldan:2018:Coalgebraic}, where supremas can be calculated pointwise. The top element of that lattice is given by the discrete pseudometric $\top \colon X \times X \to [0,1]$ such that $\top(x,y) = 0$ if $x=y$, or $\top(x,y)=1$ otherwise.

\noindent
\textbf{Behavioural distances. } We now have all the ingredients to formalise the concept of how much apart two regular behaviours are. Notation wise, we will write $\Pfin (X)$ for the set of \emph{finite} subsets of the set $X$. Given a prechart $(Q, E, D)$, we can equivalently see it as a pair $(Q, \beta)$, where $\beta$ is a combined transition function $Q \to \Pfin(\Sigma \times Q + V)$ taking each state $q \in Q$, to the set $\beta(q) = D(q) \cup E(q)$ of possible successors, that include labelled transitions and variable outputs. Given a pseudometric space defined on a state-space of a prechart, we can \emph{lift} it to the set of possible transitions through the following construction.
\begin{definition}[Transitions lifting]\label{def:edge}
	Let $(X,d)$ be a pseudometric space. We write $d^\uparrow$ for the pseudometric on $\Sigma \times X + V$ defined by $d^\uparrow(m, n) = \frac{1}{2} d(x,y)$ if $m = (a,x)$ and $n = (a,y)$, $d^\uparrow(m,n)=0$ if $m = n$ or $d^{\uparrow}(m,n)=1$ otherwise.
\end{definition}
Similarly, we can lift distances over $X$ to distances between elements of $\Pfin (X)$.
\begin{definition}[Hausdorff lifting]\label{def:hausdorff}
	Let $(X,d)$ be a pseudometric space. We can equip $\Pfin (X)$ with a distance function $
	\mathcal{H}(d)(X,Y)= \max \{\sup_{x \in X} \inf_{y \in Y} d(x,y), \sup_{y \in Y} \inf_{x \in X} d(y,x) \}$ making $(\Pfin(X), \mathcal{H}(d))$ into a pseudometric.
\end{definition}
Given a prechart $(Q, \beta)$, whose state-space is equipped with a pseudometric $d_Q$, we can define a new pseudometric $\Phi_{\beta}(d_Q)$ that calculates the distance between any pair $q_1, q_2 \in Q$ of states, by lifting $d_Q$ to the set $\Pfin(\Sigma \times Q + V)$ and comparing $\beta(q_1)$ with $\beta(q_2)$, namely 
	$
		\Phi_\beta(d_Q)(q_1,q_2) = \mathcal{H}\left(d_Q^\uparrow\right) (\beta(q_1), \beta(q_2))
	$. This is used to define the \emph{behavioural distance}.
	\begin{restatable}{theorem}{behdist}\label{thm:beh_dist}
		Let $(Q, \beta)$ be a prechart. Then, the following properties hold: \circlednum{1} $d_Q \mapsto \Phi_\beta(d_Q)$ is a monotone mapping on the lattice $D_Q$, \circlednum{2} $\Phi_\beta$ has a least fixpoint $\mathsf{bd}_\beta$, \circlednum{3} $x \sim y \implies \mathsf{bd}_\beta(x,y) = 0$ and \circlednum{4} a homomorphism $f \colon Q \to R$ between precharts $(Q, \beta)$ and $(R, \gamma)$ is an isometry between $(Q, \mathsf{bd}_\beta)$ and $(R, \mathsf{bd}_\gamma)$.
	\end{restatable}
	\noindent
	The theorem above allows one to define a distance between regular behaviours, by calculating the distance in the prechart structure on $\Omega$. We will simply refer to this distance as the \emph{behavioural distance} and denote it by $\mathsf{bd}$. Since prechart homomorphisms are isometries, given two finite charts $C_1$ and $C_2$, their distance is given by $\mathsf{bd}([C_1],[C_2])$, the distance between their corresponding regular behaviours. Such distance can be characterised concretely via Hennessy and Milner's \emph{stratification of bisimilarity}~\cite{hennessy:1985:algebraic}. 
	\begin{definition}[Stratification of bisimilarity]
	Let $C_i = (Q_i, D_i, E_i)$ for $i \in \{1,2\}$ be precharts. We can define a family $\{\sim^{(n)} \}_{n \in \N}$ of equivalence relations on $Q_1 \times Q_2$ given by the following. For all $(q_1,q_2) \in Q_1 \times Q_2$, we have that $q_1 \sim^{(0)} q_2$. Given $(q_1, q_2) \in Q_1 \times Q_2$, we have that $q_1 \sim ^{(n+1)} q_2$ if \circlednum{1} $E_1(q_1) = E_2(q_2)$, \circlednum{2} $q_1 \tr{a}_{C_1} q'_1$ implies that there exists $q'_2 \in Q_2$, such that $q_2 \tr{a} q'_2$ and $q'_1 \sim^{(n)} q'_2$ and symmetrically.
\end{definition}
We can now formally state the correspondence between behavioural distance and stratification of bisimilarity.
\begin{restatable}{theorem}{concrete}
	Let $C_i = (Q_i, s_i, D_i, E_i)$ for $i = \{1,2\}$ be finite charts. We have that $\mathsf{bd}([C_1],[C_2]) = 0$ if $s_1 \sim s_2$. Otherwise, $\mathsf{bd}([C_1],[C_2]) = 2^{-n}$, where $n \in \N$ is the largest such that $s_1\sim^{(n)} s_2$.
\end{restatable}
\section{Monoidal Syntax}
\label{sec:syntax}

We adopt the diagrammatic syntax for NFA that has appeared in a number of previous papers~\cite{piedeleu2023finite,antoinecsl2025}. 
We refer the reader to Selinger's classic survey~\cite{Selinger_2010}, or to Piedeleu and Zanasi's recent text for a more gentle introduction to the language of string diagrams~\cite{piedeleu2023introduction}.

This syntax is formalised as a product and permutation category, or prop, a structure which generalises algebraic theories. Formally, a \emph{prop} is a strict symmetric monoidal category (SMC) whose objects are words over a set of generators and whose monoidal product $\proptimes$ is given by concatenation. 
More specifically, our syntax is the free prop $\freeP{\Signature}$ over the signature $\Signature = (\Obj,\Morph)$, given by a set $\Obj$ of generating objects and a set $\Morph$  of generating morphisms $g\from v\to w$, with $v,w\in \Obj^*$ (we use $\epsilon$ to denote the empty word). Morphisms of  $\freeP{\Signature}$ can be combined in two different ways, using the composition operation $(-);(-)\from \freeP{\Signature}(u,v)\times \freeP{\Signature}(v,w)\to \freeP{\Signature}(u,w)$ or the monoidal product $(-)\proptimes(-)\from \freeP{\Signature}(v_1,w_1)\times \freeP{\Signature}(v_2,w_2)\to \freeP{\Signature}(v_1 v_2,w_1 w_2)$. We also have distinguished constants: identities $\id_w\from w\to w$, which are the unit for composition, and symmetries $\sigma^v_w\from vw\to wv$, to reorder the letters of a given object. In summary, morphisms of $\freeP{\Signature}$ can be described as terms of the $(\Obj^*,\Obj^*)$-sorted syntax generated from the constants $\Morph + \{\id_w : w\in \Obj^*\}+ \{\sigma^v_w : v,w\in \Obj^*\}$ using the operations $;$ and $\proptimes$, \emph{quotiented} by the axioms of SMCs. However, the terms of this syntax are very cumbersome to work with. 

We adopt a more convenient way to represent morphisms of $\freeP{\Signature}$, using the graphical notation of \emph{string diagrams}. In this view, a morphism $f\from v\to w$ of $\freeP{\Signature}$ is depicted as a $f$-labelled box with a $v$-labelled wire on the left and a $w$-labelled wire on the right. The operations of composition and monoidal product are represented by connecting two boxes horizontally and juxtaposing two boxes vertically, respectively:
\begin{equation*}\label{eq:composition-monoidal-product}
\scalebox{0.9}{
\tikzset{x=1em, y=2.1ex}
\InputIfFileExists{comp-sequential-fg.tikz}{}{\input{./tikz/comp-sequential-fg.tikz}}
\tikzset{x=1em, y=1.5ex}
\qquad \quad 
\tikzset{x=1em, y=2.1ex}
\InputIfFileExists{comp-parallel-fg.tikz}{}{\input{./tikz/comp-parallel-fg.tikz}}
\tikzset{x=1em, y=1.5ex}
}
\end{equation*}
Wires $
\tikzset{x=1em, y=2.1ex}
\begin{tikzpicture}
	\begin{pgfonlayer}{nodelayer}
		\node [style=none] (0) at (-1.5, 0) {};
		\node [style=none] (1) at (0.5, 0) {};
		\node [style=none] (3) at (-0.5, 0.5) {\scriptsize  $w$};
	\end{pgfonlayer}
	\begin{pgfonlayer}{edgelayer}
		\draw (0.center) to (1.center);
	\end{pgfonlayer}
\end{tikzpicture}
}
\tikzset{x=1em, y=1.5ex}
$ represent identities, the wire crossing $
\tikzset{x=1em, y=2.1ex}
\InputIfFileExists{sym-vxw.tikz}{}{\input{./tikz/sym-vxw.tikz}}
\tikzset{x=1em, y=1.5ex}
$ represents the symmetry $\sigma^v_w$, and the empty diagram $\idzero$ the identity $\id_\epsilon\from \epsilon\to \epsilon$.
\begin{definition}\label{def:syntax}
	 We call $\Syn$ the free prop over the signature given by
	\begin{itemize}
		\item two generating objects $\objl$ ("left") and $\objr$ ("right"), with their identity morphisms depicted respectively as $\arrowleft$ and $\arrowright$;
 		\item generating morphisms $
 		
\tikzset{x=1em, y=2.1ex}
\InputIfFileExists{lr-copy.tikz}{}{\input{./tikz/lr-copy.tikz}}
\tikzset{x=1em, y=1.5ex}
\quad
\tikzset{x=1em, y=2.1ex}
\begin{tikzpicture}
	\begin{pgfonlayer}{nodelayer}
		\node [style=black] (37) at (0.75, 0) {};
		\node [style=none] (43) at (0.25, 0) {};
		\node [style=none] (44) at (-0.5, 0) {};
	\end{pgfonlayer}
	\begin{pgfonlayer}{edgelayer}
		\draw (43.center) to (37);
		\draw [->] (44.center) to (43.center);
	\end{pgfonlayer}
\end{tikzpicture}
}
\tikzset{x=1em, y=1.5ex}
\quad
\tikzset{x=1em, y=2.1ex}
\InputIfFileExists{lr-merge.tikz}{}{\input{./tikz/lr-merge.tikz}}
\tikzset{x=1em, y=1.5ex}
\quad
\tikzset{x=1em, y=2.1ex}
\begin{tikzpicture}
	\begin{pgfonlayer}{nodelayer}
		\node [style=black] (37) at (-0.5, 0) {};
		\node [style=none] (43) at (0.25, 0) {};
		\node [style=none] (44) at (0.75, 0) {};
	\end{pgfonlayer}
	\begin{pgfonlayer}{edgelayer}
		\draw [->] (37) to (43.center);
		\draw (44.center) to (43.center);
	\end{pgfonlayer}
\end{tikzpicture}
}
\tikzset{x=1em, y=1.5ex}
 \quad
\tikzset{x=1em, y=2.1ex}
\InputIfFileExists{cap-down.tikz}{}{\input{./tikz/cap-down.tikz}}
\tikzset{x=1em, y=1.5ex}
 \quad
\tikzset{x=1em, y=2.1ex}
\InputIfFileExists{cup-down.tikz}{}{\input{./tikz/cup-down.tikz}}
\tikzset{x=1em, y=1.5ex}
\quad \scalar{a} \quad (a\in \Sigma)$.
 	\end{itemize}
 \end{definition}
Morphisms of $\Syn$ are thus vertical and horizontal composites of the generators above, potentially including wire crossings and identity wires, \emph{up to} the laws of symmetric monoidal categories---see Section~\ref{app:smc-axioms}. The direction of the arrows on the wires denotes the corresponding object: for example, $
\tikzset{x=1em, y=2.1ex}
\InputIfFileExists{lr-copy.tikz}{}{\input{./tikz/lr-copy.tikz}}
\tikzset{x=1em, y=1.5ex}
$ represents an operation of type $\objr\to \objr\objr$, while $
\tikzset{x=1em, y=2.1ex}
\InputIfFileExists{cap-down.tikz}{}{\input{./tikz/cap-down.tikz}}
\tikzset{x=1em, y=1.5ex}
$ has type $\objl\objr\to \epsilon$. Note that, when we have $n$ parallel wires of the same type, say $\objr$, we depict them as a single directed wire labelled by a natural number label, as $\idright^{\!\!\!\!\!\! n}$. We call \emph{inputs} the incoming wires of a diagram, and \emph{outputs} its outgoing wires; formally, the inputs (resp. outputs) of $f\from v\to w$ are the set of positions of the word $v$ which are $\objr$ (resp. $\objl$) and the position of $w$ which are $\objl$ (resp. $\objr$). 
	
\section{Monoidal Semantics}\label{sec:semantics}
In order to interpret the string diagrams described in \Cref{sec:syntax}, we construct an appropriate semantic universe out of regular behaviours. As much as the technical development makes use of category theory, we will keep the description of the formalism high-level. The fine-grained details are relegated to \Cref{apx:semantics}. 
We will write $V_n$ for the set $V_n = \{v_1, \dots, v_n\}\subseteq V$ and $\Omega(n)$ for the set of all regular behaviours whose live variables are contained in the set $V_n$. For any $m,n \in \N$, we will write $\RegBeh(n,m)$ for the set of $n$-tuples of elements of $\Omega(m)$. There are several operations involving these sets that are of interest:
\begin{itemize}[topsep=0pt]
	\item \textbf{Identity.} For every $n \in \N$, we define $\id_n \in \RegBeh(n,n)$ as  $\id_n=\vec{v}_n=(v_1, \dots, v_n)$. When $n$ is clear from the context, we will abuse the notation and simply write $\vec{v}$ instead.
	 \item \textbf{Sequential composition.} Given $f \in \RegBeh(m,n)$ and $g \in \RegBeh(n,p)$, we can define their sequential composition $f;g \in \RegBeh(m,p)$ to be given by $(f_1[\vec{g} / \vec{v}], \dots, f_m[\vec{g}/ \vec{v}])$, where $\vec{v}=(v_1, \dots, v_n)$. The sequential composition is associative, with identity being a neutral element when composed both on the left and right (see \Cref{lem:comp_associative} in the appendix of the full version of the paper).
	\item \textbf{Initial object.} For every $n \in \N$, there is a unique element $0_n \in \RegBeh(0,n)$ given by the empty tuple.
	\item \textbf{Pairing.} Given $f \in \RegBeh(k,m)$ and $g \in \RegBeh(l,m)$, we define their pairing $\langle f , g \rangle \in \RegBeh (k + l,m)$, by setting $\langle f , g \rangle = (f_1, \dots, f_k, g_1, \dots g_l)$.
	\item \textbf{Parallel composition.} Given $f \in \RegBeh(k,l)$ and $g \in \RegBeh(m,n)$, we can define their parallel composition $f \oplus g \in \RegBeh(k + m, l + n)$, by setting $$f \oplus g = (f_1, \dots, f_k, g_1[(v_{l + 1},\dots, v_{l + n})/\vec{v}], \dots, g_m[(v_{l + 1},\dots, v_{l + n})/\vec{v}])$$
	\item \textbf{Codiagonal.} For any $n \in \N$, there is a $n+n$-tuple $\nabla_n \in \RegBeh(n + n, n)$ called codiagonal given by $\nabla_n = \langle\id_n, \id_n\rangle$.
	\item \textbf{Dagger.} For any $f \in \RegBeh(1, p + 1)$, we can define $f^\dagger \in \RegBeh(1,p)$ to be given by $f^\dagger = \mu v_{p+1}.f$. Following~\cite[Remark~3.2]{Esik:1999:Group}, we can inductively extend the dagger to the map taking each $f \in \RegBeh(n, p + n)$ to $f^\dagger \in \RegBeh(n, p)$. 
	\end{itemize}
\textbf{Categorical structure. } 
We now define a category $\RegBeh$, whose objects are natural numbers and morphisms $f \colon m \to n$ are elements $f \in \RegBeh(m,n)$. This category has a rich structure that will be useful when defining the semantics of our diagrammatic language.
\begin{theorem}
	The category $\RegBeh$ has the following properties:
	\begin{itemize}[topsep=0pt]
		\item $\RegBeh$ has all finite coproducts. 
		\item $(\RegBeh, \oplus, 0)$ is a (co-Cartesian) strict symmetric monoidal category.
		\item $\RegBeh$ equipped with a dagger is a Conway theory~\cite{Esik:1999:Group}.
		\item Each morphism $g \colon p + n \to q + n$ has a trace $\Tr^{n}_{p,q}(g) \colon p \to q$ defined in terms of $(-)^\dagger$. This equipment makes $\RegBeh$ into a traced monoidal category~\cite{Joyal:1996:Traced}. 
	\end{itemize}
\end{theorem}
\noindent
\textbf{Pseudometric structure.} $\RegBeh$ additionally carries a well-behaved pseudometric structure. For all $m,n \in \N$ each set $\RegBeh(m,n)$ can be made into a pseudometric space by equipping it with a distance function, given by $d^{m,n}((f_1, \dots, f_m), (g_1, \dots, g_m)) = \sup_{1 \leq i \leq m} \left\{ \mathsf{bd}(f_i, g_i) \right\}$.
Intuitively, we calculate the distance between $m$-tuples of regular behaviours, by taking the pointwise behavioural distance of elements of tuples and then taking the maximum. In the corner case, when both tuples are empty, then they are simply at distance zero. This equipment satisfies the following property.
\begin{restatable}{proposition}{propnexp}
	Equipping each set $\RegBeh(m,n)$ of morphisms of $\RegBeh$ with a pseudometric defined above makes a sequential composition, pairing, parallel composition, dagger and trace into nonexpansive maps. 
\end{restatable}
\begin{remark}
    Although we do not pursue such a perspective in this paper, a categorically inclined reader may observe that equipping $\RegBeh$ with a pseudometric structure yields an enrichment over $\mathsf{PMet}$---the monoidal category of pseudometric spaces with the categorical product as its monoidal product. Since $\mathsf{PMet}$ is a concrete category (i.e., it admits a faithful functor to $\mathsf{Set}$), the conditions for enrichment simplify significantly: they reduce to verifying that all the relevant operations (composition, monoidal product, etc.) are nonexpansive on each homset and that the components of the natural transformations that define the monoidal structure are nonexpansive maps, which is exactly what we have shown above. For further details on pseudometric enrichment of monoidal categories and quantitative equational theories, we refer the reader to the recent work of Lobbia et al.~\cite{Lobbia:2024:Quantitative}.
\end{remark}
\noindent
\textbf{Bidirectional maps and loops. } Each morphism $f \colon m \to n$ of $\RegBeh$ can be informally thought of as a process that has \emph{directionality} to it, i.e. it takes $m$ inputs and produces $n$ outputs. Informally speaking, the trace operator provides a \emph{global} feedback operation. At the same time, the syntax of our diagrammatic language is \emph{bidirectional} and the notion of feedback is introduced \emph{locally} by bending the wires. To reconcile these points of view, we rely on the \textbf{Int} construction~\cite{Joyal:1996:Traced}, which takes a traced monoidal category $\cat{C}$ and completes it into a compact closed category $\Int{\cat{C}}$, a categorical structure with sequential and parallel composition equipped with duals (allowing to swap directionality) and adjoints (allowing to form local loops representing feedback)~\cite{kellylaplaza}. $\Int{\cat{C}}$ carries the same information as $\cat{C}$, but represents it in an alternative, bidirectional way. 
We now briefly describe $\Int{\RegBeh}$ (see \Cref{sec:apx_int} for more detail). 
	\begin{itemize}[topsep=0pt]
		\item The objects of $\Int{\RegBeh}$ are pairs $(m,n)$ of natural numbers. 
		\item A morphism $f \colon (k,l) \to (m,n)$, representing a process with $k$ left inputs, $l$ left outputs, $m$ right outputs and $n$ right inputs is a map $f \colon k + m \to l + n$ in $\RegBeh$, i.e. we group inputs and outputs together. Composition of $f \colon (k,l) \to (m,n)$ and $g \colon (m,n) \to (p,q)$ is defined by forming a trace that resolves the feedback involving $m$ and $n$.
		\item The parallel composition of $f \colon (m,n) \to (p,q)$ and $g \colon (m',n') \to (p',q')$ is given by the map $f \otimes g \colon (m + m', n + n') \to (p + p', q + q')$ that is defined via parallel composition in $\RegBeh$ combined with an appropriate reordering of elements of tuples involved.
		\item A dual of the object $(m,n)$ of $\Int{\RegBeh}$ is given by $(n,m)$. Intuitively, inputs become swapped with outputs. For each object $(m,n)$ of $\Int{\RegBeh}$, there is a unit map $\eta_{(m,n)} \colon (0,0) \to (m+n, m+n)$ and counit $\epsilon_{(m,n)} \colon (m+n, m+n) \to (0,0)$. These represent the bending of the wires on the right and left respectively. Following~\cite{Joyal:1996:Traced}, one can equip $\Int{\RegBeh}$ with a trace operator defined in terms of the units and counits.
		\item $\Int{\RegBeh}$ inherits the pseudometric structure from $\RegBeh$, by setting $d^{(k,l),(m,n)}$ to be given by $d^{k + n, l +n}$ (defined before).
	\end{itemize}
	The following theorem intuitively states that on \emph{directional} processes $\Int{\RegBeh}$ is exactly the same as $\RegBeh$. 
	\begin{theorem}[{\cite[Proposition~5.1]{Joyal:1996:Traced}}]\label{thm:trace_embeds_int}
		There is a full and faithful traced monoidal functor $N \colon \RegBeh \to \Int{\RegBeh}$ that takes each $f \colon n \to m$ in $\RegBeh$ to $f \colon (n,0) \to (m,0)$
	\end{theorem}
	Moreover, the pseudometric structure interacts well with the operations of $\Int{\RegBeh}$.
	\begin{restatable}{proposition}{propnexpint}\label{cor:sem_enriched}
		The sequential and parallel composition in $\Int{\RegBeh}$ is nonexpansive. Moreover, the fully faithful functor $N \colon \RegBeh \to \Int{\RegBeh}$ is locally an isometry, i.e. for all $f,g \colon m \to n$, we have that $d^{m,n}(f,g)=d^{(m,0), (n,0)}(N(f),N(g))$.
	\end{restatable}
	\noindent
	\textbf{Functorial semantics.} We are ready to state the semantics of our diagrammatic language $\sem{-} \colon \Syn \to \Int{\RegBeh}$ as a symmetric monoidal functor from $\Syn$ to $\Int{\RegBeh}$. Since the syntax is a freely generated prop, in order to interpret arbitrary string diagrams it is enough to just define the interpretation of the generating morphisms of $\Syn$. We have:
	\begin{gather*}
	\sem{
\tikzset{x=1em, y=2.1ex}
\InputIfFileExists{lr-copy.tikz}{}{\input{./tikz/lr-copy.tikz}}
\tikzset{x=1em, y=1.5ex}
}=  N(v_1 + v_2)  \qquad 
	\sem{
\tikzset{x=1em, y=2.1ex}
\begin{tikzpicture}
	\begin{pgfonlayer}{nodelayer}
		\node [style=black] (37) at (0.75, 0) {};
		\node [style=none] (43) at (0.25, 0) {};
		\node [style=none] (44) at (-0.5, 0) {};
	\end{pgfonlayer}
	\begin{pgfonlayer}{edgelayer}
		\draw (43.center) to (37);
		\draw [->] (44.center) to (43.center);
	\end{pgfonlayer}
\end{tikzpicture}
}
\tikzset{x=1em, y=1.5ex}
} =  N(0)  \qquad
	\sem{
\tikzset{x=1em, y=2.1ex}
\begin{tikzpicture}
	\begin{pgfonlayer}{nodelayer}
		\node [style=black] (37) at (-0.5, 0) {};
		\node [style=none] (43) at (0.25, 0) {};
		\node [style=none] (44) at (0.75, 0) {};
	\end{pgfonlayer}
	\begin{pgfonlayer}{edgelayer}
		\draw [->] (37) to (43.center);
		\draw (44.center) to (43.center);
	\end{pgfonlayer}
\end{tikzpicture}
}
\tikzset{x=1em, y=1.5ex}
} =  N(())  \qquad
	\sem{
\tikzset{x=1em, y=2.1ex}
\InputIfFileExists{lr-merge.tikz}{}{\input{./tikz/lr-merge.tikz}}
\tikzset{x=1em, y=1.5ex}
} = N(\nabla_1) \\
	\sem{\scalar{a}} =  N(a.v_1)  \qquad 
	\sem{
\tikzset{x=1em, y=2.1ex}
\InputIfFileExists{cap-down.tikz}{}{\input{./tikz/cap-down.tikz}}
\tikzset{x=1em, y=1.5ex}
} = \epsilon_{(1, 0)} \qquad
	\sem{
\tikzset{x=1em, y=2.1ex}
\InputIfFileExists{cup-down.tikz}{}{\input{./tikz/cup-down.tikz}}
\tikzset{x=1em, y=1.5ex}
} = \eta_{(1, 0)} \qquad
\end{gather*}
We interpret $
\tikzset{x=1em, y=2.1ex}
\InputIfFileExists{lr-copy.tikz}{}{\input{./tikz/lr-copy.tikz}}
\tikzset{x=1em, y=1.5ex}
$ as nondeterministic choice, $
\tikzset{x=1em, y=2.1ex}
\begin{tikzpicture}
	\begin{pgfonlayer}{nodelayer}
		\node [style=black] (37) at (0.75, 0) {};
		\node [style=none] (43) at (0.25, 0) {};
		\node [style=none] (44) at (-0.5, 0) {};
	\end{pgfonlayer}
	\begin{pgfonlayer}{edgelayer}
		\draw (43.center) to (37);
		\draw [->] (44.center) to (43.center);
	\end{pgfonlayer}
\end{tikzpicture}
}
\tikzset{x=1em, y=1.5ex}
$ as the behaviour of the empty chart, while $
\tikzset{x=1em, y=2.1ex}
\begin{tikzpicture}
	\begin{pgfonlayer}{nodelayer}
		\node [style=black] (37) at (-0.5, 0) {};
		\node [style=none] (43) at (0.25, 0) {};
		\node [style=none] (44) at (0.75, 0) {};
	\end{pgfonlayer}
	\begin{pgfonlayer}{edgelayer}
		\draw [->] (37) to (43.center);
		\draw (44.center) to (43.center);
	\end{pgfonlayer}
\end{tikzpicture}
}
\tikzset{x=1em, y=1.5ex}
$ corresponds to the empty tuple. The generator $
\tikzset{x=1em, y=2.1ex}
\InputIfFileExists{lr-merge.tikz}{}{\input{./tikz/lr-merge.tikz}}
\tikzset{x=1em, y=1.5ex}
$ implements variable substitution: given a pair of processes $(t_1, t_2)$ with outputs $v_1$ and $v_2$ respectively and composing it with $
\tikzset{x=1em, y=2.1ex}
\InputIfFileExists{lr-merge.tikz}{}{\input{./tikz/lr-merge.tikz}}
\tikzset{x=1em, y=1.5ex}
$ on the right \emph{merges} the two output variables, replacing $v_2$ with $v_1$ to yield $(t_1, t_2[v_1/v_2])$. For each letter $a \in \Sigma$ of the alphabet, we view $\scalar{a}$ as the prefixing operation. Finally, $
\tikzset{x=1em, y=2.1ex}
\InputIfFileExists{cap-down.tikz}{}{\input{./tikz/cap-down.tikz}}
\tikzset{x=1em, y=1.5ex}
$ and $
\tikzset{x=1em, y=2.1ex}
\InputIfFileExists{cup-down.tikz}{}{\input{./tikz/cup-down.tikz}}
\tikzset{x=1em, y=1.5ex}
$ are interpreted using counit and unit of $\Int{\RegBeh}$, allowing one to create loops. Such loops correspond to expressions in Milner's syntax featuring the recursion operator $\mu$.

\section{Axiomatisation}
\label{sec:axioms}

Our main aim in this paper is to find a set of (quantitative) equations to reason about semantic distance directly at the level of the diagrams themselves. To do so, we distinguish two different relations on diagrams of the same type:
\begin{itemize}[topsep=0pt]
\item An equational theory intended to capture (strong) bisimilarity of regular behaviours, allowing us to simplify the diagrams whose distance is being compared.
\item A quantitative equational theory intended to capture the behavioural distance of ~\Cref{sec:semantics}, that is the subject of the completeness theorem (\Cref{thm:completeness}) described in \Cref{sec:completeness}. Note that this theory contains the equational axioms as rules for distance zero.
\end{itemize}
\noindent
\textbf{Equational theory. } Our equational theory is the smallest congruence (w.r.t to vertical and horizontal compositions) that includes the axioms of Fig.~\ref{fig:equational-axioms}. In practice, this means that, if we find a sub-diagram that matches one side of an axiom in a larger diagram, we can replace it with the other side of the axiom (the left and right-hand side of any axiom have the same type)~\cite[Section 2.1]{piedeleu2023introduction}. 
\begin{figure}[h!]
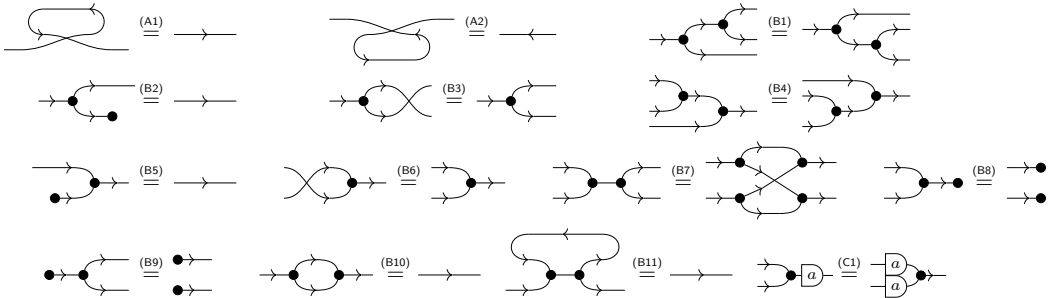

{
\scalebox{0.85}{ 
\[
\begin{aligned}

\tikzset{x=1em, y=2.1ex}
\InputIfFileExists{yanking-bis.tikz}{}{\input{./tikz/yanking-bis.tikz}}
\tikzset{x=1em, y=1.5ex}
\;&\myeq{A1}\;\arrowright \qquad\qquad
\tikzset{x=1em, y=2.1ex}
\InputIfFileExists{yanking.tikz}{}{\input{./tikz/yanking.tikz}}
\tikzset{x=1em, y=1.5ex}
\;\myeq{A2}\;\arrowleft\qquad \qquad
 
\tikzset{x=1em, y=2.1ex}
\InputIfFileExists{co-associativity.tikz}{}{\input{./tikz/co-associativity.tikz}}
\tikzset{x=1em, y=1.5ex}
 \; \myeq{B1} \; 
\tikzset{x=1em, y=2.1ex}
\InputIfFileExists{co-associativity-1.tikz}{}{\input{./tikz/co-associativity-1.tikz}}
\tikzset{x=1em, y=1.5ex}
\\
  
\tikzset{x=1em, y=2.1ex}
\InputIfFileExists{right-co-unitality.tikz}{}{\input{./tikz/right-co-unitality.tikz}}
\tikzset{x=1em, y=1.5ex}
 &\myeq{B2} \; \arrowright\qquad\qquad 

\tikzset{x=1em, y=2.1ex}
\InputIfFileExists{co-commutativity.tikz}{}{\input{./tikz/co-commutativity.tikz}}
\tikzset{x=1em, y=1.5ex}
\; \myeq{B3}\; 
\tikzset{x=1em, y=2.1ex}
\InputIfFileExists{large-copy.tikz}{}{\input{./tikz/large-copy.tikz}}
\tikzset{x=1em, y=1.5ex}
\qquad\qquad 
\tikzset{x=1em, y=2.1ex}
\InputIfFileExists{associativity.tikz}{}{\input{./tikz/associativity.tikz}}
\tikzset{x=1em, y=1.5ex}
 \; \myeq{B4} \; 
\tikzset{x=1em, y=2.1ex}
\InputIfFileExists{associativity-1.tikz}{}{\input{./tikz/associativity-1.tikz}}
\tikzset{x=1em, y=1.5ex}

\\

\tikzset{x=1em, y=2.1ex}
\InputIfFileExists{right-unitality.tikz}{}{\input{./tikz/right-unitality.tikz}}
\tikzset{x=1em, y=1.5ex}
\; &\myeq{B5} \; \arrowright \qquad 
\tikzset{x=1em, y=2.1ex}
\InputIfFileExists{commutativity.tikz}{}{\input{./tikz/commutativity.tikz}}
\tikzset{x=1em, y=1.5ex}
\; \myeq{B6}\; 
\tikzset{x=1em, y=2.1ex}
\InputIfFileExists{large-merge.tikz}{}{\input{./tikz/large-merge.tikz}}
\tikzset{x=1em, y=1.5ex}
 \qquad

\tikzset{x=1em, y=2.1ex}
\InputIfFileExists{bimonoid.tikz}{}{\input{./tikz/bimonoid.tikz}}
\tikzset{x=1em, y=1.5ex}
\; \myeq{B7} \;
\tikzset{x=1em, y=2.1ex}
\InputIfFileExists{bimonoid-1.tikz}{}{\input{./tikz/bimonoid-1.tikz}}
\tikzset{x=1em, y=1.5ex}
\qquad

\tikzset{x=1em, y=2.1ex}
\InputIfFileExists{merge-delete.tikz}{}{\input{./tikz/merge-delete.tikz}}
\tikzset{x=1em, y=1.5ex}
\; \myeq{B8} \;
\tikzset{x=1em, y=2.1ex}
\InputIfFileExists{merge-delete-1.tikz}{}{\input{./tikz/merge-delete-1.tikz}}
\tikzset{x=1em, y=1.5ex}

\\

\tikzset{x=1em, y=2.1ex}
\InputIfFileExists{copy-co-delete.tikz}{}{\input{./tikz/copy-co-delete.tikz}}
\tikzset{x=1em, y=1.5ex}
\; &\myeq{B9} \;
\tikzset{x=1em, y=2.1ex}
\InputIfFileExists{copy-co-delete-1.tikz}{}{\input{./tikz/copy-co-delete-1.tikz}}
\tikzset{x=1em, y=1.5ex}

\qquad

\tikzset{x=1em, y=2.1ex}
\InputIfFileExists{idempotence.tikz}{}{\input{./tikz/idempotence.tikz}}
\tikzset{x=1em, y=1.5ex}
\; \myeq{B10} \;\arrowright \quad 
\tikzset{x=1em, y=2.1ex}
\InputIfFileExists{feedback.tikz}{}{\input{./tikz/feedback.tikz}}
\tikzset{x=1em, y=1.5ex}
\; \myeq{B11} \; \arrowright\quad {\small 
\tikzset{x=1em, y=2.1ex}
\InputIfFileExists{co-copy-atom.tikz}{}{\input{./tikz/co-copy-atom.tikz}}
\tikzset{x=1em, y=1.5ex}
} \myeq{C1}\: {\small 
\tikzset{x=1em, y=2.1ex}
\InputIfFileExists{co-copy-atom-1.tikz}{}{\input{./tikz/co-copy-atom-1.tikz}}
\tikzset{x=1em, y=1.5ex}
}
\end{aligned}
\]
}
}
\caption{Equational axioms for regular behaviours.}
\label{fig:equational-axioms}
\end{figure}

Axioms \textsf{A1}-\textsf{A2} are those of \emph{compact closed categories}~\cite{kellylaplaza} and allow us to bend and straighten wires at will, only keeping track of their directions. Crucially, they also allow us to manipulate feedback loops; in Milner's syntax, they are the structural rules that let one rearrange a system of recursive equations without changing its solution.

 \textsf{B1-B3} encode the fact that $\Bcomult$ and $\Bcounit$ form a \emph{cocommutative comonoid}. These guarantee that nondeterministic choice behaves like a commutative, associative and unital operation in our diagrammatic syntax. \textsf{B1} corresponds to associativity $t_1 + (t_2 + t_3) = (t_1 + t_2) + t_3$, \textsf{B2} to the right unit law $t + 0 = t$, and \textsf{B3} to commutativity $t_1 + t_2 = t_2 + t_1$. At the same time, \textsf{B4-B6} are the dual of the previous three, and make $\Bmult$ and $\Bunit$ into a \emph{commutative monoid}. 
 These guarantee that output wires behave like the variables of our diagrammatic syntax: read in Milner's notation, \textsf{B4} states that successive merges associate, so that merging $v_3$ into $v_2$ and then $v_2$ into $v_1$ yields the same result as merging $v_3$ into $v_2$ after $v_2$ has already been merged into $v_1$; \textsf{B5} states that renaming an empty tuple has no effect; and \textsf{B6} states that merging is symmetric in the two variables being identified. 
 
 \textsf{B7}-\textsf{B9} make the previous monoid-comonoid pair into a \emph{bimonoid} expressing how nondeterministic choice and variable-merging interact. Concretely, \textsf{B7} states that merging two variables and then forming a nondeterministic choice can be expressed as a pair of nondeterministic choices followed by pairwise merging. \textsf{B8} says that the merging operation followed by a deadlock is the same as a pair of deadlocked processes, and \textsf{B9} says that the empty tuple of processes is itself unchanged by variable merging. \textsf{B10} makes nondeterministic choice an idempotent operation $t + t = t$. 
 
 Axiom \textsf{B11} allows one to remove unguarded loops, capturing the standard equation $\mu v_1. v_1 = 0$. Finally, \textsf{C1} encodes the fact that merging tuples of variables interacts as expected with prefixing: $(a.v_1, a.v_1)=(a.v_1, a.v_2[v_1/v_2]) $. Note that, if we replace the $\Bmult$ with $\Bcomult$, the resulting equality (prefixes distribute over nondeterministic sum) is not valid, which is precisely the law that fails for bisimilarity, but holds for trace equivalence~\cite{piedeleu2023finite}.
\begin{restatable}[Soundness]{lemma}{soundnessequality}
For any two diagrams $f,g\from v\to w$ of $\Syn$, if $f=g$ then $\sem{f}=\sem{g}$. 
\end{restatable}
\noindent
\textbf{Quantitative theory. } We define $E_q$ to be the set of triples of the form $(s, \varepsilon, t)$, where $s,t$ are string diagrams of the same type and $\varepsilon \in \Qp$, to be the least set closed under the rules of Fig.~\ref{fig:quantitative-axioms}. We will call elements of that set \emph{derivable equations}. 
\begin{figure}[h!]
{
\mprset{vskip=1ex}
\scriptsize
\begin{equation*}
\inferrule{\dbox{f}{m}{n} \disteq{\varepsilon} \dbox{g}{m}{n}\qquad \vec{a} \in\Sigma^m}{
\tikzset{x=1em, y=2.1ex}
\InputIfFileExists{prefix-f.tikz}{}{\input{./tikz/prefix-f.tikz}}
\tikzset{x=1em, y=1.5ex}
 \disteq{\epsilon/2} 
\tikzset{x=1em, y=2.1ex}
\InputIfFileExists{prefix-g.tikz}{}{\input{./tikz/prefix-g.tikz}}
\tikzset{x=1em, y=1.5ex}
}{\mathsf{(Pref)}} 
\qquad 	\inferrule{\diagbox{f}{v}{w} \disteq{\varepsilon} \diagbox{g}{v}{w}\qquad \varepsilon' \geq \varepsilon}{\diagbox{g}{v}{w} \disteq{\varepsilon'} \diagbox{f}{v}{w}}{\mathsf{(Max)}}
\end{equation*}
\vspace{1em}
\begin{equation*}
\inferrule{\diagbox{f}{v}{w} = \diagbox{g}{v}{w}}{\diagbox{f}{v}{w}\disteq{0} \diagbox{g}{v}{w}}{\mathsf{(Refl)}}
\qquad
	\inferrule{\diagbox{f}{v}{w} \disteq{\varepsilon} \diagbox{g}{v}{w}}{\diagbox{g}{v}{w} \disteq{\varepsilon} \diagbox{f}{v}{w}}{\mathsf{(Sym)}}
\end{equation*}
\vspace{1em}
\begin{equation*}
	\inferrule{\left\{\diagbox{f}{v}{w} \disteq{\varepsilon'} \diagbox{g}{v}{w}\right\}_{\varepsilon' > \varepsilon}}{\diagbox{f}{v}{w} \disteq{\varepsilon} \diagbox{g}{v}{w}}{\mathsf{(Cont)}}
\qquad
	\inferrule{\diagbox{f}{v}{w} \disteq{\varepsilon_1} \diagbox{g}{v}{w}\qquad \diagbox{g}{v}{w} \disteq{\varepsilon_2} \diagbox{h}{v}{w}}{\diagbox{f}{v}{w} \disteq{\varepsilon_1 + \varepsilon_2} \diagbox{h}{v}{w}}{\mathsf{(Triang)}}
\end{equation*}
\vspace{1em}
\begin{equation*}
	\inferrule{\diagbox{f}{u}{v}\disteq{\varepsilon_1} \diagbox{h}{u}{v}\qquad \diagbox{g}{v}{w} \disteq{\varepsilon_2} \diagbox{i}{v}{w}}{
\tikzset{x=1em, y=2.1ex}
\InputIfFileExists{comp-sequential-fg.tikz}{}{\input{./tikz/comp-sequential-fg.tikz}}
\tikzset{x=1em, y=1.5ex}
  \disteq{\max\{\varepsilon_1,\varepsilon_2\}} 
\tikzset{x=1em, y=2.1ex}
\InputIfFileExists{comp-sequential-hi.tikz}{}{\input{./tikz/comp-sequential-hi.tikz}}
\tikzset{x=1em, y=1.5ex}
}{\mathsf{(Seq)}}\qquad
	\inferrule*[width=20pt]{ }{\diagbox{f}{v}{w} \disteq{1} \diagbox{g}{v}{w}}{\mathsf{(Top)}}
\end{equation*}
\vspace{1em}
\begin{equation*}
	\inferrule{\diagbox{f_1}{v_1}{w_1} \disteq{\varepsilon_1} \diagbox{g_1}{v_1}{w_1}\qquad \diagbox{f_2}{v_2}{w_2} \disteq{\varepsilon_2} \diagbox{g_2}{v_2}{w_2}}{
\tikzset{x=1em, y=2.1ex}
\InputIfFileExists{comp-parallel-fg.tikz}{}{\input{./tikz/comp-parallel-fg.tikz}}
\tikzset{x=1em, y=1.5ex}
  \disteq{\max\{\varepsilon_1,\varepsilon_2\}} 
\tikzset{x=1em, y=2.1ex}
\InputIfFileExists{comp-parallel-hi.tikz}{}{\input{./tikz/comp-parallel-hi.tikz}}
\tikzset{x=1em, y=1.5ex}
}{\mathsf{(Tens)}}
	\qquad
	\inferrule{\qquad}{\;\Bunitn{n} \;\disteq{0}\; 
\tikzset{x=1em, y=2.1ex}
\InputIfFileExists{co-del-d.tikz}{}{\input{./tikz/co-del-d.tikz}}
\tikzset{x=1em, y=1.5ex}
\;}{\mathsf{(Codel)}}
\end{equation*}
}
\caption{Quantitative axioms for regular behaviours.}
\label{fig:quantitative-axioms}
\end{figure}
For any two diagrams $f,g\from v\to w$ of $\Syn$, we say a quantitative equation $f\disteq{\varepsilon} g$ is \emph{valid} if $d^{\sem{v},\sem{w}}(\sem{f}, \sem{g}) \leq \varepsilon$. Analogously, an inference rule is valid if, whenever all quantitative equations in the premise are valid, then the equation in the conclusion in the equation is also valid. 

We now briefly explain each of the rules of our inference system depicted on \Cref{fig:quantitative-axioms}. \textsf{(Refl)}, \textsf{(Sym)} and \textsf{(Triang)} respectively capture reflexivity, symmetry and triangle inequality of pseudometric spaces. Importantly, rule \textsf{(Refl)} allows one to state that equal diagrams (modulo strictly equational rules described above) are at distance zero of each other. \textsf{(Top)} allows us to state that any two diagrams are at most within distance $1$ of each other, while \textsf{(Max)} allows one to always weaken the bound on the distance at which two diagrams are.  \textsf{(Cont)} is the key analytic inference rule here: it allows us to conclude that two diagrams are within distance $\varepsilon$, provided we can show that they are within all distances that are strictly greater than $\varepsilon$, thereby passing to the limit. 

Next, \textsf{(Seq)} and \textsf{(Tens)} relate the horizontal and vertical compositions of diagrams to the distance. There are two domain-specific rules that we use. \textsf{(Pref)} witnesses the fact that prefixing by the same actions decreases the distance by half. 
Note that it applies to diagram $\objr^m\to\objr^n$ and that we use $
\tikzset{x=1em, y=2.1ex}
\InputIfFileExists{m-a-m.tikz}{}{\input{./tikz/m-a-m.tikz}}
\tikzset{x=1em, y=1.5ex}
$ to denote the vertical composition of $m$ many $\scalar{a}$ generators, for any $m\in\mathbb{N}$. We note that the factor $\frac{1}{2}$ used in \Cref{def:edge} can be replaced with arbitrary real number $k \in \interval{0}{1}$, but this would require appropriately changing the rule \textsf{(Pref)}, as it is the case in the quantitative axiomatisation of the behavioural distance of DFA~\cite{Rozowski:2024:Complete}. Finally, \textsf{(Codel)} axiom encodes that any diagram with no initial state has no behaviour and is therefore at distance zero of the empty chart.

Through a straightforward structural induction, one can show the soundness of the proposed quantitative rules.
\begin{restatable}[Quantitative soundness]{theorem}{soundness}\label{thm:soundness}
Every derivable equation $f \disteq{\varepsilon} g$ is valid.
\end{restatable}

\begin{example}
We revisit the charts from \Cref{fig:intro}, which in Milner's syntax read $a.(a.0 + b.\mu v_1.a.v_1)+b.\mu v_1.a.v_1$ on the left and $\mu v_2. (a.v_2 + b.\mu v_1. a.a.v_1)$ on the right. We show that their distance is bounded by $\frac{1}{4}$, which is the same as showing that they agree up to depth of two. The derivation below shows how this bound is obtained \emph{syntactically}, by manipulating the diagrammatic representations through the rules of our inference system. We use compositionality to our advantage and break them into two parts which we will compose later with the \textsf{(Comp)} rule. In \circlednum{1}, we show that $a.(a.0 + b.v_1) + b.v_1$ and $\mu v_2.(a.v_2 + b.v_1)$ are in distance $\frac{1}{4}$. It is the case, as both processes can be distinguished at the depth of three; in particular, taking three successive $a$-transitions is possible in the right diagram, while in the left one, one reaches a deadlock after two $a$-transitions. In the proof, we use \textsf{(Pref)} rule twice, as we show that we agree for all prefixes of the depth of two. \\
	{
\circlednum{1}
\scalebox{0.65}
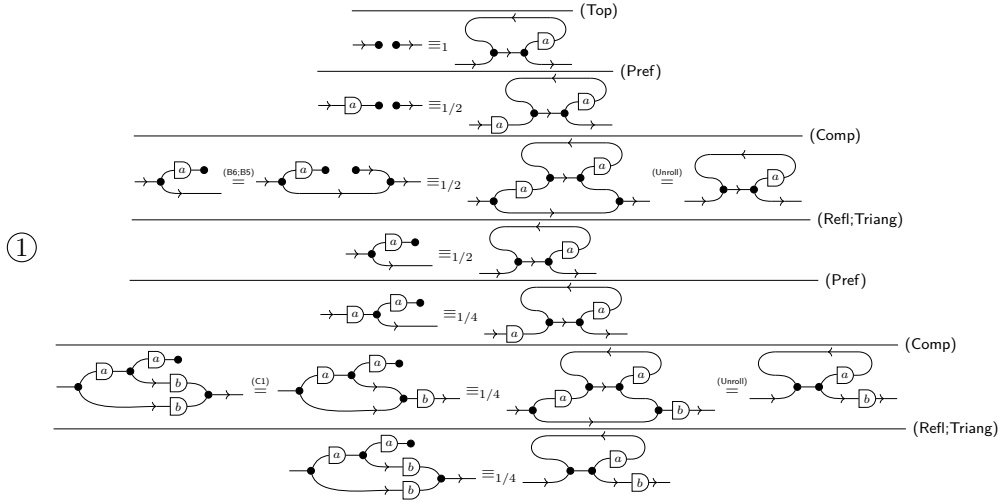
{
\[
\begin{gathered}
\inferrule*[Right={\textsf{(Refl;Triang)}}]{
\inferrule*[Right={\textsf{(Comp)}}]{
\inferrule*[Right={\textsf{(Pref)}}]{
\inferrule*[Right={\textsf{(Refl;Triang)}}]{
\inferrule*[Right={\textsf{(Comp)}}]{
\inferrule*[Right={\textsf{(Pref)}}]{
\inferrule*[Right={\textsf{(Top)}}]{ }{

\tikzset{x=1em, y=2.1ex}
\InputIfFileExists{delete-zero.tikz}{}{\input{./tikz/delete-zero.tikz}}
\tikzset{x=1em, y=1.5ex}
 \disteq{1} 
\tikzset{x=1em, y=2.1ex}
\InputIfFileExists{a-star.tikz}{}{\input{./tikz/a-star.tikz}}
\tikzset{x=1em, y=1.5ex}
}}
{
\tikzset{x=1em, y=2.1ex}
\InputIfFileExists{a-zero.tikz}{}{\input{./tikz/a-zero.tikz}}
\tikzset{x=1em, y=1.5ex}
 \disteq{1/2} 
\tikzset{x=1em, y=2.1ex}
\InputIfFileExists{a-a-star.tikz}{}{\input{./tikz/a-a-star.tikz}}
\tikzset{x=1em, y=1.5ex}
}}
{
\tikzset{x=1em, y=2.1ex}
\InputIfFileExists{copy-a-delxid.tikz}{}{\input{./tikz/copy-a-delxid.tikz}}
\tikzset{x=1em, y=1.5ex}
\;\myeq{B6;B5}\;
\tikzset{x=1em, y=2.1ex}
\InputIfFileExists{copy-a-zeroxid-merge.tikz}{}{\input{./tikz/copy-a-zeroxid-merge.tikz}}
\tikzset{x=1em, y=1.5ex}
 \disteq{1/2} 
\tikzset{x=1em, y=2.1ex}
\InputIfFileExists{a-star-unroll.tikz}{}{\input{./tikz/a-star-unroll.tikz}}
\tikzset{x=1em, y=1.5ex}
 \;\myeq{Unroll}\; 
\tikzset{x=1em, y=2.1ex}
\InputIfFileExists{a-star.tikz}{}{\input{./tikz/a-star.tikz}}
\tikzset{x=1em, y=1.5ex}
 }}
{
\tikzset{x=1em, y=2.1ex}
\InputIfFileExists{copy-a-delxid.tikz}{}{\input{./tikz/copy-a-delxid.tikz}}
\tikzset{x=1em, y=1.5ex}
 \disteq{1/2} 
\tikzset{x=1em, y=2.1ex}
\InputIfFileExists{a-star.tikz}{}{\input{./tikz/a-star.tikz}}
\tikzset{x=1em, y=1.5ex}
}}
{
\tikzset{x=1em, y=2.1ex}
\InputIfFileExists{a-copy-a-delxid.tikz}{}{\input{./tikz/a-copy-a-delxid.tikz}}
\tikzset{x=1em, y=1.5ex}
 \disteq{1/4} 
\tikzset{x=1em, y=2.1ex}
\InputIfFileExists{a-a-star.tikz}{}{\input{./tikz/a-a-star.tikz}}
\tikzset{x=1em, y=1.5ex}
}}
{
\tikzset{x=1em, y=2.1ex}
\InputIfFileExists{ex-diag-1-l.tikz}{}{\input{./tikz/ex-diag-1-l.tikz}}
\tikzset{x=1em, y=1.5ex}
\;\myeq{C1}\;
\tikzset{x=1em, y=2.1ex}
\InputIfFileExists{ex-diag-1-l-alt.tikz}{}{\input{./tikz/ex-diag-1-l-alt.tikz}}
\tikzset{x=1em, y=1.5ex}
 \disteq{1/4} 
\tikzset{x=1em, y=2.1ex}
\InputIfFileExists{a-star-unroll-b.tikz}{}{\input{./tikz/a-star-unroll-b.tikz}}
\tikzset{x=1em, y=1.5ex}
\;\myeq{Unroll}\;
\tikzset{x=1em, y=2.1ex}
\InputIfFileExists{a-star-b.tikz}{}{\input{./tikz/a-star-b.tikz}}
\tikzset{x=1em, y=1.5ex}
 }}
{
\tikzset{x=1em, y=2.1ex}
\InputIfFileExists{ex-diag-1-l.tikz}{}{\input{./tikz/ex-diag-1-l.tikz}}
\tikzset{x=1em, y=1.5ex}
 \disteq{1/4} 
\tikzset{x=1em, y=2.1ex}
\InputIfFileExists{a-star-b.tikz}{}{\input{./tikz/a-star-b.tikz}}
\tikzset{x=1em, y=1.5ex}
}
\end{gathered}
\]
}
}

In the rules labelled with $(\mathsf{Refl};\mathsf{Triang})$ we have used the strict equality of diagrams to simplify the quantitative equations. The equalities marked with \textsf{(Unroll)} follow from \Cref{lem:unroll}, which we state in the next section. Then, for the second part of our diagrams, we show that $\mu v_1. a.v_1$ is bisimilar to $\mu v_1. a.a.v_1$. In \circlednum{2}, we first show that we can easily conclude that both diagrams are in a distance bounded by $\frac{1}{2}$. \\

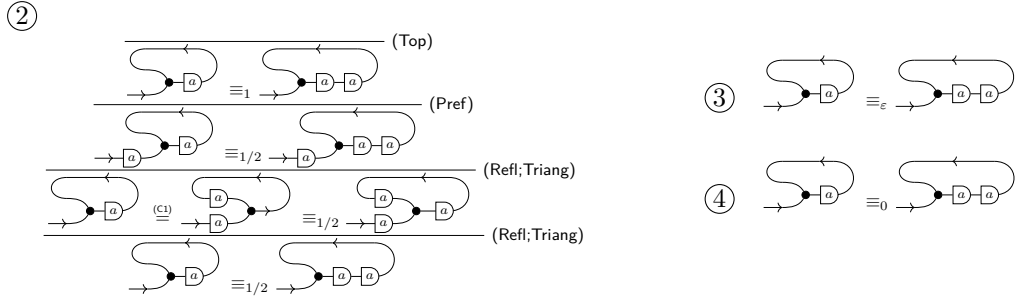
\hspace{0.5em}
{
\begin{minipage}{.65\textwidth}
\circlednum{2}\\
	\scalebox{0.7}{
\[
\begin{gathered}
\inferrule*[Right={\textsf{(Refl;Triang)}}]{
\inferrule*[Right={\textsf{(Refl;Triang)}}]{
\inferrule*[Right={\textsf{(Pref)}}]{
\inferrule*[Right={\textsf{(Top)}}]{ }{

\tikzset{x=1em, y=2.1ex}
\InputIfFileExists{ex-diag-1-r.tikz}{}{\input{./tikz/ex-diag-1-r.tikz}}
\tikzset{x=1em, y=1.5ex}
\disteq{1} 
\tikzset{x=1em, y=2.1ex}
\InputIfFileExists{ex-diag-2-r.tikz}{}{\input{./tikz/ex-diag-2-r.tikz}}
\tikzset{x=1em, y=1.5ex}
}}
{
\tikzset{x=1em, y=2.1ex}
\InputIfFileExists{a-ex-diag-1-r.tikz}{}{\input{./tikz/a-ex-diag-1-r.tikz}}
\tikzset{x=1em, y=1.5ex}
\disteq{1/2} 
\tikzset{x=1em, y=2.1ex}
\InputIfFileExists{a-ex-diag-2-r.tikz}{}{\input{./tikz/a-ex-diag-2-r.tikz}}
\tikzset{x=1em, y=1.5ex}
}}
{
\tikzset{x=1em, y=2.1ex}
\InputIfFileExists{ex-diag-1-r.tikz}{}{\input{./tikz/ex-diag-1-r.tikz}}
\tikzset{x=1em, y=1.5ex}
\;\myeq{C1}\;
\tikzset{x=1em, y=2.1ex}
\InputIfFileExists{a-ex-diag-1-r-alt.tikz}{}{\input{./tikz/a-ex-diag-1-r-alt.tikz}}
\tikzset{x=1em, y=1.5ex}
\disteq{1/2} 
\tikzset{x=1em, y=2.1ex}
\InputIfFileExists{a-ex-diag-2-r-alt.tikz}{}{\input{./tikz/a-ex-diag-2-r-alt.tikz}}
\tikzset{x=1em, y=1.5ex}
}}
{
\tikzset{x=1em, y=2.1ex}
\InputIfFileExists{ex-diag-1-r.tikz}{}{\input{./tikz/ex-diag-1-r.tikz}}
\tikzset{x=1em, y=1.5ex}
\disteq{1/2} 
\tikzset{x=1em, y=2.1ex}
\InputIfFileExists{ex-diag-2-r.tikz}{}{\input{./tikz/ex-diag-2-r.tikz}}
\tikzset{x=1em, y=1.5ex}
}
\end{gathered}
\]}
\end{minipage}
\begin{minipage}{.34\textwidth}
	\circlednum{3}\quad
	\scalebox{0.7}{\[
\tikzset{x=1em, y=2.1ex}
\InputIfFileExists{ex-diag-1-r.tikz}{}{\input{./tikz/ex-diag-1-r.tikz}}
\tikzset{x=1em, y=1.5ex}
\disteq{\varepsilon} 
\tikzset{x=1em, y=2.1ex}
\InputIfFileExists{ex-diag-2-r.tikz}{}{\input{./tikz/ex-diag-2-r.tikz}}
\tikzset{x=1em, y=1.5ex}
\]}\\
	\vspace{0.5em}
	\\
	\circlednum{4}\quad
	\scalebox{0.7}{\[
\tikzset{x=1em, y=2.1ex}
\InputIfFileExists{ex-diag-1-r.tikz}{}{\input{./tikz/ex-diag-1-r.tikz}}
\tikzset{x=1em, y=1.5ex}
\disteq{0} 
\tikzset{x=1em, y=2.1ex}
\InputIfFileExists{ex-diag-2-r.tikz}{}{\input{./tikz/ex-diag-2-r.tikz}}
\tikzset{x=1em, y=1.5ex}
\]}\\
\end{minipage}

}
Then, using the same reasoning, we can inductively show \circlednum{3} for $\varepsilon = 1/2^n$ for any $n\in\mathbb{N}$, and thus for any $\varepsilon>0$. Finally, the \textsf{(Cont)} rule allows us to conclude \circlednum{4}. Finally, combining \circlednum{1} and \circlednum{4} together with the \textsf{(Comp)} rule, allows us to recover the equality we wanted to show:\\
\scalebox{0.85}{\[
\tikzset{x=1em, y=2.1ex}
\InputIfFileExists{ex-diag-1.tikz}{}{\input{./tikz/ex-diag-1.tikz}}
\tikzset{x=1em, y=1.5ex}
\disteq{1/4} 
\tikzset{x=1em, y=2.1ex}
\InputIfFileExists{ex-diag-2.tikz}{}{\input{./tikz/ex-diag-2.tikz}}
\tikzset{x=1em, y=1.5ex}
\]}
\end{example}

\section{Completeness}
\label{sec:completeness}
We finally arrive at the main technical section of the paper, where we gradually present a sequence of results leading to the completeness of our axioms for the behavioural distance. First, we show that we can safely focus solely on $\objr^m \to \objr^n$ diagrams. Then, through an analytic proof relying on \textsf{(Cont)} rule, we show that each diagram can be \emph{co-copied}. Consequently, we can decompose any $\objr^m \to \objr^n$ diagram into a collection of $\objr \to \objr^n$ diagrams, which correspond to individual charts. We argue that each of those $\objr \to \objr^n$ diagrams has a normal form from which one can extract a finite prechart structure. The distance between states of this prechart precisely captures the distance between the diagrams being related and very importantly, admits a simple characterisation that can be simulated through the inference rules of our system, eventually leading to the completeness result.

\noindent
\textbf{Left-to-right diagrams. } The following result allows us to turn any bidirectional diagram into a left-to-right one by appropriately \emph{bending} the wires using $
\tikzset{x=1em, y=2.1ex}
\InputIfFileExists{cap-down.tikz}{}{\input{./tikz/cap-down.tikz}}
\tikzset{x=1em, y=1.5ex}
$ and $
\tikzset{x=1em, y=2.1ex}
\InputIfFileExists{cup-down.tikz}{}{\input{./tikz/cup-down.tikz}}
\tikzset{x=1em, y=1.5ex}
$. We will see later that this process does not change distances between diagrams. 
\begin{restatable}{lemma}{compact}\label{lem:compact}
There are bijections between the sets $\Syn(v_1\!\objl\! v_2,w)$ and $\Syn(v_1v_2,w\!\objr)$, and between $\Syn(v,w_1\!\objl\! w_2)$ and $\Syn(v\!\objr,w_1w_2)$, \emph{i.e.} between sets of string diagrams of the form
$$

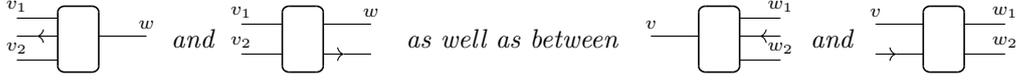
\tikzset{x=1em, y=2.1ex}
\InputIfFileExists{wrong-way-left.tikz}{}{\input{./tikz/wrong-way-left.tikz}}
\tikzset{x=1em, y=1.5ex}
\;\text{ and }\; 
\tikzset{x=1em, y=2.1ex}
\InputIfFileExists{right-way-right.tikz}{}{\input{./tikz/right-way-right.tikz}}
\tikzset{x=1em, y=1.5ex}
\quad\text{as well as between}\quad
 
\tikzset{x=1em, y=2.1ex}
\InputIfFileExists{wrong-way-right.tikz}{}{\input{./tikz/wrong-way-right.tikz}}
\tikzset{x=1em, y=1.5ex}
\;\text{ and }\; 
\tikzset{x=1em, y=2.1ex}
\InputIfFileExists{right-way-left.tikz}{}{\input{./tikz/right-way-left.tikz}}
\tikzset{x=1em, y=1.5ex}
 
$$
where $v,w, v_i, w_i$ are words over $\{\objr,\objl\}$.
\end{restatable}
\noindent
A \emph{block} is simply a diagram freely composed from a restricted set of generators (possibly including identities and symmetries). In this paper, we will make use of two special kinds of diagrams that can be factored into blocks:
\begin{itemize}[topsep=0pt]
	\item A \emph{matrix-diagram} is a diagram $\objr^m\to\objr^n$ that factors as a composition of a block of $
\tikzset{x=1em, y=2.1ex}
\InputIfFileExists{lr-copy.tikz}{}{\input{./tikz/lr-copy.tikz}}
\tikzset{x=1em, y=1.5ex}
, 
\tikzset{x=1em, y=2.1ex}
\begin{tikzpicture}
	\begin{pgfonlayer}{nodelayer}
		\node [style=black] (37) at (0.75, 0) {};
		\node [style=none] (43) at (0.25, 0) {};
		\node [style=none] (44) at (-0.5, 0) {};
	\end{pgfonlayer}
	\begin{pgfonlayer}{edgelayer}
		\draw (43.center) to (37);
		\draw [->] (44.center) to (43.center);
	\end{pgfonlayer}
\end{tikzpicture}
}
\tikzset{x=1em, y=1.5ex}
$, another of $\scalar{a}$ for $a\in \Sigma$, and a last one of $
\tikzset{x=1em, y=2.1ex}
\InputIfFileExists{lr-merge.tikz}{}{\input{./tikz/lr-merge.tikz}}
\tikzset{x=1em, y=1.5ex}
, 
\tikzset{x=1em, y=2.1ex}
\begin{tikzpicture}
	\begin{pgfonlayer}{nodelayer}
		\node [style=black] (37) at (-0.5, 0) {};
		\node [style=none] (43) at (0.25, 0) {};
		\node [style=none] (44) at (0.75, 0) {};
	\end{pgfonlayer}
	\begin{pgfonlayer}{edgelayer}
		\draw [->] (37) to (43.center);
		\draw (44.center) to (43.center);
	\end{pgfonlayer}
\end{tikzpicture}
}
\tikzset{x=1em, y=1.5ex}
$. A matrix-diagram is \emph{guarded} when any path from left to right port encounters at least one $\scalar{a}$. In \Cref{lem:guarded-precompose} in the appendix of the full version of the paper, we show that one can generalise \textsf{(Pref)} inference rule to arbitrary guarded matrix-diagrams, rather than just vectors of $\scalar{a}$. In other words, prepending a guarded matrix-diagram to any pair of diagrams contracts the distance between them. 
	\item A \emph{relation-diagram} is a diagram $\objr^m\to\objr^n$ that factors as a composition of a block of $
\tikzset{x=1em, y=2.1ex}
\InputIfFileExists{lr-copy.tikz}{}{\input{./tikz/lr-copy.tikz}}
\tikzset{x=1em, y=1.5ex}
, 
\tikzset{x=1em, y=2.1ex}
\begin{tikzpicture}
	\begin{pgfonlayer}{nodelayer}
		\node [style=black] (37) at (0.75, 0) {};
		\node [style=none] (43) at (0.25, 0) {};
		\node [style=none] (44) at (-0.5, 0) {};
	\end{pgfonlayer}
	\begin{pgfonlayer}{edgelayer}
		\draw (43.center) to (37);
		\draw [->] (44.center) to (43.center);
	\end{pgfonlayer}
\end{tikzpicture}
}
\tikzset{x=1em, y=1.5ex}
$ followed by the block of $
\tikzset{x=1em, y=2.1ex}
\InputIfFileExists{lr-merge.tikz}{}{\input{./tikz/lr-merge.tikz}}
\tikzset{x=1em, y=1.5ex}
, 
\tikzset{x=1em, y=2.1ex}
\begin{tikzpicture}
	\begin{pgfonlayer}{nodelayer}
		\node [style=black] (37) at (-0.5, 0) {};
		\node [style=none] (43) at (0.25, 0) {};
		\node [style=none] (44) at (0.75, 0) {};
	\end{pgfonlayer}
	\begin{pgfonlayer}{edgelayer}
		\draw [->] (37) to (43.center);
		\draw (44.center) to (43.center);
	\end{pgfonlayer}
\end{tikzpicture}
}
\tikzset{x=1em, y=1.5ex}
$.
\end{itemize} 
Intuitively, matrix-diagrams are representations of labelled transition relations, while relation-transition diagrams are representations of the output relations. This idea can be captured formally through the notion of a \emph{representation}. For a diagram $d\from \objr^m\to \objr^n$, a \emph{representation} is a pair $(a,o)$ of a guarded matrix-diagram $c\from \objr^{\ell+m}\to \objr^{\ell+m}$ and a relation-diagram $o\from \objr^{\ell+m}\to \objr^n$, such that
\begin{equation}\label{eq:rep}
\scalebox{0.8}{
\dbox{d}{m}{n} \; = \; 
\tikzset{x=1em, y=2.1ex}
\InputIfFileExists{automata-rep-star.tikz}{}{\input{./tikz/automata-rep-star.tikz}}
\tikzset{x=1em, y=1.5ex}
\; := \;
\tikzset{x=1em, y=2.1ex}
\InputIfFileExists{automata-rep.tikz}{}{\input{./tikz/automata-rep.tikz}}
\tikzset{x=1em, y=1.5ex}

}
\end{equation} 
\noindent
Using the rules of ${=}$, we can rearrange the any diagram into the form described above.
\begin{restatable}{theorem}{representation}
\label{thm:representation}
Any diagram $\objr^m\to \objr^n$ has a representation.
\end{restatable}
\noindent
The matrix-diagram in the representation that is being fed through feedback, can be \emph{unrolled}.
\begin{restatable}{lemma}{unrolling}\label{lem:unroll}
For any matrix-diagram $d\from \objr^n\to \objr^n$, \scalebox{0.7}{
$
\tikzset{x=1em, y=2.1ex}
\InputIfFileExists{d-star.tikz}{}{\input{./tikz/d-star.tikz}}
\tikzset{x=1em, y=1.5ex}
 \;= \; 
\tikzset{x=1em, y=2.1ex}
\InputIfFileExists{d-star-unroll.tikz}{}{\input{./tikz/d-star-unroll.tikz}}
\tikzset{x=1em, y=1.5ex}
$}.
\end{restatable}
\noindent
We informally state the following connection between representations and charts. 
\begin{remark}
	The diagrammatic syntax presented in our paper is \emph{expressive}, that is, the semantics of every chart can be equivalently encoded as an appropriate diagram. More explicitly, to encode a chart $(Q,s,D,E)$ we first define its transition relation $D$ as a matrix-diagram $c$ and its outputs $E$ as a relation diagram $o$; we then compose them together as in~\eqref{eq:rep} following the definition of representations above. Semantically, this corresponds to taking a dagger (\Cref{lem:regbehconway}), which solves the system of equations (described by $c$ and $o$) using $\mu$ operation on charts, recovering the original expressivity result for Milner's operations on finite charts~\cite[Corollary~5.8]{Milner:1984:Complete}, as summarised in \Cref{sec:operations}. Co-deleting all input wires except the one corresponding to the initial state $s \in Q$, yields a diagram whose semantics coincides precisely with the semantics of the chart of interest.
\end{remark}
\textbf{Co-copying.} Bringing each diagram to a form corresponding to its representation, combined with the usage of unrolling (\Cref{lem:unroll}) and generalisation of \textsf{(Pref)} to guarded matrix-diagrams (\Cref{lem:guarded-precompose}) allows to show that the following two diagrams are arbitrarily close and hence by \textsf{(Cont)} are in zero distance
\begin{restatable}{theorem}{globalcocopy}
\label{thm:co-copy-delete}
For any diagram $d\from \objr^m\to \objr^n$, we have that \scalebox{0.8}{$
\tikzset{x=1em, y=2.1ex}
\InputIfFileExists{global-merge.tikz}{}{\input{./tikz/global-merge.tikz}}
\tikzset{x=1em, y=1.5ex}
\quad \disteq{0} \quad 
\tikzset{x=1em, y=2.1ex}
\InputIfFileExists{global-merge-1.tikz}{}{\input{./tikz/global-merge-1.tikz}}
\tikzset{x=1em, y=1.5ex}
$}.
\end{restatable}
\noindent
\vspace{-2pt}
A very important consequence of the above combined with the usage of \textsf{(Codel)} rule is the fact that each $\objr^m\to \objr^n$ diagram can be separated to a collection of $\objr\to \objr^n$ intuitively corresponding to individual entries of tuples being manipulated in $\RegBeh$.
\begin{lemma}
\label{lem:injections}
Let $f\from \objr^m\to \objr^n$ and define $f_i$ to be the diagram $\objr\to \objr^n$ obtained by composing all but the $i$-th input of $f$ with $\Bunit$ (co-deleting all inputs except the $i$-th one). We have that

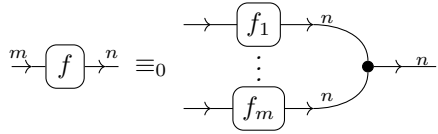
\[\dbox{f}{m}{n} \;\disteq{0} \;
\tikzset{x=1em, y=2.1ex}
\InputIfFileExists{d-decomposition.tikz}{}{\input{./tikz/d-decomposition.tikz}}
\tikzset{x=1em, y=1.5ex}
\]
\end{lemma}
Moreover, in \Cref{lem:distance_on_merge} in the appendix of the full version of the paper, we show that the behavioural distance between the denotations of arbitrary diagrams $f,g \colon \objr^m \to \objr^n$ is simply the maximum of the component-wise distances between each $f_i$ and $g_i$ for $i \in \{1, \dots, m\}$. From now on, we will temporarily shift focus to $\objr \to \objr^n$ diagrams.

\noindent
\textbf{One-to-n diagrams. } Each of the $\objr \to \objr^n$ diagrams represents a behaviour of the single state of the prechart structure on $\Omega$ or equivalently, defines a chart. Turns out that appropriately combining these diagrams corresponds to operations on charts described in \Cref{sec:operations}.  
\begin{restatable}{lemma}{sumconvolutions}\label{lem:sum-convolution}
For any two $c,d\from\objr\to\objr^n$, we have that \scalebox{0.8}{$
\sem{
\tikzset{x=1em, y=2.1ex}
\InputIfFileExists{convolution-cxd.tikz}{}{\input{./tikz/convolution-cxd.tikz}}
\tikzset{x=1em, y=1.5ex}
} = \sem{e} + \sem{f}$}.
\end{restatable}
\noindent
This operation of combining two $\objr \to \objr^n$ string diagrams (that we call \emph{convolution}) can easily be extended to any finite collection of $\objr \to \objr^n$ diagrams and is well defined up to permutations and removing duplicates, while staying at distance zero; (see \Cref{lem:convolution-aci}). Thanks to the ability to turn each diagram $f \colon \objr^m \to \objr^n$ into its representation (\Cref{thm:representation}), as well as global co-copying (\Cref{lem:injections}), we can express each the subdiagrams $f_i \colon \objr \to \objr^n$ as  convolution.
\begin{restatable}{lemma}{fundamental}\label{lem:fundamental}
	For any diagram $f \colon \objr^m \to \objr^n$ and $f_i$, $1 \leq i \leq m$ defined as above, for all $i \in \{1, \dots,m\}$, we can derive

	\[
	\dbox{f_i}{}{n} \;\disteq{0} 
\tikzset{x=1em, y=2.1ex}
\InputIfFileExists{state-eq.tikz}{}{\input{./tikz/state-eq.tikz}}
\tikzset{x=1em, y=1.5ex}

	\]
	
	\noindent
	where, for $1\leq j\leq \ell$, each $v_{q_j}\from \objr\to \objr^n$ is a diagram encoding the output variables to which the $i$-th input wire of $f$ is directly connected, that is, without going through any $\scalar{a}$ generator (in particular, each $v_{q_j}$ is a monoidal product of a single identity with $n-1$ $\Bunit$ generators).
\end{restatable}
\noindent
The informal intuition is that each of the $f_i \colon \objr \to \objr^n$ diagrams represents a state of a prechart and the behaviour of each such state is the union of all possible labelled transitions to other states and variable outputs. This can be made formal, by extracting a prechart over the set $Q_f = \{f_1, \dots, f_m\}$, whose transition function $\beta$ is given by the following; we define $f_i \tr{a} f_{j}$ iff $\dbox{f_i}{}{n}$ contains $
\tikzset{x=1em, y=2.1ex}
\InputIfFileExists{prefix-transition.tikz}{}{\input{./tikz/prefix-transition.tikz}}
\tikzset{x=1em, y=1.5ex}
$ and similarly $f_i \rhd v_s$ iff $\dbox{f_i}{}{n}$ contains $\dbox{v_s}{}{n}$. The behavioural distance between states of this prechart, precisely captures the behavioural distance between each of the $f_i$ diagrams.
\begin{restatable}{lemma}{extracteddist}\label{cor:dist-chart-diag}
	For all $f_i, f_j \in Q_f$, we have that $\mathsf{bd}_{\beta}(f_i,f_j) = \mathsf{bd}(\sem{f_i}, \sem{f_j})$
\end{restatable}

\noindent
\textbf{Completeness.} The behavioural distance between states of finite charts, including those derived from the normal form of diagrams (\Cref{lem:fundamental}), can be given a simpler characterisation via a decreasing chain of approximants.
\begin{restatable}{lemma}{finitedist}\label{lem:finite_dist}
	Let $(Q, \beta)$ be a finite prechart. The behavioural distance between any pair $q_1, q_2 \in Q$ of states can be calculated by $\mathsf{bd}_\beta(q_1,q_2) = \inf_{p \in \N} \left \{ \Phi^{(p)}_\beta(q_1, q_2) \right\}$, where $\Phi^{(0)}_\beta$ is a discrete pseudometric and for any $p \in \N$, we define $\Phi^{(p+1)}_\beta = \Phi_\beta \left( \Phi^{(p)}_\beta\right)$.
\end{restatable}
\noindent
The proof of the fact above makes use of the fact that the map $\Phi_\beta$ has a \emph{unique fixpoint} and when dealing with finite precharts, we can establish necessary preconditions allowing to use Kleene's fixpoint theorem for the greatest fixpoint. Details of these arguments can be found in \Cref{apx:dist}. The characterisation described above is particularly useful, as upper bounds on each of the approximants can be derived syntactically through the means of axiomatic manipulation.
\begin{restatable}{lemma}{approximation}\label{lem:approximation}
	Let $f\from \objr^m\to \objr^n$, $f_i$, $1\leq i \leq m$ and $(Q_f, \beta)$ be defined as above. For all $f_g, f_h \in Q_f$, all $p \in \N$ and any $\varepsilon \geq \Phi^{(p)}_\beta({f_g}, {f_h})$, we have that
	$\dbox{f_g}{}{n} \disteq{\varepsilon} \dbox{f_h}{}{n} $ is derivable.
\end{restatable}
\noindent
The key idea of the lemma above is that the behaviour of lifting of pseudometric from states to edges (\Cref{def:edge}) and Hausdorff lifting (\Cref{def:hausdorff}) used in the definition of $\Phi$ can be simulated using the rules of our deduction system. Using \textsf{(Cont)} rule and characterisation from \Cref{lem:finite_dist}, we obtain a completeness result for distances between $f_i \colon \objr \to \objr^n$ components of $f \colon \objr^m \to \objr^n$.
\begin{restatable}{lemma}{onemcompleteness}\label{lem:one-to-m-completeness}
	Let $f\from \objr^m\to \objr^n$ and $f_i$, $1\leq i \leq m$ be defined as above. For all $g,h\in\{1,\dots,m\}$, any valid equation $\dbox{f_g}{}{n} \disteq{\varepsilon} \dbox{f_h}{}{n}$ is derivable.
\end{restatable}
\noindent
Relying on \Cref{lem:injections}, we can reduce the problem of deriving distance between arbitrary $\objr^m \to \objr^n$ diagrams, to the case of $\objr \to \objr^n$ solved above. This yields completeness result for left-to-right diagrams. 
\begin{restatable}{theorem}{completeness}\label{lem:mncompleteness}
Let $f,g \colon \objr^m \to \objr^n$. Any valid equation $\dbox{f}{m}{n} \disteq{\varepsilon} \dbox{g}{m}{n}$ is derivable.
\end{restatable}
\noindent
Since arbitrary diagrams can be turned into $\objr^m \to \objr^n$ diagrams by appropriately composing $
\tikzset{x=1em, y=2.1ex}
\InputIfFileExists{cup-down.tikz}{}{\input{./tikz/cup-down.tikz}}
\tikzset{x=1em, y=1.5ex}
$ and $
\tikzset{x=1em, y=2.1ex}
\InputIfFileExists{cap-down.tikz}{}{\input{./tikz/cap-down.tikz}}
\tikzset{x=1em, y=1.5ex}
$, while preserving distances between diagrams, we arrive at the desired result.
\begin{restatable}[Quantitative completeness]{theorem}{generalcompleteness}\label{thm:completeness}
	Let $f,g \colon v \to w$ be two arbitrary diagrams. Any valid equation  $\diagbox{f}{v}{w} \disteq{\varepsilon} \diagbox{g}{v}{w}$ is derivable.
\end{restatable}

\section{Discussion}
\label{sec:discussion}

In this paper, we presented a sound and complete quantitative axiomatisation of the behavioural distance of Milner's charts~\cite{Milner:1984:Complete}. We have relied on a compositional, string diagrammatic syntax~\cite{piedeleu2023finite,antoinecsl2025} and equipped it with a quantitative inference system for reasoning about bounds on behavioural distance, inspired by recent advances in metric universal algebra~\cite{Mardare:2016:Quantitative,MiliusU19,Mio:2024:Universal}.

Originally introduced for probabilistic systems~\cite{Breugel:2001:Towards,Desharnais:2004:Metrics}, behavioural distances have recently been generalised to a broad range of systems modelled via the abstract framework of universal coalgebra~\cite{Rutten:2000:Universal}, leveraging pseudometric liftings of functors~\cite{Baldan:2018:Coalgebraic}. The notion of behavioural distance for charts used in the paper is an instance of this coalgebraic framework. Our concrete characterisation closely resembles the metric on trees studied by Nivat~\cite{Nivat:1979:Infinite}. A similar characterisation was studied by Golson and Rounds~\cite{Golson:1984:Connections}, who instead examined de Bakker and Zucker's metric domain for nondeterministic processes~\cite{Bakker:1982:Processes}, also derived via a fixpoint construction involving the Hausdorff distance. 

The idea of reasoning about distances between string diagrams has been explored before in quantum theory~\cite{kissinger2017pictureperfectquantumkeydistribution,Breiner:2019:Graphical,HLarsen2021} and probability theory~\cite{Perrone:2024:Markov}. However, in contrast with the growing body of work on cartesian quantitive algebra~\cite{Mardare:2016:Quantitative,Mio:2024:Universal, Bacci:2024:Sum}, a systematic foundation to axiomatising distances between string diagrams appeared only very recently, in the work of Lobbia et al~\cite{Lobbia:2024:Quantitative}. Besides the basic examples provided in~\cite{Lobbia:2024:Quantitative}, our work is the first to propose a complete axiomatisation of a quantitative calculus of string diagrams. The approach of~\cite{Lobbia:2024:Quantitative} is based on enriched category theory: similarly, one could observe that the equipment of our semantic category with pseudometric structures making sequential and parallel composition nonexpansive yields the enrichment in the category of pseudometric spaces. However, our axiomatisation relies on the domain-specific implicational rule \textsf{(Pref)} and an axiom schema \textsf{(Codel)} that cannot be expressed in the framework of Lobbia et al~\cite{Lobbia:2024:Quantitative}, which only supports quantitative equations. Reconciling those rule formats with the general framework of Lobbia et al~\cite{Lobbia:2024:Quantitative} is an interesting direction for future research.

Axiomatising behavioural distances have been originally studied through ad-hoc inference systems~\cite{Larsen:2011:Metrics,Argenio:2014:Axiomatizing}. The introduction of quantitative equational theories made more principled approaches possible, leading to axiomatisations of behavioural distance for probabilistic systems~\cite{Bacci:2018:Complete1,Bacci:2018:Algebraic}. In recent work, R\'{o}\.{z}owski~\cite{Rozowski:2024:Complete} extended these results within a coalgebraic framework, focusing on the simple case of DFA, which enjoys a straightforward algebraic representation via the syntax of Kleene Algebra. While we rely here on the general pattern of the completeness proof from that work, the case of Milner's charts was significantly more involved, requiring the ability to simulate the behaviour of Hausdorff lifting syntactically.

In constructing the semantic category, we have used the fact that charts form a Conway theory~\cite{Esik:1999:Group}, studied in the literature on parametrised fixpoint operators~\cite{Haghverdi:2000:Categorical,Simpson:2000:Complete,Abramsky:2002:Geometry}, and which can be seen as a relaxation of iteration theories~\cite{Bloom:1993:Iteration}. The connection between charts and these structures was previously investigated by Bloom et al~\cite{Bloom:1993:IterationTrees} and Sewell~\cite{Sewell:1995:Algebra}, while the interplay of parametrised fixpoint operators with traced monoidal categories was studied by Hasegawa~\cite{Hasegawa:1997:Models}, Haghverdi~\cite{Haghverdi:2000:Categorical}, and Simpson and Plotkin~\cite{Simpson:2000:Complete} independently.

Our work constitutes a necessary first step towards a similar diagrammatic treatment of behavioural distance for quantitative automata, such as probabilistic and weighted systems, for which distances are a more suitable way of reasoning rather than Boolean equivalences. The string diagrammatic point of view would enable a desirable, compositional treatment that reflects well the underlying operational models, in a way that is not available through conventional syntaxes. Another promising direction for future work would be to consider Guarded Kleene Algebra with Tests (GKAT)~\cite{Smolka:2020:Guarded,Schmid:2021:Guarded}, an efficiently decidable language for reasoning about equivalence of uninterpreted programs. The completeness proof of GKAT relies on a metric argument that could be internalised within an inference system like the one introduced in this paper. Moreover, the syntax of GKAT is insufficiently expressive, as it can only describe a part of the behaviours of its underlying operational model -- in fact there is no finite purely algebraic syntax that could do so~\cite{Cate:2025:Algebras}. A string diagrammatic treatment could allow us to express all such behaviours and obtain a simpler yet more expressive completeness result, that would in turn enable axiomatic reasoning about decompilation algorithms, that were recently shown to be expressible via GKAT automata (but not in GKAT)~\cite{Zhang:2025:Efficient}.



\bibliography{icalp}

\appendix

\section{Appendix}

\subsection{Algebra of Regular Behaviours}
In the appendix, we rely on a slightly different, but equivalent~\cite{Milner:1984:Complete}, characterisation of the set $\Omega$ of all regular behaviours. In this section, we elaborate on this characterisation and related conventions. In particular, we will write ${\Expr}/{\sim}$ instead of $\Omega$, as we will view this set as a bisimilarity quotient of a certain specification language that can be given operational semantics by equipping it with a prechart structure. The syntax of the language corresponds to operations described in \Cref{sec:operations} and is given by the following syntax
$$e,f \in \Expr ::= 0 \mid v \in V \mid a.e \mid e + f \mid \mu v. e$$
where $V=\{v_1, v_2, \dots\}$ and $\alphabet$ be sets of \emph{variables} and \emph{letters} respectively. Given an expression $f$ containing a variable $v$, we say that $v$ is \emph{free} in $f$, if it appears outside of the scope of the $\mu v.e$ operator or say that it is \emph{bound} otherwise. Given an expression $e \in \Expr$, we write $\fv(e) \subseteq V$ for the set of its free variables. Given a set $X$, we define $\mathcal{B}X$ to be $\Pfin(V + \alphabet \times X)$. Recall that each prechart can be equivalently seen as a pair $(X, \beta \colon X \to \mathcal{B}X)$. Given a prechart $(X, \beta)$, we write $x \tr{a}_\beta y \iff (a,y) \in \beta(x)$ and $x \rhd_\beta v_i \iff v_i \in \beta(x)$. When $\beta$ is obvious from the context, we will omit writing it in the subscript.  One can easily verify that $\sim_\beta$ is a prechart bisimulation itself and we will refer to it as the \emph{greatest bisimulation} on $(X,\beta)$. When $\beta$ is obvious from the context we will omit writing it in the subscript. 
\begin{definition}
	Let $(X, \beta)$ be a prechart and let $Y \subseteq X$. We call $(Y, \beta)$ a subprechart of $(X, \beta)$, if the canonical inclusion map $i \colon Y \to X$ is a prechart homomorphism. 
\end{definition}
Given a prechart $(X, \beta)$ and a state $x$ we write $\langle x \rangle_{\beta} \subseteq X$ for the set of states reachable from $x$ (note that this definition can be easily extended to sets of states). Equipping in with $\beta$ restricted to $\langle x \rangle_\beta$ defines a subprechart $(\langle x \rangle_\beta, \beta)$ of $(X, \beta)$. 		
\begin{definition}[{\cite{Milner:1984:Complete}}]\label{def:subset}
	Given vectors $\vec{v}$ of binders and $\vec{e}$ of expressions of the same size, we define a syntactic substitution operator $[\vec{e}/\vec{v}] \colon \Expr \to \Expr$ by the following
	\begin{gather*}
		v[\vec{e}/\vec{v}] = \begin{cases}
			\vec{e}_i & \text{if }v=\vec{v}_i\\
			v & \text{otherwise}
		\end{cases}
		\qquad
		(a.e)[\vec{e}/\vec{v}] = a.(e[\vec{e}/\vec{v}])\\
		(e + f)[\vec{e}/\vec{v}] = e[\vec{e}/\vec{v}] + f[\vec{e}/\vec{v}]\\
		(\mu w.e)[\vec{e}/\vec{v}] = \begin{cases}
			\mu w. (e[\vec{e}/\vec{v}]) & \text{if } w \text{ is not in } \vec{v} \text{ nor free in } \vec{e}\\
			\mu w. (e[z/w][\vec{e}/\vec{v}]) & \text{otherwise for some } z \text{ not in } \vec{v} \text{ nor free in } \vec{e}\\
		\end{cases}
	\end{gather*}
\end{definition}
\begin{definition}[{\cite{Milner:1984:Complete}}]\label{def:operational_semantics}
	Let $(\Expr, \partial)$ be a prechart whose transition function (called \emph{derivative}) is a least one satisfying the following inference rules
	\begin{gather*}
		\inferrule{e \tr{a} e' }{a.e \tr{a} e'} \qquad \inferrule{ }{v \rhd v} \qquad \inferrule{e \tr{a} e'}{e + f \tr{a} e'} \qquad \inferrule{f \tr{a} f'}{e + f \tr{a} f'}\\
		\inferrule {e \rhd v}{e + f \rhd v} \qquad \inferrule {f \rhd v}{e + f \rhd v} \qquad \inferrule{e \rhd v\quad v \neq w}{\mu w.e \rhd v } \qquad \inferrule{e \tr{a} e'}{\mu v. e \tr{a} e'[\mu v .e / v]}
	\end{gather*}
\end{definition}
The syntactic prechart $\langle e \rangle_\partial$ is locally finite.
\begin{lemma}[{\cite[Proposition~5.1]{Milner:1984:Complete}}]
	For all $e \in \Expr$, $\langle e \rangle_\partial$ is finite.
\end{lemma}
\begin{lemma}[{\cite[Proposition~7]{Sewell:1995:Algebra}}]\label{lem:congruence}
	$\sim$ is a congruence on $\Expr$ with respect to all operations of the algebra of regular behaviours.
\end{lemma}
\begin{lemma}[{\cite[Proposition~5.8]{Rutten:2000:Universal}}]\label{lem:quotient_chart}	
	We can equip ${\Expr}/{\sim}$ with a transition function given by
	$$
	\inferrule{e \tr{a}_\partial e'}{[e]_\sim \tr{a}_{\overline\partial} [e']_\sim} \qquad \inferrule{e \rhd_{\partial} v_i}{[e]_\sim \rhd_{\overline{\partial}} v_i}
	$$
	This map is a unique transition function on ${\Expr}/{\sim}$ that makes the quotient map $[-]_{\sim} \colon \Expr \to {\Expr}/{\sim}$ into prechart homomorphism.
\end{lemma}
From now on, we will overload the notation and simply write $e$ for the equivalence class $[e]_{\sim}$. Note that because of \Cref{lem:congruence} all operations of algebra of regular behaviours are well defined on that quotient. 
\begin{remark}
	The quotient prechart on ${\Expr}/{\sim}$ and corresponding operations of algebra of regular behaviours are in one-to-one correspondence with the prechart structure on $\Omega$ and operations defined in \Cref{sec:operations} in the main body of the paper. 
\end{remark}
\begin{remark}[{\cite{Sewell:1995:Algebra}}]
	The last rule (that defines the transition behaviour of the $\mu$ recursion operator) can be replaced by the following:
	\begin{gather*}
		\inferrule{e[\mu v. e / v] \tr{a} e'}{\mu v. e \tr{a} e'}
	\end{gather*}
\end{remark}

\begin{remark}[{\cite[Proposition~5.4.]{Milner:1984:Complete}}]\label{rem:semantic-substitution}
	Syntactic substitution can be described operationally using the following rules
	\begin{gather*}
		\inferrule{e \rhd v\quad f\tr{a} f'}{e[f / v] \tr{a} f'}\qquad \inferrule{e \tr{a} e'}{e[f / v] \tr{a} e'[f/v]}\\ \inferrule{e \rhd w \quad w \neq v }{e[f/v] \rhd w} \qquad \inferrule{e \rhd v \quad f \rhd w}{e[f/v] \rhd w}
	\end{gather*}
\end{remark}
\begin{lemma}\label{lem:subst_lemma}
	For all $e, f_1, \dots, f_m, g_1, \dots, g_m \in \Expr$ and vectors $\vec{v}=(v_{i_1}, \dots, v_{i_m}), \vec{w}=(v_{j_1}, \dots, v_{j_n})$, such that all free variables of $e$ are contained in $\vec{v}$ and all free variables of $\vec{f}$ are contained in $\vec{w}$, we have that 
	$$
	(e[\vec{f}/ \vec{v}])[\vec{g}/\vec{w}] =e [(f_1[\vec{g}/\vec{w}], \dots, f_m[\vec{g}/\vec{w}])/\vec{v}] 
	$$
\end{lemma}
\begin{proof}
	Let $\Delta = \{(e,e) \in \Expr\}$ be the diagonal relation. We define a relation ${R} \subseteq {\Expr \times \Expr}$, given by the following:
	\begin{align*}
		R = \Delta \cup \{&((e[\vec{f}/ \vec{v}])[\vec{g}/\vec{w}], e [(f_1[\vec{g}/\vec{w}], \dots, f_m[\vec{g}/\vec{w}])/\vec{v}])\\ &\mid e, f_1, \dots, f_m, g_1, \dots, g_n \in \Expr, \vec{v}=(v_{i_1}, \dots, v_{i_m}), \vec{w}=(v_{j_1}, \dots, v_{j_n}),\\ &\fv(e) \subseteq 		\vec{v},  \fv(\vec{f}) \subseteq \vec{w}\}
	\end{align*}

	 We claim that $R$ is a bisimulation. For pairs $(e,e) \in \Delta$, the conditions of bisimulation are immediately satisfied. 
	
	For the remaining pairs, assume that $(e[\vec{f}/ \vec{v}])[\vec{g}/\vec{w}] \rhd u$. In such a case at least one of the following is true:
	\begin{itemize}
		\item $e[\vec{f}/ \vec{v}] \rhd u$
		\item $e[\vec{f}/\vec{v}] \rhd w_j$ for $w_j \in \vec{w}$ and $g_j \rhd u$
	\end{itemize}
	Since all free variables of $e$ are contained in $\vec{v}$ and all free variables of $\vec{f}$ are contained in $\vec{w}$, we have that all free variables of $e[\vec{f}/\vec{v}]$ are also contained in $\vec{w}$, which makes the first case impossible. Through a similar line of reasoning, we can deduce that $e \rhd v_i$ for some $v_i \in \vec{v}$ and $f_i \rhd w_j$. Since $g_j \rhd u$, we have that $f_i[\vec{g}/\vec{w}] \rhd u$. Finally, we have that $e[(f_1[\vec{g}/\vec{w}], \dots, f_n[\vec{g}/\vec{w}]), \vec{v}] \rhd u$.
	
	Assume that $(e[\vec{f}/ \vec{v}])[\vec{g}/\vec{w}] \tr{a} h$. Then, at least one of the following is true:
	\begin{itemize}
		\item $e[\vec{f}/\vec{v}] \rhd w_j$ and $g_j \tr{a} h$.
		\item $h = h'[\vec{g}/\vec{w}]$, such that $e[\vec{f}/\vec{v}] \tr{a} h'$ 
	\end{itemize}
	In the first case, through a similar line of reasoning as before, we can conclude that $e \rhd{ v_i}$ for some $v_i \in \vec{v}$ and $f_i \rhd w_j$. Hence, $f_i [\vec{g}/\vec{w}] \tr{a} h $. Finally, we can deduce that $e[(f_1[\vec{g}/\vec{w}], \dots, f_m[\vec{g}/\vec{w}])/\vec{v}] \tr{a} h$. Obviously, $(h,h) \in R$.
	
	In the second case, we have that $e[\vec{f}/\vec{v}] \tr{a} h'$. There are two subcases, that need to be considered
	\begin{itemize}
		\item $e \rhd v_i$ and $f_i \tr{a} h'$
		\item $h' = h''[\vec{f}/ \vec{w}]$ and $e \tr{a} h''$ 
	\end{itemize}
	In the first subcase, we have that $f_i[\vec{g}/ \vec{w}] \tr{a} h'[\vec{g}/\vec{w}]$ and hence $$e[(f_1[\vec{g}/ \vec{w}], \dots, f_m[\vec{g}/ \vec{w}])/\vec{v}] \tr{a} h'[\vec{g}/\vec{w}]$$ or equivalently $e[(f_1[\vec{g}/ \vec{w}], \dots, f_m[\vec{g}/ \vec{w}])/\vec{v}] \tr{a} h$. As before, of course $(h,h) \in R$. 
	
	Finally, moving on to the second subcase, we have that $e \tr{a} h''$ and hence $$e[(f_1[\vec{g}/ \vec{w}], \dots, f_m[\vec{g}/ \vec{w}])/\vec{v}] \tr{a} h''[(f_1[\vec{g}/ \vec{w}], \dots, f_m[\vec{g}/ \vec{w}])/\vec{v}]$$ Recall that $(e[\vec{f}/ \vec{v}])[\vec{g}/\vec{w}] \tr{a} h$ and $h = (h''[\vec{f}/\vec{v}])[\vec{g}/\vec{v}]$. Both of those reachable expressions are actually in the relation $R$. The remaining conditions of bisimulation, can be shown via a symmetric argument. 
\end{proof}
\subsection{Behavioural distance of precharts}\label{apx:dist}
\behdist*
\begin{proof}
	$\mathcal{H}$ and $(-)^\uparrow$ are liftings for the functors $\Pfin$ and $\Sigma \times (-) + V$ respectively, that preserve isometries~\cite[Theorem~5.8]{Baldan:2018:Coalgebraic}. The rest follows from~\cite[Theorem~5.2]{Baldan:2018:Coalgebraic}.
\end{proof}
\begin{remark}[{\cite[Example~5.31]{Baldan:2018:Coalgebraic}}]\label{rem:hausdorff_duality}
Let $(X, d)$ be a pseudometric space and let $A, B \in \Pfin(X)$. Let $\Gamma(A,B)$ denote the set of relational couplings of $A$ and $B$, namely elements $R \in \Pfin(A \times B)$, such that $\pi_1 (R)=A$ and $\pi_2(R)=B$. The Hausdorff distance between $A$ and $B$ can be alternatively presented as:
$$
	\mathcal{H}(d)(A,B) = \inf \left\{\sup_{(x,y) \in R} d(x,y) \mid R \in \Gamma(A,B) \right\}
$$
\end{remark}
We will make use of the fact that the set of pseudometrics can be equipped with a norm. We write $\overline{\R} = [-\infty, \infty]$ for the set of extend reals. For any set $X$, the set of functions $\overline{\R}^{X\times X}$, which is the superset of $D_X$, can be seen as a Banach space by means of the sup norm $\|d\| = \sup_{x,y \in X} |d(x,y)|$.
	
\begin{lemma}[{\cite{Breugel:2012:Closure}}]
	Hausdorff lifting $\mathcal{H} \colon D_X \to D_{\Pfin(X)}$ is nonexpansive with respect to the metric induced by the sup norm. Namely,
	$$\|\mathcal{H}(d) - \mathcal{H}(d')\| \leq \|d - d' \|$$ 
\end{lemma}
\begin{lemma}
	$(-)^\uparrow \colon D_X \to D_{\Sigma \times X + V}$, the lifitng for $\Sigma \times (-) + V$ is contractive with respect to the metric induced by the sup norm. Namely,
	$\|d^\uparrow - d'^\uparrow \| \leq \frac{1}{2} \|d - d' \|$ 
\end{lemma}
\begin{proof}
	For the sake of simplicity, assume that $d' \sqsubseteq d$, and hence $d'^\uparrow \sqsubseteq d^\uparrow$. It suffices that we show that for all $u,w \in \Sigma \times X + V$, we have that $ d^\uparrow(u,w) - d'^\uparrow(u,w) \leq \frac{1}{2}\|d-d'\|$. Recall that in all cases, except when $u=(a,x)$ and $w=(a,y)$ for some $a \in \Sigma$ and $x,y \in X$, $d^\uparrow(u,w)=d'^\uparrow(u,w)$ and hence  $d^\uparrow(u,w)-d'^\uparrow(u,w) = 0 \leq \frac{1}{2} \|d - d'\|$. In the remaining case, we have that
	\begin{align*}
		d^\uparrow((a,x),(a,y)) - d'^\uparrow((a,x),(a,y)) &= \frac{1}{2}d(x,y) - \frac{1}{2}d'(x,y) \\
		&\leq \frac{1}{2} \|d'-d\|
	\end{align*}
	which completes the proof.
\end{proof}
\begin{lemma}
	$\Phi_\beta \colon D_X \to D_X$ is contractive with respect to the metric induced by the sup norm, namely
	$$
	\|\Phi_\beta(d)-\Phi_\beta(d')\| \leq \frac{1}{2}\|d-d'\|
	$$
\end{lemma}
\begin{proof}
	For the sake of simplicity, assume that $d' \sqsubseteq d$ and hence $\Phi_\beta(d') \sqsubseteq \Phi_{\beta}(d)$. It suffices to show that for all $x,y \in X$, we have that $\mathcal{H}(d^\uparrow)(\beta(x), \beta(y)) -  \mathcal{H}(d'^\uparrow)(\beta(x), \beta(y))  \leq \frac{1}{2}\|d-d'\|$. We can combine the previous results and for arbitrary $x,y \in X$ obtain the following
	\begin{align*}
		\mathcal{H}(d^\uparrow)(\beta(x), \beta(y)) -  \mathcal{H}(d'^\uparrow)(\beta(x), \beta(y)) & \leq \|\mathcal{H}(d^\uparrow) - \mathcal{H}(d'^\uparrow)\| \\
		&\leq \|d^\uparrow - d'^\uparrow\|\\
		& \leq \frac{1}{2} \|d - d'\|
	\end{align*}
\end{proof}
As a consequence, we have the following
\begin{corollary}
	$\Phi_\beta$ has a unique fixpoint.
\end{corollary}
	We will argue that for (locally) finite precharts  we can give a simpler characterisation of the behavioural distance relying on Kleene's theorem for the greatest fixpoint.
	Recall that $\omega$-cochain is a sequence $\{d_i\}_{i \in \N}$, such that for all $i \in \N$, we have that $d_i \sqsupseteq d_{i+1}$. 
\begin{theorem}[Kleene fixpoint theorem]\label{thm:kleene}
	Let $(L, \sqsubseteq)$ be a partial order where every $\omega$-cochain has an infimum and let $f \colon L \to L$ be an endomap that is cocontinuous, namely for all $\omega$-cochains $\{l_i\}_{i \in \N}$, we have that $f\left(\inf_{i \in \N} l_i \right) = \inf_{i \in \N} f(l_i)$. The greatest fixpoint of $f$ can be characterised as
	$$
	\mathsf{gfp}(f) = \inf_{i \in \N} f^{(i)}
	$$
	where $f^{(0)} = \top$ and $f^{(n+1)} = f^{(n)}(f)$
\end{theorem}
\begin{lemma}\label{lem:cocontinuous}
	For a finite prechart $(X,\beta)$, $\Phi_\beta$ is cocontinuous.
\end{lemma}
\begin{proof}
	Identical proof to~\cite[Lemma~20]{Rozowski:2024:Complete}.
\end{proof}
We can combine the above results into the following statement.
\finitedist*
\begin{proof}
	Since $\Phi_\beta$ has a unique fixpoint, we can rely on a Kleene fixpoint theorem for the greatest fixpoint (\Cref{thm:kleene}), which is preconditioned on $\Phi_\beta$ being cocontinuous, which is true for finite precharts (\Cref{lem:cocontinuous}). The final formula is given by the fact that in the lattice of pseudometrics, the infima of $\omega$-cochains can be calculated pointwise~\cite[Lemma~6]{Rozowski:2024:Complete}.
\end{proof}
Moreover, this can be extended beyond finite charts to the locally finite ones.
\begin{corollary}\label{cor:kleene_locally_finite}
	For any locally finite prechart $(X,\beta)$, the distance between $x, y \in X$, can be calculated by:
	\[
		\mathsf{bd}_\beta(x,y) = \inf_{i \in \N} \left(\Phi_\beta^{(i)}(x,y)\right)
	\]
\end{corollary}
\begin{proof}
Let $\beta'$ denote $\beta$ restricted to $\langle x , y \rangle_\beta$.
	Since $(X, \beta)$ is locally finite, then its subprechart $(\langle x , y \rangle_\beta, \beta')$ is finite. Since homomorphisms are isometries, calculating distance between $x$ and $y$ in $(X, \beta)$ is the same as calculating it in $(\langle x , y \rangle_\beta, \beta)$. Because of \Cref{lem:cocontinuous}, $\Phi_\beta$ is cocontinuous (when restricted to $\langle x , y \rangle_{\beta'}$) and hence we can employ \Cref{thm:kleene}. Since the infima in the lattice of pseudometrics can be calculated pointwise~\cite[Lemma~6]{Rozowski:2024:Complete}, we have that
	$$
	\mathsf{bd}_\beta(x,y) = \mathsf{bd}_{\beta'}(x,y) = \inf_{i \in I} \left( \Phi^{(i)}_{\beta'}(x,y)\right)
	$$
	Since $\beta'$ is a restriction of $\beta$ to ${\langle x , y \rangle_\beta}$ and each $\Phi^{(i)}_{\beta'}$ makes only use of the states in $\langle x , y \rangle_\beta$, we can rewrite the above as $
	\mathsf{bd}_\beta(x,y) = \inf_{i \in \N} \left(\Phi_\beta^{(i)}(x,y)\right)
	$ as desired.
\end{proof} 

\begin{lemma}\label{lem:distance_power_of_two}
	Let $(X,\beta)$ be a locally finite prechart. For all $x,y \in X$, $i \in \N$, either $\Phi_\beta^{(i)}=0$ or there exists $k \in \N$, such that $\Phi_{\beta}^{(i)}(x,y)=2^{-k}$
\end{lemma}
\begin{proof}
	Let $x,y \in X$. We proceed by induction on $i$. 
	
	\begin{itemize}
		\item If $i = 0$, then $\Phi^{(0)}_\beta(x,y)=1=2^{0}$, if $x\not=0$ or $\Phi^{(0)}_\beta(x,y)=0$, otherwise. 
		\item If $i = j + 1$, then unrolling the definition of $\Phi_\beta^{(j+1)}$ yields the following:
			$$\Phi^{j + 1}_\beta(x,y) = \max \left\{\sup_{u \in \beta(x)} \inf_{w \in \beta(y)} {\Phi^{(j)}_\beta}^\uparrow(u,w), \sup_{w \in \beta(y)} \inf_{u \in \beta(x)}{\Phi^{(j)}_\beta}^\uparrow(w,u)\right\}$$
			For any two transitions $u = (a,x')$ and $w=(a,y')$ with the same prefix, the following holds:
			$${\Phi^{(j)}_\beta}^\uparrow(u,w)=\frac{1}{2} {\Phi^{(j)}_\beta}(x',y')$$
			We can apply the induction hypothesis, which states that one of the following is true:
			\begin{itemize}
				\item${\Phi^{(j)}_\beta}(x',y')=0$, which entails that ${\Phi^{(j)}_\beta}^\uparrow(u,w)=0$.
				\item There exists a $k \in N$, such that ${\Phi^{(j)}_\beta}(x',y')=2^{-k}$. This implies that ${\Phi^{(j)}_\beta}^\uparrow(u,w)=2^{-(k+1)}$.
			\end{itemize}
			Since the infima range over finite sets, their values are either $1=2^0$ if the sets are empty or are one of the values from the set, which we have shown to be of the desired form. Similarly, suprema range over finite sets and are either $0$ for empty sets or are on of the values from the set, which are in the desired form. Taking the maximum of values in the desired form, still results in a value in the desired form.\qedhere
	\end{itemize}
\end{proof}

\begin{lemma}\label{lem:chain_stabilises}
	Let $(X, \beta)$ be a locally finite prechart and let $x,y \in X$, such that $\neg (x \sim y)$. There exists $i \in \N$, such that $\Phi_\beta^{(i)}(x,y)=\Phi^{(i+1)}_\beta(x,y)$
\end{lemma}
\begin{proof}
	Assume that for all $i \in \N$, $\Phi_\beta^{(i)}(x,y)\neq\Phi^{(i+1)}_\beta(x,y)$. Essentially, that means we have an infinite, strictly decreasing $\omega$-cochain $\{\Phi^{(i)}_\beta(x,y)\}_{i \in I}$. By \Cref{lem:distance_power_of_two}, we know that the values of the chain are either $0$ or $2^{-k}$. Since all the values of the chain are nonegative, if any of it is equal to zero $0$, we reach a contradiction, as the chain would have to contain values strictly below $0$. Hence, we can safely assume that the chain is in the form $\{1,\frac{1}{2}, \frac{1}{4}, \dots\}$. But in such a case its infimum is $0$, contradicting the assumption.
\end{proof}
\begin{corollary}\label{lem:finite_witness_for_positive_distance}
	Let $(X, \beta)$ be a locally finite prechart and let $x,y \in X$. If $\mathsf{bd}_\beta(x,y) > 0$, then there exists an $i \in \N$, such that $\mathsf{bd}_\beta(x,y) = \Phi^{(i)}_\beta(x,y)$
\end{corollary}
\begin{proof}
	If $\mathsf{bd}_\beta(x,y) >0$, we have that $\neg (x \sim y)$ and using \Cref{lem:chain_stabilises}, we know that the chain of approximants stabilises and hence the infimum of the chain is equal to the point where it stabilises.
\end{proof}
\begin{lemma}\label{lem:bound_on_stratified_bisim}
	Let $(X, \beta)$ be a prechart. For any $x,y \in X$, we have that 
	$$
	x \sim^{(k)} y \iff \mathsf{bd}_{\beta}(x,y) \leq 2^{-k}
	$$
\end{lemma}
\begin{proof}
	By induction on $k$. The base case is trivial, as we immediately have that $x \sim^{(0)} y$ and $\mathsf{bd}_\beta(x,y) \leq 2^{0} = 1$.
	
	For the inductive step assume that for some $k' \in \N$, the induction hypothesis holds. First, assume that $x \sim^{(k'+1)} y$. Unrolling the definition of stratification of bisimilarity, we have that
	\begin{itemize}
		\item If $x \rhd v$, then $y \rhd v$
		\item If $x \tr{a} x'$, then there exists $y'$ such that $y \tr{a} y'$ and $x \sim^{(k')} y$.
	\end{itemize}
	and symmetrically.
	
	We want to show that $\mathsf{bd}_\beta(x,y) \leq 2^{-(k'+1)}$. We can unroll the definition of $\mathsf{bd}_\beta$ and rewrite the desired result as
	$$
	\sup_{u \in \beta(x)} \left( \inf_{w \in \beta(y)} \mathsf{bd}_\beta^\uparrow(u,w)\right)  \leq 2^{-(k'+1)}\quad \wedge \quad \sup_{w \in \beta(y)} \left( \inf_{u \in \beta(x)} \mathsf{bd}_\beta^\uparrow(w,u)  \right)\leq 2^{-(k'+1)}
	$$
	We focus on the left hand side of the conjunction above, as the right hand side is symmetric. We are aiming to show that
	$$
	\forall {u \in \beta(x)}\ldotp~\inf_{w \in \beta(y)} \mathsf{bd}_\beta^\uparrow(u,w)  \leq 2^{-(k'+1)} 
	$$
	Let $u \in V$. By the assumption, we know that also $u \in B(y)$, which means that $$\inf_{w \in \beta(y)} \mathsf{bd}_\beta^\uparrow(u,w) = 0 \leq 2^{-(k'+1)}$$
	
	Now, let $u \in \Sigma \times X$, i.e. $u = (a,x')$. By the assumption, we know that there exists $y' \in X$, such that $(a,y') \in \beta(y)$ and $x' \sim^{(k')} y'$. By induction hypothesis, we know that $\mathsf{bd}_{\beta}(x',y') \leq 2^{-k'}$. Hence, we have that $\mathsf{bd}_{\beta}^\uparrow((a,x'),(a,x')) \leq \frac{1}{2} \cdot  2^{-k'} = 2^{-(k'+1)}$. Hence, we again have that 
		$$
		\forall {u \in \beta(x)}\ldotp~\inf_{w \in \beta(y)} \mathsf{bd}_\beta^\uparrow(u,w)  \leq 2^{-(k'+1)} 
		$$
		
		Now, for the converse assume that $\mathsf{bd}_\beta(x,y) \leq 2^{-(k'+1)}$. Through a similar line of reasoning, as before, we have that
		$$
		\forall{u \in \beta(x)}\ldotp~\inf_{w \in \beta(y)} \mathsf{bd}_\beta^\uparrow(u,w) \leq 2^{-(k'+1)} 
		$$
		Assume that $x \rhd v$, i.e. $v \in \beta(x)$. Assume that $\neg (y \rhd v)$. That means that for all $w \in \beta(y)$, we have that $\mathsf{bd}_\beta^\uparrow(v,w) = 1$ and hence $\inf_{w \in \beta(y)} \mathsf{bd}_\beta^\uparrow(v,w)=1$, which contradicts the assumption as $1>2^{-{(k'+1)}}$. Hence, $y \rhd v$.
		
		Now, assume that $x \tr{a} x'$, i.e. $(a,x') \in \beta(y)$. Through a similar argument as before, we know that there must exist $(a,y')\in \beta(y)$, such that $\mathsf{bd}_\beta^\uparrow((a,x'),(a,y')) \leq 2^{-(k'+1)}$. Unrolling the definitions of $\mathsf{bd}_\beta^\uparrow$, we obtain $\frac{1}{2} \cdot \mathsf{bd}_{\beta}(x',y') \leq 2^{-(k'+1)}$ and hence $\mathsf{bd}_{\beta}(x',y') \leq 2^{-k'}$. Using the induction hypothesis, we get that $x' \sim^{(k')} y'$ as desired. The remaining part of the proof is symmetric and hence is omitted.
\end{proof}
\begin{theorem}\label{thm:concrete_distance}
Let $(X,\beta)$ be a locally finite prechart and let $x,y \in X$. The behavioural distance between $x$ and $y$ is given by:
	$$\mathsf{bd}_\beta(x,y) = \begin{cases}
		0 & \text{ if } x \sim y\\
		2^{-n} & \text{ if } n \text{ is the largest number such that } x\sim^{(n)} y
	\end{cases}$$
\end{theorem}
\begin{proof}
	For the first case, because of \Cref{thm:beh_dist}, if $x \sim y$, then $\mathsf{bd}_{\beta}(x,y)=0$. For the converse, if $\mathsf{bd}_{\beta}(x,y)=0$, we have that $x \sim^{(k)} y$ for all $k \in \N$ and hence by~\cite[Theorem~2.1]{hennessy:1985:algebraic}, it holds that $x \sim y$.
	
	 In the second case, because of \Cref{lem:chain_stabilises}, we know that if $\mathsf{bd}_\beta(x,y) > 0$, then the behavioural distance is equal to some element of the chain of approximants. By \Cref{lem:distance_power_of_two}, we know that all non-zero elements of that chain are equal to $2^{-k}$, for some $k \in \N$. Combining it with \Cref{lem:bound_on_stratified_bisim} yields the desired result. 
	 For the converse, if $n \in N$ is the largest number such that $x \sim^{(n)} y$, then because of \Cref{lem:bound_on_stratified_bisim}, we have that $\mathsf{bd}_\beta(x,y)\leq 2^{-n}$. Assume that $\mathsf{bd}_\beta(x,y)=0$. In such a case, using \Cref{lem:bound_on_stratified_bisim}, we could conclude that $x \sim^{(n+1)} y$, which would lead to contradiction. Hence, $\mathsf{bd}_\beta(x,y)>0$. Because of \Cref{lem:chain_stabilises}, we have that $\mathsf{bd}_\beta(x,y)$ is equal to some power of two. Combining that with \Cref{lem:bound_on_stratified_bisim} again yields the desired result.\qedhere
\end{proof}
\begin{lemma}\label{lem:hom_stratification}
	Let $(X, \beta)$, $(Y, \gamma)$ and $(Z, \zeta)$ be arbitrary precharts, such that there exist homomorphisms $f \colon X \to Z$ and $g \colon Y \to Z$.
	For all $p \in \N$, $(x,y) \in X \times Y$, we have that if $x \sim^{(p)} y$ then $f(x) \sim^{(p)} g(y)$.
\end{lemma}
\begin{proof}
	Recall that if $f \colon X \to Y$ is a homomorphism, then $G(f) = \{(x,f(x)) \mid x \in X\}$ is a bisimulation, which implies that for all $p \in \N$, $x \in X$ we have that $x \sim^{(p)} f(x)$~\cite[Theorem~2.1]{hennessy:1985:algebraic}. Similarly, for all $p \in \N$, $y \in Y$, we have that $y \sim^{(p)} g(y)$. Since relations constituting the stratification of bisimilarity are all equivalence relations~\cite{hennessy:1985:algebraic}, we know that if $x \sim^{(p)} y$, $x \sim^{(p)} f(x)$ and $y \sim^{(p)} g(y)$ jointly imply that $f(x) \sim^{(p)} g(y)$ for all $p \in \N$ and $(x,y) \in X \times Y$, which completes the proof.
\end{proof}
\concrete*
\begin{proof}
	Immediately follows from \Cref{thm:concrete_distance} and \Cref{lem:hom_stratification}.
\end{proof}

\section{Axioms of SMCs}
\label{app:smc-axioms}
\begin{equation*}
\begin{gathered}
{
\tikzset{x=1em, y=2.1ex}
\InputIfFileExists{smc/sequential-associativity.tikz}{}{\input{./tikz/smc/sequential-associativity.tikz}}
\tikzset{x=1em, y=1.5ex}
 \SMCeq 
\tikzset{x=1em, y=2.1ex}
\InputIfFileExists{smc/sequential-associativity-1.tikz}{}{\input{./tikz/smc/sequential-associativity-1.tikz}}
\tikzset{x=1em, y=1.5ex}
} \qquad {
\tikzset{x=1em, y=2.1ex}
\InputIfFileExists{smc/parallel-associativity.tikz}{}{\input{./tikz/smc/parallel-associativity.tikz}}
\tikzset{x=1em, y=1.5ex}
 \SMCeq 
\tikzset{x=1em, y=2.1ex}
\InputIfFileExists{smc/parallel-associativity-1.tikz}{}{\input{./tikz/smc/parallel-associativity-1.tikz}}
\tikzset{x=1em, y=1.5ex}
}\qquad  {
\tikzset{x=1em, y=2.1ex}
\InputIfFileExists{smc/interchange-law.tikz}{}{\input{./tikz/smc/interchange-law.tikz}}
\tikzset{x=1em, y=1.5ex}
\SMCeq
\tikzset{x=1em, y=2.1ex}
\InputIfFileExists{smc/interchange-law-1.tikz}{}{\input{./tikz/smc/interchange-law-1.tikz}}
\tikzset{x=1em, y=1.5ex}
 }
 \\
{
\tikzset{x=1em, y=2.1ex}
\InputIfFileExists{smc/unit-right.tikz}{}{\input{./tikz/smc/unit-right.tikz}}
\tikzset{x=1em, y=1.5ex}
 \SMCeq \diagbox{c}{}{} \SMCeq 
\tikzset{x=1em, y=2.1ex}
\InputIfFileExists{smc/unit-left.tikz}{}{\input{./tikz/smc/unit-left.tikz}}
\tikzset{x=1em, y=1.5ex}
}
\qquad\quad
{ 
\tikzset{x=1em, y=2.1ex}
\InputIfFileExists{smc/parallel-unit-above.tikz}{}{\input{./tikz/smc/parallel-unit-above.tikz}}
\tikzset{x=1em, y=1.5ex}
 \SMCeq \diagbox{c}{}{} \SMCeq  
\tikzset{x=1em, y=2.1ex}
\InputIfFileExists{smc/parallel-unit-below.tikz}{}{\input{./tikz/smc/parallel-unit-below.tikz}}
\tikzset{x=1em, y=1.5ex}
}
\\
{
\tikzset{x=1em, y=2.1ex}
\InputIfFileExists{smc/sym-natural.tikz}{}{\input{./tikz/smc/sym-natural.tikz}}
\tikzset{x=1em, y=1.5ex}
 \SMCeq 
\tikzset{x=1em, y=2.1ex}
\InputIfFileExists{smc/sym-natural-1.tikz}{}{\input{./tikz/smc/sym-natural-1.tikz}}
\tikzset{x=1em, y=1.5ex}
}
\qquad
{
\tikzset{x=1em, y=2.1ex}
\InputIfFileExists{smc/sym-iso.tikz}{}{\input{./tikz/smc/sym-iso.tikz}}
\tikzset{x=1em, y=1.5ex}
 \SMCeq 
\tikzset{x=1em, y=2.1ex}
\begin{tikzpicture}
	\begin{pgfonlayer}{nodelayer}
		\node [style=none] (0) at (2, -0.75) {};
		\node [style=none] (1) at (-2, -0.75) {};
		\node [style=none] (2) at (-2, 0.5) {};
		\node [style=none] (3) at (2, 0.5) {};
	\end{pgfonlayer}
	\begin{pgfonlayer}{edgelayer}
		\draw (0.center) to (1.center);
		\draw (3.center) to (2.center);
	\end{pgfonlayer}
\end{tikzpicture}}
\tikzset{x=1em, y=1.5ex}
}
\end{gathered}
\end{equation*}

\section{Semantics}\label{apx:semantics}
\begin{restatable}{lemma}{regbehcategory}\label{lem:comp_associative}
	Let $f \colon m \to n$, $g \colon n \to p$, $h \colon p \to q$. We have that: 
	\begin{enumerate}
		\item $(f;g);h = f;(g;h)$
		\item $\id_m ; f = f$
		\item $f ; \id_n = f$
	\end{enumerate}	
\end{restatable}
\begin{proof}
	We respectively prove each of the properties.
	\begin{enumerate}
		\item \begin{align*}
			(f ; g) ; h &= (f_1[\vec{g}/\vec{v}], \dots, f_m[\vec{g}/\vec{v}]);h \\
			&= \left((f_1[\vec{g}/\vec{v}])[\vec{h}, \vec{v}], \dots, (f_m[\vec{g}/\vec{v}])[\vec{h}, \vec{v}]\right)\\
			&= \left(f_1[(g_1[\vec{h}/\vec{v}], \dots, g_n[\vec{h}/\vec{v}])/\vec{v}],\dots, f_m[(g_1[\vec{h}/\vec{v}], \dots, g_n[\vec{h}/\vec{v}])/\vec{v}]\right) \tag{\Cref{lem:subst_lemma}} \\
			&= \left(f_1 [\vec{(g;h)}/\vec{v}], \dots f_m [\vec{(g;h)}/\vec{v}] \right) \\
			&= f ; (g;h)
		\end{align*}
		\item $id_m ; f = (v_1[\vec{f}/\vec{v}], \dots, v_m[\vec{f}/\vec{v}])=(f_1, \dots, f_m) = f$
		\item $f ; id_n = (f_1[\vec{v}/\vec{v}], \dots, f_m[\vec{v}/\vec{v}]]) = (f_1, \dots, f_m) = f$
	\end{enumerate}
\end{proof}
\begin{restatable}{lemma}{binarycoproducts}
	$\RegBeh$ has binary coproducts. In particular, given $k, l \in \N$, the inclusions $\inl_{k,l} \colon k \to k +l$ and $\inr_{k,l} \colon l \to k + l$ are given by $\inl_{k,l} = (v_1, \dots, v_k)$ and $\inr_{k,l} = (v_{k+1}, \dots, v_{k+l})$ respectively, while the mediating map is given by pairing. 
\end{restatable}
\begin{proof}
	Let $f \colon k \to m$ and $g \colon l \to m$. We can safely assume that $f = (f_1, \dots, f_k)$ and $g = (g_1, \dots, g_l)$. Recall that $\inl_{k,l} = (v_1, \dots, v_k)$ and $\inr_{k,l} = (v_{k+1}, \dots, v_{k+l})$. For the existence proof, define $\langle f , g \rangle \colon k + l \to m$ as a $k + l$-tuple $(f_1, \dots, f_k, g_1, \dots, g_l)$. We show that that the coproduct diagram commutes. We start from the left triangular subdiagram.
	
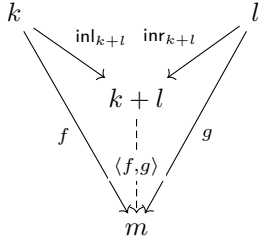
\[\begin{tikzcd}
	k && l \\
	& {k+l} \\
	\\
	& m
	\arrow["\inl_{k + l}", from=1-1, to=2-2]
	\arrow["f"', from=1-1, to=4-2]
	\arrow["\inr_{k + l}"', from=1-3, to=2-2]
	\arrow["g", from=1-3, to=4-2]
	\arrow["{\langle f , g \rangle}"{description}, dashed, from=2-2, to=4-2]
\end{tikzcd}\]
	\begin{align*}
		\inl_{k,l} ; \langle f , g\rangle &= (v_1, \dots, v_k) ; (f_1, \dots f_k, g_1, \dots, g_l) \\
		&= (f_1, \dots, f_k) \\
		&= f
	\end{align*}
	Similarly, for the right subdiagram, we have that:
	\begin{align*}
		\inr_{k,l} ; \langle f , g\rangle &= (v_{k +1}, \dots, v_{k + l}) ; (f_1, \dots f_k, g_1, \dots, g_l) \\
		&= (g_1, \dots, g_l) \\
		&= g
	\end{align*}
	For the uniqueness proof, assume that there exists a map $h \colon k + l \to m$, which makes the coproduct diagram commute. We can safely assume that $h = (h_1, \dots, h_{k + l})$. 
	Since $\inl_{k,l} ; h = f$ and $\inr_{k,l} ; h = g$, we have that:
	\begin{align*}
		(f_1, \dots, f_k) &= f\\ 
		&= \inl_{k,l} ; h\\
		&= (v_1, \dots, v_k) ; (h_1, \dots, h_{k + l}) \\
		&= (h_1, \dots, h_k)
	\end{align*}
	Similarly, we have that:
	\begin{align*}
		(g_1, \dots, g_l) &= f\\ 
		&= \inr_{k,l} ; h\\
		&= (v_{k+1}, \dots, v_{}) ; (h_1, \dots, h_{k + l}) \\
		&= (h_{k+1}, \dots, h_{k + l})
	\end{align*}
	Hence, $h = (f_1, \dots, f_k, g_1, \dots, g_l) = \langle f , g \rangle$ as desired. 
\end{proof}
\begin{lemma}
	$0$ is the initial object of $\RegBeh$.
\end{lemma}
\begin{proof}
	For any $n \in \N$, the unique universal arrow is given by $0_n$, which immediately completes the proof.
\end{proof}
Given $f \colon k \to l$ and $g \colon m \to n$, we can define their \emph{separated sum} $f \oplus g \colon k + m \to l + n $ given by the unique mediating arrow in the following diagram
\[\begin{tikzcd}
	k && {k + m} && m \\
	\\
	l && {l + n} && n
	\arrow["\inl_{k,m}", from=1-1, to=1-3]
	\arrow["f"', from=1-1, to=3-1]
	\arrow["{f \oplus g}"{description}, dashed, from=1-3, to=3-3]
	\arrow["\inr_{k,m}"', from=1-5, to=1-3]
	\arrow["g", from=1-5, to=3-5]
	\arrow["\inl_{l,n}"', from=3-1, to=3-3]
	\arrow["\inr_{l,n}", from=3-5, to=3-3]
\end{tikzcd}\]
We can define $f \oplus g$ concretely by setting $f \oplus g = \langle f, g[(v_{l + 1},\dots, v_{l + n})/(v_1, \dots, v_n)]\rangle$.
\begin{proposition}
	$(\RegBeh, \oplus , 0 )$ is a strict monoidal category.
\end{proposition}
\begin{proof}
	We verify that associators and unitors are strict equalities. Let $f \in \RegBeh(k,l), g \in \RegBeh(m,n), h \in \RegBeh(o,p)$. For the left unitor, we have the following:
	\begin{align*}
		0 \oplus f &= () \oplus (f_1, \dots, f_k) \tag{def. of $0$ and $\oplus$}\\
		&= (f_1, \dots, f_k) = f
	\end{align*}
	Similarly, for the right unitor, we have that:
	\begin{align*}
		f \oplus 0 &= () \oplus (f_1, \dots, f_k) \tag{def. of $0$ and $\oplus$}\\
		&= (f_1, \dots, f_k) = f
	\end{align*}
	Finally, for the associator we have that:
	\begin{align*}
		 &(f \oplus g) \oplus h\\
		=& (f_1, \dots, f_k, g_1[(v_{l+1}, \dots, v_{l + n})/\vec{v}], \dots, g_m[(v_{l+1}, \dots, v_{l + n})/\vec{v}]) \oplus h \\
		=& (f_1, \dots, f_k, g_1[(v_{l+1}, \dots, v_{l + n})/\vec{v}], \dots, g_m[(v_{l+1}, \dots, v_{l + n})/\vec{v}], \\
		&\qquad\qquad h_1 [(v_{l + n +1}, \dots, v_{l + n + p})/\vec{v}], \dots, h_o [(v_{l + n + 1}, \dots, v_{l + n + p})/\vec{v}])\\
		=& f \oplus  (g_1[(v_{1}, \dots, v_{n})/\vec{v}], \dots, g_m[(v_{1}, \dots, v_{n})/\vec{v}], \\
		&\qquad\qquad h_1 [(v_{n +1}, \dots, v_{n + p})/\vec{v}], \dots, h_o [(v_{n +1}, \dots, v_{n + p})/\vec{v}])\\
		=& f \oplus  (g_1, \dots, g_m, \\
		&\qquad\qquad h_1 [(v_{n +1}, \dots, v_{n + p})/\vec{v}], \dots, h_o [(v_{n +1}, \dots, v_{n + p})/\vec{v}])\\
		=& f \oplus  (g \oplus h)
	\end{align*}
	The intermediate steps in the calculation above follow from the definition of $\oplus$.
\end{proof}
\subsection{Conway Theories}
Let $\cat{C}$ be a category with finite coproducts, whose objects are natural numbers. We call $\cat{C}$ a \emph{preiteration theory} if for every morphism $f \colon n \to p + n$, there exists a morphism $f^\dagger_{n,p} \colon n \to p$ called \emph{dagger}. We will often omit the subscripts and write $f^\dagger$, when $n$ and $m$ are clear from the context. Note that the definition does not impose any conditions on the dagger. However, $f \colon 0 \to p$, when always we have that $f_{0,p}^\dagger = 0_p$.
\begin{definition}[{\cite[Definition~3.1]{Esik:1999:Group}}]\label{def:conway_theory}
	A Conway Theory is a preiteration theory, in which the following conditions are satisfied
	\begin{itemize}
		\item \textbf{(Scalar parameter identity)} $$(f ; (g \oplus \id_1))^\dagger = f^\dagger ; g$$ for all $f \colon 1 \to p + 1$, $g \colon p \to q$
		\item \textbf{(Scalar composition identity)} 
		$$
		(f ; \langle \id_p\oplus0_1, g \rangle)^\dagger = f ; \langle \id_p,  (g ;  \langle \id_p \oplus 0_1 , f\rangle)^\dagger\rangle
		$$ for all $f,g \colon 1 \to p + 1$
		\item \textbf{(Scalar double dagger identity)} 
		$$f^{\dagger\dagger} = (f ; (\id_p \oplus \nabla_1))^\dagger$$
		for all $f \colon 1 \to p + 2$
		\item \textbf{(Scalar pairing identity)}
		$$\langle f ,g \rangle^\dagger = \langle f^\dagger ; \langle \id_p, h^\dagger \rangle, h^\dagger \rangle$$ for all $f \colon n \to p + 1 + n $, $g \colon 1 \to p + 1 + n$ where 
		$$
		h = g ; \langle \id_{p + 1}, f^\dagger\rangle \colon 1 \to p + 1
		$$
	\end{itemize}
\end{definition}
\begin{remark}[{\cite[Remark~3.2]{Esik:1999:Group}}]\label{rem:defining_dagger}
	Note that in order to define a Conway theory it suffices to define $f^\dagger \colon 1 \to p$ for all $f \colon 1 \to p + 1$ that satisfies first three axioms of \Cref{def:conway_theory} and use \textbf{scalar pairing identity} to inductively define $(-)^\dagger$.
\end{remark}
Because of the above, we can define a dagger on $\RegBeh$ through the following:
\begin{definition}\label{def:Conway-operator}
Given $f \colon 1 \to p + 1$ in $\RegBeh$, we define 
$$
	f^\dagger = \mu v_{p+1}.f
$$
\end{definition}

\begin{lemma}[{\cite[Theorem~2]{Sewell:1995:Algebra}}]
\label{lem:recursion-substitution}
Terms of Milner's Algebra of Regular Behaviours (modulo bisimilarity) satisfy the following rules:
	\begin{enumerate}
		\item $\mu v_z. \left(e [(v_z, v_z) / (v_j, v_k)]\right) = \mu v_j. \mu v_k. e$ for any $v_z$ not free in $e$
		\item $ \mu v_j.\left(e[f/v_j]\right) = e[\mu v_x. \left(f[e/v_j]\right)/v_j]$
	\end{enumerate}
\end{lemma}

\begin{corollary}\label{conway1}
	The dagger on $\mathsf{RegBeg}$ satisfies \textbf{scalar composition} and \textbf{scalar double dagger} identities. 
\end{corollary}
\begin{proof}
	Follows from \Cref{lem:recursion-substitution}.
\end{proof}
\begin{lemma}
	Let $f \colon 1 \to 1 + p$, $g \colon p \to g$ be morphisms in $\RegBeh$. Then,
	$$
	(f ; (g\oplus \id_1))^\dagger = f^\dagger ; g
	$$
\end{lemma}
\begin{proof}
\begin{align*}
(f ; (g \oplus \id_1))^\dagger & = (f[(g_1, \dots, g_p, v_{q+1})/(v_1, \dots, v_{p}, v_{p+1})])^\dagger \tag{Def. of $\oplus$}\\
& = \mu v_{q+1}.\big(f[(g_1, \dots, g_p, v_{q+1})/(v_1, \dots, v_{p}, v_{p+1})]\big) \tag{Def. of $\dagger$}\\
& = \mu v_{q+1}.\big(f[v_{q+1}/v_{p+1}][\vec{g}/\vec{v}]\big) \tag{\cite[Lemma~5.6 2.]{Milner:1984:Complete}}\\
& = (\mu v_{q+1}.f[v_{q+1}/v_{p+1}])[\vec{g}/\vec{v}] \tag{\Cref{def:subset}}\\
& = (\mu v_{p+1}.f)[\vec{g}/\vec{v}] \tag{\cite[Proposition~4.6 5.]{Milner:1984:Complete}}\\
& = f^\dagger[\vec{g}/\vec{v}] \tag{Def. of $\dagger$}\\
& = f^\dagger ; g \tag*{\hfill\qedhere}
\end{align*}
\end{proof}
\begin{lemma}\label{conway2}
Let $f,g \colon 1 \to 1 + p$ be morphisms of $\RegBeh$. Then,
		$$
		(f ; \langle  \id_p \oplus 0_1 , g \rangle)^\dagger = f ; \langle \id_p, (g ;  \langle  \id_p \oplus 0_1 , f\rangle)^\dagger\rangle
		$$ 
\end{lemma}
\begin{proof}
\begin{align*}
(f ; \lc  \id_p \oplus 0_1 , g \rc)^\dagger & = \big(f[(v_1,\dots,v_p, g)/\vec{v}]\big)^\dagger	
\\
& = \mu v_{p+1}.\big(f[(v_1,\dots,v_p, g)/\vec{v}]\big)
\\
& = \mu v_{p+1}.\big(f[g/v_{p+1}]\big)
\\
& = f[\mu v_{p+1}. \left(g[f/v_{p+1}]\right)/v_{p+1}] \tag{\Cref{lem:recursion-substitution}~2.}
\\
& = f[\mu v_{p+1}. \left(g[(v_1,\dots,v_n,f)/(v_1,\dots,v_n,v_{p+1})]\right)/v_{p+1}]\\
& = f[\mu v_{p+1}. (g ;  \lc  \id_p \oplus 0_1 , f\rc)/v_{p+1}]
\\
& = f[(g ;  \lc  \id_p \oplus 0_1 , f\rc)^\dagger/v_{p+1}]
\\
& = f ; \lc \id_p, (g ;  \lc  \id_p \oplus 0_1 , f\rc)^\dagger\rc \qedhere
\end{align*}
\end{proof}
\noindent  And the \textbf{scalar double dagger identity}. 
\begin{lemma}\label{conway3}
Let $f \colon 1 \to 2 + p$ be a morphism of $\RegBeh$. Then,
		$$
		f^{\dagger\dagger} = (f ; (\id_p\oplus \nabla_1))^\dagger
		$$
\end{lemma}
\begin{proof}
\begin{align*}
f^{\dagger\dagger} &= \mu v_{p+1}.(\mu v_{p+2}.f)
\\
& = \mu v_{p+1}. \left(f [(v_{p+1}, v_{p+1}) / (v_{p+1}, v_{p+2})]\right) \tag{\Cref{lem:recursion-substitution}~ 1.}
\\
& = \mu v_{p+1}. \left(f;(\id_p\oplus \nabla_1)\right)
\\
& = (f ; (\id_p\oplus \nabla_1))^\dagger \tag*{\hfill\qedhere}
\end{align*}
\end{proof}
\begin{restatable}{lemma}{regbehconway}\label{lem:regbehconway}
$\RegBeh$ is a Conway Theory.	
\end{restatable}
\begin{proof}
	Follows from \Cref{conway1}, \Cref{conway2} and \Cref{conway3}.
\end{proof}
\subsection{Trace-fixpoint correspondence}
Turns out, that having a category $\cat{C}$ with finite coproducts and equipped with a dagger operator satisfying the axioms of Conway Theories is synonymous with $\cat{C}$ being traced symmetr	ic monoidal category. This is captured by the following theorem that was independently proved by Hasegawa~\cite{Hasegawa:1997:Models} and Haghverdi~\cite{Haghverdi:2000:Categorical}. The formulation of Hasegawa is phrased dually via the setting of products and cartesian categories.
\begin{theorem}[{\cite[Proposition~3.1.9]{Haghverdi:2000:Categorical}}]\label{thm:trace}
For any category with finite coproducts, to give a Conway operator is to give a trace (where finite coproducts are taken as the monoidal structure). 
\end{theorem}
That bijective correspondence is concretely given by the following:
\begin{gather*}
	\inferrule{f \colon  n \to p + n}{f^\dagger = \Tr^n_{n,p} (\nabla_n ; f) \colon n \to p} 
	\qquad
	\inferrule{g \colon p + n \to q + n}{\Tr^{n}_{p,q}(g) = \inl_{p,n} ; (g; \langle\id_q, \inr_{q +p,n} \rangle)^\dagger  \colon p \to q}
\end{gather*}

\subsection{Int construction}\label{sec:apx_int}
Given a traced symmetric monoidal category $\cat{C}$, we can construct a compact closed category $\Int{\cat{C}}$. The objects of $\Int{\cat{C}}$ are the pairs $(A^+,A^-)$ of objects of $\cat{C}$. Morphisms $f$ from $(A^+, A^-)$ to $(B^+, B^-)$ are the morphisms $f \colon A^+ \otimes B^- \to A^- \otimes B^+$ of $\cat{C}$. The identity of any object $(A^+, A^-)$ is given by the symmetry of $\cat{C}$, namely $\id_{(A^ +, A^-)} = \sigma_{A^+, A^-}$. The composition $f ; g \colon (A^+, A^-) \to (C^+, C^-)$ of morphisms $f \colon (A^+, A^-) \to (B^{+}, B^-)$ and $g \colon (B^+, B^-) \to (C^+,C^-)$ is defined as
	$$
		\Tr^{B^- \otimes B^+}_{A^+ \otimes C^-, A^- \otimes C^+} (\alpha ; (f \otimes g) ; \beta)
	$$
	in $\cat{C}$, where $\alpha = (\id_{A^+} \otimes \sigma_{C^-, B^-} \otimes \id_{B^+});(\id_{A^+} \otimes \id_{B^-} \otimes \sigma_{C^-, B^+})$ and $\beta = (\id_{A^-} \otimes \id_{B^+} \otimes \sigma_{B^-, C^+}); (\id_{A^-} \otimes \sigma_{B^+, C^+} \otimes \id_{B^-});(\id_{A^-} \otimes \id_{C^+} \otimes \sigma_{B^+, B^-})$.
	
	$\Int{\cat{C}}$ is equipped with the monoidal structure. The tensor product of $(A^+, A^-)$ and $(B^+, B^-)$ is given by taking the tensor product of $\cat{C}$ pointwise, namely $(A^+ \otimes B^+, A^- \otimes B^-)$. The unit of that monoidal product is given by $(I, I)$, where $I$ is the unit of the monoidal product on $\cat{C}$. The tensor product $f \otimes g \colon (A^+ \otimes C^+, B^- \otimes D^-) \to (A^- \otimes C^-, B^+ \otimes D^+)$ of $f \colon (A^+, A^-) \to (B^+, B^-)$ and $g \colon (C^+, C^-) \to (C^+, C^-) \to (D^+, D^-)$ is given by:
	$$
	f \otimes g = (\id_{A^+} \otimes \sigma_{C^+, B^-} \otimes \id_{D^-});(f \otimes g);(\id_{A^-} \otimes \sigma_{B^+, C^-} \otimes \id_{D^+})
	$$
	
	The dual $(A^+, A^-)^{\ast}$ of $(A^+, A^-)$ is given by exchanging the components, that is by $(A^-, A^+)$. Then, the unit $\eta_{(A^+, A^-)} \colon (I,I) \to (A^+, A^-) \otimes (A^+, A^-)^{\ast}$ is a morphism $\sigma_{A^-, A^+} \colon A^- \otimes A^+ \to A^+ \otimes A^-$. The counit $\epsilon_{(A^+, A^-)} \colon (A^+, A^-)^\ast \otimes (A^+, A^-) \to (I,I)$ can be similarly given by $\sigma_{A^-, A^+} \colon A^- \otimes A^+ \to A^+ \otimes A^-$ in $\cat{C}$.
	$\Int{\cat{C}}$ is equipped with a canonical trace, which takes a morphism $f \colon (A^+, A^-) \otimes (U^+, U-) \to (B^+, B^-) \otimes (U^+, U-)$ to
	$$
	\Tr^{(U^+, U^-)}_{(A^+, A^-), (B^+, B^-)} (f) = \left(\id_{(A^+, A-)} \otimes \eta_{(U^+, U-)}\right) ;\left(f \otimes \id_{{(U^+, U^-)}^\ast}\right); \left(\id_{(B^+, B^-)} \otimes \epsilon_{(U^+, U-)}  \right)
	$$

\subsection{Pseudometric structure on the semantic category}
\begin{lemma}\label{lem:subst-stratified-bisim}
	Let $i_1, \dots, i_m \in \N$. For all $n \in \N$ and for all $e,f,g_1, \dots, g_m, h_1, \dots, h_m \in {\Expr}/{\sim}$, we have that
	$$
	e \sim^{(n)} f \wedge \bigwedge_{j=1}^{j\leq m} g_j \sim^{(n)} h_j \implies e[(g_1, \dots, g_m)/ (v_{i_1}, \dots v_{i_m})] \sim^{(n)}  f[(h_1, \dots, h_m)/ (v_{i_1}, \dots v_{i_m})] 
	$$  
\end{lemma}
\begin{proof}

	Base case holds immediately, since we always have that $e[(g_1, \dots, g_m)/ (v_{i_1}, \dots v_{i_m})] \sim^{(0)}  f[(h_1, \dots, h_m)/ (v_{i_1}, \dots v_{i_m})] 
$ for all $e,f,g_1, \dots, g_j, h_1, \dots, h_j \in {\Expr}/{\sim}$
	
	Assume that $e \sim^{(n + 1)} f$, $g_j \sim^{(n + 1)} h_j$ for all $j \in \{1, \dots, m\}$. For the successor case assume that $e[(g_1, \dots, g_m)/ (v_{i_1}, \dots v_{i_m})] \sim^{(n)} f[(h_1, \dots, h_m)/ (v_{i_1}, \dots v_{i_m})]$. We will argue that $$e[(g_1, \dots, g_m)/ (v_{i_1}, \dots v_{i_m})] \sim^{(n + 1)}  f[(h_1, \dots, h_m)/ (v_{i_1}, \dots v_{i_m})]$$ To do so, we will study the operational semantics of both (equivalence classes of) expressions. 
	
		Assume that $e[(g_1, \dots, g_m)/ (v_{i_1}, \dots v_{i_m})] \rhd v_k$. Using \Cref{rem:semantic-substitution}, we can observe that it is the case if any of the following is true: 
		\begin{itemize}
			\item $e \rhd v_k$ and $v_k \notin \{v_{i_1}, \dots v_{i_{j}}\}$ 
			\item $e \rhd v_{i_l}$ for some $v_{i_l} \in \{v_{i_1}, \dots, v_{i_j}\}$ and $g_l \rhd {v_k}$
		\end{itemize}
		Consider the first subcase. Since $e \rhd v_k$ (for some $k \notin \{v_{i_1}, \dots, v_{i_j}\}$) and by assumption $e \sim^{n+1} f$, we have that $f \rhd {v_k}$ and hence $f[(h_1, \dots, h_j)/(v_{i_1}, \dots, v_{i_j})] \rhd v_k$. Now, consider the second subcase. By a similar line of reasoning, we can obtain $f \rhd v_l$ and $h_l \rhd v_k$. Hence, again we have that $f[(h_1, \dots, h_j)/(v_{i_1}, \dots, v_{i_j})] \rhd v_k$. In other words, we have shown that $e[(g_1, \dots, g_m)/ (v_{i_1}, \dots v_{i_m})] \rhd v_k$ implies $f[(h_1, \dots, h_j)/(v_{i_1}, \dots, v_{i_j})] \rhd v_k$. One can easily obtain a reverse implication through a symmetric proof.
		
		Now, assume that  $e[(g_1, \dots, g_m)/ (v_{i_1}, \dots v_{i_m})] \tr{a} s$. Using \Cref{rem:semantic-substitution}, we know that such transition can be made only if any of the following is true:
		\begin{itemize}
			\item $s = e'[(g_1, \dots, g_m)/ (v_{i_1}, \dots v_{i_m})]$, for some $e'$ such that $e \tr{a} e'$
			\item For some $v_{i_l}$, such that $v_{i_l} \in \{v_{i_1}, \dots, v_{i_m}\}$, we have that $e \rhd v_l$ and $g_l \tr{a} s$
		\end{itemize}
		Consider the first subcase. We know that $f \tr{a} f'$ and $e' \sim^{(n)} f'$. Using the induction hypothesis, we can conclude that 
		$$
			e[(g_1, \dots, g_m)/(v_{i_1}, \dots, v_{i_m})] \sim^{(n)} f[(h_1, \dots, h_m)/(v_{i_1}, \dots, v_{i_m})]
		$$
		Hence, there exists a $t$, such that $f[(h_1, \dots, h_m)/(v_{i_1}, \dots, v_{i_m})] \tr{a} t$, such that $s \sim^{(n)} t$.
		
		Now, consider the second subcase. We can easily conclude that $f \rhd v_l$ and $h_l \tr{a} t$, with $s \sim^{(n)} t$.
		
		In other words, we have show that $e[(g_1, \dots, g_m)/ (v_{i_1}, \dots, v_{i_m})] \tr{a} s$, then there exists $t$, such that $f[(h_1, \dots, h_m)/(v_{i_1}, \dots, v_{i_m})] \tr{a} t$, such that $s \sim^{(n)} t$. The reverse implication can be again shown via a symmetric argument. Combining all of the above, we can conclude that
		$$e[(g_1, \dots, g_m)/ (v_{i_1}, \dots, v_{i_m})] \sim^{(n+1)} f[(h_1, \dots, h_m)/(v_{i_1}, \dots, v_{i_m})] \tr{a} t$$
\end{proof}
\begin{corollary}\label{cor:seq_nexp}
	Let $e,f,g_1, \dots , g_m,h_1, \dots, g_m \in {\Expr}/{\sim}$. Then,
	$$\mathsf{bd}_{\overline{\partial}}\left(e[(g_1, \dots, g_m/(v_{i_1}, \dots v_{i_m})], f[(h_1, \dots, h_m/(v_{i_1}, \dots v_{i_m})]\right) \leq \max\{\mathsf{bd}_{\overline{\partial}}(e,f),  \max_{j \in \{1, \dots, m\}} \{\mathsf{bd}_{\overline{\partial}}(g_j,h_j)\}\}$$
\end{corollary}
\begin{proof}
If the right hand side of the inequality is equal to zero, then we have that $e \sim f$ and $g_j \sim h_j$ for $j \in \{1, \dots, m\}$. We can use \Cref{lem:congruence}, to conclude that $e[(g_1, \dots, g_m/(v_{i_1}, \dots v_{i_m})] \sim f[(h_1, \dots, h_m/(v_{i_1}, \dots v_{i_m})]$ and hence $\mathsf{bd}_{\overline{\partial}}(e[(g_1, \dots, g_m/(v_{i_1}, \dots v_{i_m})], f[(h_1, \dots, h_m/(v_{i_1}, \dots v_{i_m})]) = 0$, which implies nonexpansivity.
	
If the right hand side is greater than zero, we can employ the characterisation of $\mathsf{bd}_{\overline{\partial}}$ from \Cref{thm:concrete_distance} and use \Cref{lem:subst-stratified-bisim} to conclude the desired result. 
\end{proof}

\begin{lemma}\label{lem:rec_startified_bisim}
	Let $e, f \in {\Expr}/{\sim_{\partial}}$. Then, for all $n \in \N$, we have that
	$$
	e \sim^{(n)} f \implies \mu v_x. e \sim^{(n)} \mu v_x.f
	$$
\end{lemma}
\begin{proof}
	Base case holds immediately, as $\mu v_x.e \sim^{(0)} \mu v_x.f$ for all $e, f \in {\Expr}/{\sim}$. 
	
	Assume that $e \sim^{(n+1)} f$. For the successor case assume that $\mu v_x . e \sim^{(n)} \mu v_x. f$. We will argue that $\mu v_x . e \sim^{(n+1)} \mu v_x. f$.
	
	Assume that $\mu v_x.e \rhd v_k$. It is only the case, when $e \rhd {v_k}$ and $v_k \neq v_x$. Since $e \sim^{(n+1)} f$, then $f \rhd {v_k}$ and hence $\mu v_x . f \rhd v_k$. In other words, $\mu v_x.e \rhd v_k$ implies $\mu v_x. f \rhd v_k$. The reverse implication can be easily obtained via a symmetric proof.
	
	Now, assume that $\mu v_x.e \tr{a} s$. It is the case, when $e \tr{a} e'$ and $s = e'[\mu v_x.e / v_x]$. Since $e \sim^{(n+1)} f$, then there exists $f'$, such that $f \tr{a} f'$ and $e \sim^{(n)} f$. We can now use induction hypothesis and \Cref{lem:subst-stratified-bisim} and conclude that $e'[\mu v_x. e / v_x] \sim^{(n)} f' [\mu v. f / v_x]$. Moreover, we have that $\mu v.x f \tr{a} f'[\mu v_x.f/v_x]$. Hence, if $\mu v_x.e \tr{a} s$, then there exists $t$, such that $\mu v_x.f \tr{a} t$ and $s \sim^{(n)} t$. A reverse implication can be obtained via a symmetric proof. 
	
	Combing the above, allows us to conclude that $\mu v_x. e \sim^{(n+1)} \mu v_x . f$.
\end{proof}
\begin{corollary}\label{cor:rec_nexp}
	Let $e,f \in {\Expr}/{\sim}$. We have that
	$$
		\mathsf{bd}_{\overline\partial} (\mu v_x.e,\mu v_x.f)\leq \mathsf{bd}_{\overline\partial} (e,f)
	$$
\end{corollary}
\begin{proof}
	Analogous proof to \Cref{cor:seq_nexp}, but utilising \Cref{lem:rec_startified_bisim} instead.
\end{proof}
\begin{lemma}\label{lem:prefixing_step}
	Let $e,f \in {\Expr}/{\sim}$. Then for all $n \in \N$, $a \in \Sigma$, we have that
	$$
	e \sim^{(n)} f \implies a.e \sim^{(n+1)} a.f 
	$$
\end{lemma}
\begin{proof}
	$a.e$ does not output anything and so does $a.f$. Now, assume that $a.e \tr{a} e'$. Then, the only possibility is that $e'=e$. We can match that transition with an expression $a.f$ that performs an $a$-labelled transition to $f$. Since $e \sim^{(n)} f$, then $a.e \sim^{(n+1)} a.f$. The remaining condition works through a symmetric line of reasoning. 
\end{proof}
\begin{restatable}{corollary}{prefixingdiscounts}\label{cor:prefixing_discounts}
	Let $e,f \in {\Expr}/{\sim}$. Then for all $n \in \N$, $a \in \Sigma$, we have that
	$
		\mathsf{bd}_{\overline{\partial}}(a.e,a.f) \leq \frac{1}{2}\mathsf{bd}_{\overline{\partial}}(e,f)
	$
\end{restatable}
\begin{proof}
	We employ the characterisation from~\Cref{thm:concrete_distance}. If $e \sim f$, then by \Cref{lem:congruence} we are done. Otherwise, we have that $\mathsf{bd}_{\overline\partial}(e,f)=2^{-k}$ and $e \sim^{(k)}_{\overline\partial} f$ for some $k \in \N$. By applying~\Cref{cor:prefixing_discounts}, we get that $a.e \sim^{(k+1)} a.f$ and hence $\mathsf{bd}_{\overline\partial}(e,f)\leq 2^{-(k+1)} = \frac{1}{2} 2^{-k} = \frac{1}{2}\mathsf{bd}_{\overline\partial}(e,f)$, as desired. 
\end{proof}
\begin{restatable}{lemma}{seqnonexpansive}\label{lem:seq_nonexpansive}
	Sequential composition in $\RegBeh$ is nonexpansive.
\end{restatable}
\begin{proof}
	Let $f,h \in \RegBeh(n,m)$ and $g,i \in \RegBeh(m,k)$. 
	\begin{align*}
		d^{n,k}(f ;g, h;i) &= \sup_{1 \leq j \leq n} \left\{ \mathsf{bd}_{\overline\partial} \left( \langle f_j[(g_1, \dots, g_m)/(v_1, \dots, v_m)],   h_j[(i_1, \dots, i_m)/(v_1, \dots, v_m)] \right)\right\}\\
		&\leq \sup_{1 \leq j \leq n} \left\{\max \left\{\mathsf{bd}_{\overline\partial} (f_j,h_j), \sup_{ 1 \leq l \leq m} \{\mathsf{bd}_{\overline\partial} (g_l, i_l)\}\right\}\right\} \tag{\Cref{cor:seq_nexp}} \\
		&= \max \{ \sup_{1 \leq j \leq n} \{\mathsf{bd}_{\overline\partial} (f_j,h_j)\}, \sup_{ 1 \leq l \leq m} \{\mathsf{bd}_{\overline\partial} (g_l, i_l)\}\}\\
		&= \max \{d^{n,m}(f,h), d^{m,k}(g,i)\}
	\end{align*}
\end{proof}
\begin{lemma}\label{lem:pairs_nonexpansive}
	Let $f,f' \colon m \to k$, $g,g' \colon n \to k$ be morphisms of $\RegBeh$. We have that $d^{m +n, k}(\langle f, g \rangle, \langle f', g' \rangle) \leq \max \{d^{m,k}(f,f'), d^{n,k}(g,g')\}$
\end{lemma}
\begin{proof}
	\begin{align*}
		d^{m +n, k}(\langle f, g \rangle, \langle f', g' \rangle) &= d^{m +n, k}((f_1, \dots, f_m, g_1, \dots, g_n), (f'_1, \dots, f'_m, g'_1, \dots, g'_n)) \\
		&=\max \left\{d^{m,k}((f_1, \dots, f_m), (f'_1, \dots, f'_m)),d^{n,k}((g_1, \dots, g_m), (g'_1, \dots, g'_m)) \right\}\\
		&\leq \max \left\{d^{m,k}(f,f'),d^{n,k}(g,g') \right\}\\
	\end{align*}
\end{proof}
\begin{restatable}{lemma}{tensnonexpansive}\label{lem:tens_nonexpansive}
	The coproduct in $\RegBeh$ is nonexpansive.
\end{restatable}
\begin{proof}
	Let $f,h \in \RegBeh(k,m)$ and $g,i \in \RegBeh(l,n)$. Given $j \in \N$, such that $1 \leq j \leq l$, we define
	$$
	g'_j = g_j[(v_{m + 1}, \dots, v_{m + n})/(v_1, \dots, v_n)]
	$$ 
	Similarly, we write 
	$$
	i'_j = i_j[(v_{m + 1}, \dots, v_{m + n})/(v_1, \dots, v_n)]
	$$
	Using \Cref{cor:seq_nexp} one can easily obtain that for all $j \in \N$, such that $1 \leq j \leq l$, we have that
	\begin{align*}
		\mathsf{bd}_{\overline\partial} (g'_j, i'_j) \leq \mathsf{bd}_{\overline\partial} (g_j, i_j)
	\end{align*}
	Using that fact, we can prove the following
	\begin{align*}
		d^{k + l, m + n}(f \oplus g, h \oplus i) &= d^{k + l, m + n}(\langle f_1, \dots, f_k, g'_1, \dots, g'_l\rangle,\langle h_1, \dots, h_k, i'_1, \dots, i'_l\rangle)\\
		&=\max \left\{\sup_{1 \leq p \leq k}\{\mathsf{bd}_{\overline\partial} (f_p, h_p)\}, \sup_{1 \leq j \leq l} \{\mathsf{bd}_{\overline\partial} (g'_j, i'_j)\} \right\}\\
		&\leq \max \left\{\sup_{1 \leq p \leq k}\{\mathsf{bd}_{\overline\partial} (f_p, h_p)\}, \sup_{1 \leq j \leq l} \{\mathsf{bd}_{\overline\partial} (g_j, i_j)\} \right\}\\
		&= \max \{d^{k,m}(f,g), d^{l,n}(h,i)\}
	\end{align*}
\end{proof}

\begin{restatable}{lemma}{daggernexp}\label{lem:dagger_nonexpansive}
	The dagger on $\RegBeh$ is nonexpansive.
\end{restatable}
\begin{proof}
	Let $f,g \colon n \to p + n$ be morphisms in $\mathsf{RefBeh}$. We proceed by induction on $n$. If $n = 0$, then $f^\dagger = f = 0_p = g = g^\dagger$ and hence $d^{0, p}(f^\dagger, g^\dagger) \leq d^{0, p}(f, g)$.
	
	If $n = 1$, then we have the following
	\begin{align*}
		d^{1,p}(f^\dagger, g^\dagger) &= \mathsf{bd}_{\overline\partial} (\mu v_{p+1}.f, \mu v_{p+1}. g)\leq \mathsf{bd}_{\overline\partial} (f, g) \tag{\Cref{cor:rec_nexp}} 
	\end{align*}
	
	Let $n = n' + 1$. Recall, that we can represent $f$ and $g$ in the following way
	\begin{gather*}
		f = \langle f_1, \dots, f_{n'}, f_{n' + 1} \rangle = \langle \langle f_1, \dots, f_{n'} \rangle , f_{n' + 1}\rangle\qquad
		g = \langle g_1, \dots, g_{n'}, f_{g' + 1} \rangle = \langle \langle g_1, \dots, g_{n'} \rangle , g_{n' + 1}\rangle\qquad
	\end{gather*}
	We can apply the induction hypothesis and obtain the following:
	$$
	d^{n', p + 1}(\langle f_1, \dots, f_{n'}\rangle^\dagger, \langle g_1, \dots, g_{n'}\rangle^\dagger ) \leq d^{n', n' + p + 1}(\langle f_1, \dots, f_{n'}\rangle, \langle g_1, \dots, g_{n'}\rangle ) 
	$$
	Using it, one can show that:
	\begin{align*}
	d^{p + 1 + n' , p + 1}\left(\langle\id_{p + 1},\langle f_1, \dots, f_{n'}\rangle^\dagger \rangle,\langle\id_{p + 1},\langle g_1, \dots, g_{n'}\rangle^\dagger\rangle \right) &=d^{n', p + 1}(\langle f_1, \dots, f_{n'}\rangle^\dagger, \langle g_1, \dots, g_{n'}\rangle^\dagger)\\
	&\leq  d^{n', p + 1 + n'}(\langle f_1, \dots, f_{n'}\rangle, \langle g_1, \dots, g_{n'}\rangle ) 
	\end{align*}
	For the sake of simplicity, let
	$
	k = f_{n' + 1};\langle\id_{p + 1} ,\langle f_1, \dots, f_{n'}\rangle^\dagger\rangle
	$ and $l = g_{n' + 1};\langle\id_{p + 1} ,\langle g_1, \dots, g_{n'}\rangle^\dagger\rangle$.
	Using \Cref{cor:seq_nexp} we know that: 
	\begin{align*}
	d^{1,p+1}(k,l)&\leq \max\left\{d^{1,p + 1}(f_{n' + 1},g_{n' + 1}), 	d^{p + n', p + 1}\left(\langle\id_{p + 1} ,\langle f_1, \dots, f_{n'}\rangle^\dagger\rangle,\langle\id_{p + 1} ,\langle g_1, \dots, g_{n'}\rangle^\dagger \rangle \right) \right\}\\
	&\leq \max \{d^{1,p + 1}(f_{n' + 1},g_{n' + 1}), d^{n', p + 1 + n'}(\langle f_1, \dots, f_{n'}\rangle, \langle g_1, \dots, g_{n'}\rangle ) \} \\
	&\leq d^{n' + 1, p + 1 + n'}(f,g)
	\end{align*}
	Because of the above, we can use \Cref{cor:rec_nexp} and obtain:
	\begin{align*}
		d^{1,p}(k^\dagger, l^\dagger) \leq 	d^{1,p + 1}(k,l) \leq d^{n', p + 1 + n'}(f,g)
	\end{align*}
	Recall the \textbf{scalar pairing identity}, which states that:
	\begin{gather*}
	f^\dagger = \langle\langle f_1, \dots, f_{n'}\rangle^\dagger ; \langle \id_p ,k^\dagger\rangle  \rangle, k^\dagger\rangle \qquad \text{and} \qquad g^\dagger = \langle\langle g_1, \dots, g_{n'}\rangle^\dagger ; \langle \id_p, l^\dagger \rangle  \rangle, l^\dagger\rangle
	\end{gather*}
	
	We have that
	\begin{align*}
		d^{n' + 1, p}(f^\dagger, g^\dagger) &\leq \max \{d^{n', p + 1}(\langle f_1, \dots, f_{n'}\rangle^\dagger; \langle \id_p , k^\dagger\rangle, \langle g_1, \dots, g_{n'}\rangle; \langle \id_p , l^\dagger\rangle), d^{1,p}(k^\dagger, l^\dagger)\}\\
		&\leq \max \{d^{n', 1 + p}(\langle f_1, \dots, f_{n'}\rangle^\dagger, \langle g_1, \dots, g_{n'}\rangle^\dagger), d^{p + 1, p}(\langle id_p , k^\dagger\rangle, \langle id_p ,l^\dagger\rangle),  d^{1,p}(k^\dagger, l^\dagger) \tag{\Cref{cor:seq_nexp}}\\
		&\leq \max \{d^{n', p + 1}(\langle f_1, \dots, f_{n'}\rangle^\dagger, \langle g_1, \dots, g_{n'}\rangle^\dagger), d^{1,p}(k^\dagger, l^\dagger)\\
		&\leq  \max \{^{n', p + 1 + n'}(\langle f_1, \dots, f_{n'}\rangle, \langle g_1, \dots, g_{n'}\rangle ), d^{1,p}(k^\dagger, l^\dagger)\} \tag{Induction hypothesis}\\
		&\leq \max \{^{n', p + 1 + n'}(\langle f_1, \dots, f_{n'}\rangle, \langle g_1, \dots, g_{n'}\rangle ),  d^{n' + 1, p + 1 + n'}(f,g)\} \\
		&=  d^{n' + 1, p + 1 + n'}(f,g)
	\end{align*}
	which completes the proof.
\end{proof}
\begin{corollary}\label{cor:trace_nonexpansive}
	The trace on $\RegBeh$ is nonexpansive.
\end{corollary}
\begin{proof}
	Immediate consequence of \Cref{thm:trace}, \Cref{lem:seq_nonexpansive} and \Cref{lem:dagger_nonexpansive}.
\end{proof}
\propnexp
\begin{proof}
	Follows from \Cref{lem:seq_nonexpansive}, \Cref{lem:pairs_nonexpansive},\Cref{lem:tens_nonexpansive}, \Cref{lem:dagger_nonexpansive}, \Cref{cor:trace_nonexpansive}.
\end{proof}

Every homset $\Int{\RegBeh}((A^+, A^-),(B^+,B^-))$ of $\Int{\RegBeh}$ can be equipped with a pseudometric space $d^{A^+ + B^-, A^- + B^+}$ associated with the homset $\RegBeh(A^+ + B^-, A^- + B^+)$ of $\RegBeh$. 
\begin{corollary}\label{cor:int_isometry}
	The fully faithful functor $N \colon \RegBeh \to \Int\RegBeh$ is an isometry on homsets.
\end{corollary}\label{cor:embedding_iso}
\begin{proof}
	Let $f,g \in \RegBeh(m,n)$. From, the definition of $N$, we immediately have that
	$$
		d^{N(m), N(n)}(N(f),(N(g))=d^{m,n}(f,g)
	$$ 
\end{proof}
\propnexpint
\begin{proof}
	Immediate consequence of \Cref{cor:trace_nonexpansive}, \Cref{lem:seq_nonexpansive}, \Cref{lem:tens_nonexpansive} and \Cref{cor:int_isometry}.
\end{proof}
\section{Soundness}
\soundnessequality*
\begin{proof}
	We verify that all equations defining $\Syn$ are satisfied. When dealing with left-to-right diagrams, we will make use of the fact that $\RegBeh$ fully faithfully embeds into $\Int{\RegBeh}$ (\Cref{thm:trace_embeds_int}) and hence it suffices to verify the axioms in $\RegBeh$, rather then in their completion to $\Int{\RegBeh}$. \textsf{(A1)} is satisfied because of the yanking property of trace operation defined on $\RegBeh$, while \textsf{(A2)} is its dual in $\Int{\RegBeh}$ and can be verified similarly. \textsf{(B1)}, \textsf{(B2)} and \textsf{(B3)} are satisfied because $+$ defined on ${\Expr}/{\equiv}$ is a commutative monoid with $0$ being its identity. Similarly, \textsf{(B4)}, \textsf{(B5)} and \textsf{(B6)} are satisfied because of the universal property of coproduct on $\RegBeh$ and $\nabla_1$ being the codiagonal morphism. For \textsf{(B8)} we rely on the fact that $\nabla_1 ; \langle v_1 + v_2 \rangle = \langle v_1, v_1\rangle ; \langle v_1 + v_2 \rangle = \langle v_1 + v_2, v_1 + v_2 \rangle$. \textsf{(B8)} holds because $\nabla_1 ; \langle  0\rangle= \langle v_1, v_1\rangle ; \langle 0 \rangle = \langle 0, 0 \rangle $. 
	\textsf{(B9)} is satisfied because $+$ is idempotent. Finally, \textsf{(B10)} corresponds to taking the dagger of $\langle v_1 + v_2 \rangle$ and captures the identity $\mu v_2.(v_1 + v_2) = v_1$ of Milner's Algebra of Regular Behaviours. Finally, \textsf{(C1)} holds, because $\nabla_1 ; \langle a . v_1 \rangle = \langle v_1, v_1 \rangle ; \langle a. v_1 \rangle = \langle a. v_1, a.v_1\rangle$.
\end{proof}
	\begin{lemma}\label{lem:soundness_sublemma}
		All the inference rules defining the distance on $\Syn$ are satisfied in $\Int{\RegBeh}$.
	\end{lemma}
	\begin{proof}
		For most of the rules, the proof is straightforward. The soundness of \textsf{(Top)} follows from the fact that the distance on morphisms of $\Int{\RegBeh}$ is $1$-bounded, while \textsf{(Max)} captures the transitivity of partial order on the rational numbers. \textsf{(Refl)}, \textsf{(Sym)} and \textsf{(Triang)} are satisfied because the distance function on each hom-set of $\Int{\RegBeh}$ is a pseudometric space. \textsf{(Cont)} captures the Cauchy completeness of reals, while \textsf{(Seq)} and \textsf{(Tens)} are immediate consquence of \Cref{cor:sem_enriched}, stating that $\Int{\RegBeh}$ is $\PMet$-enriched symmetric monoidal category. For the remaining two rules, we will make use of the fact that the fully faithful embedding $N \colon \RegBeh \to \Int{\RegBeh}$ is an isometry (\Cref{cor:embedding_iso}), hence for the left-to-right diagrams it suffices to check the rules in $\RegBeh$. The soundness of \textsf{(Pref)} is an immediate consequence of \Cref{cor:prefixing_discounts}. Finally, \textsf{(Codel)} follows from the uniqueness of maps from the initial object in $\RegBeh$.
		\end{proof}
	\soundness*
	\begin{proof}
	Induction on the length of derivation and the usage of \Cref{lem:soundness_sublemma}.
	\end{proof}
\section{Completeness}
\begin{lemma}[Trace canonical form]\label{lem:traceform}
For any diagram $d\from \objr^m\to\objr^n$, we can always find a relation-diagram $c\from \objr^{\ell+m}\to\objr^{\ell+n}$ such that
\begin{equation*}
\dbox{d}{m}{n} \quad = \quad \traceaction{c}{m}{n}{\ell}{x}
\end{equation*}
where $\scalar{x}^{\!\!\ell}$ denotes a vertical composite of $\ell$-many $\scalar{a}$ generators.
\end{lemma}
\begin{proof}
The proof is the same as~\cite[Lemma 4.11]{piedeleu2023finite} which only uses the axioms of SMCs.
\end{proof}
\begin{lemma}\label{lem:guarded-precompose}
For any guarded matrix-diagram $c\from \objr^\ell\to\objr^m$ and any two diagrams $d_1,d_2\from \objr^m\to\objr^n$ such that
\[\dbox{d_1}{m}{n}\disteq{\epsilon}\dbox{d_2}{m}{n}\]
we have
\[
\tikzset{x=1em, y=2.1ex}
\InputIfFileExists{c-d1.tikz}{}{\input{./tikz/c-d1.tikz}}
\tikzset{x=1em, y=1.5ex}
\disteq{\epsilon/2}
\tikzset{x=1em, y=2.1ex}
\InputIfFileExists{c-d2.tikz}{}{\input{./tikz/c-d2.tikz}}
\tikzset{x=1em, y=1.5ex}
\]
\end{lemma}
\begin{proof}
We rely on the definition of guarded matrix diagrams. Recall that since $c$ is guarded, we can factor it as
\[\dbox{c}{\ell}{m} = 
\tikzset{x=1em, y=2.1ex}
\InputIfFileExists{matrix-diagram-blocks.tikz}{}{\input{./tikz/matrix-diagram-blocks.tikz}}
\tikzset{x=1em, y=1.5ex}
\]
where $c_0$ is a diagram composed only of $\Bcomult,\Bcounit$, $\vec{a}$ is a vertical composite of $k$ $\scalar{a_i}$ generators, and $c_1$ is a diagram composed only of $\Bmult,\Bunit$. Hence,
\begin{equation*}
	\inferrule{
	\inferrule{
	\inferrule{
	\dbox{d_1}{m}{n}\disteq{\epsilon}\dbox{d_2}{m}{n}\\ \dbox{c_1}{k}{m}\disteq{0}\dbox{c_1}{k}{m}}{
\tikzset{x=1em, y=2.1ex}
\InputIfFileExists{c1-d1.tikz}{}{\input{./tikz/c1-d1.tikz}}
\tikzset{x=1em, y=1.5ex}
\disteq{\epsilon} 
\tikzset{x=1em, y=2.1ex}
\InputIfFileExists{c1-d2.tikz}{}{\input{./tikz/c1-d2.tikz}}
\tikzset{x=1em, y=1.5ex}
 \\ \vec{a}\in\Sigma^k}{\mathsf{(Seq)}}}
	{
\tikzset{x=1em, y=2.1ex}
\InputIfFileExists{prefix-c1-d1.tikz}{}{\input{./tikz/prefix-c1-d1.tikz}}
\tikzset{x=1em, y=1.5ex}
\disteq{\epsilon/2} 
\tikzset{x=1em, y=2.1ex}
\InputIfFileExists{prefix-c1-d2.tikz}{}{\input{./tikz/prefix-c1-d2.tikz}}
\tikzset{x=1em, y=1.5ex}
 \\ \dbox{c_0}{k}{m}\disteq{0}\dbox{c_0}{k}{m}}{\mathsf{(Pref)}}}{
\tikzset{x=1em, y=2.1ex}
\InputIfFileExists{c0-prefix-c1-d1.tikz}{}{\input{./tikz/c0-prefix-c1-d1.tikz}}
\tikzset{x=1em, y=1.5ex}
\disteq{\epsilon/2} 
\tikzset{x=1em, y=2.1ex}
\InputIfFileExists{c0-prefix-c1-d2.tikz}{}{\input{./tikz/c0-prefix-c1-d2.tikz}}
\tikzset{x=1em, y=1.5ex}
}{\mathsf{(Seq)}}
\end{equation*}
The last line is what we wanted to show.
\end{proof}
\representation*
\begin{proof}
The proof is the same as~\cite[Proposition 4.7]{piedeleu2023finite}. All axioms used in that proof are in our theory.
\end{proof}
\begin{lemma}\label{lem:matrix-cocopy}
For any matrix-diagram $d\from \objr^m\to \objr^n$, we have
\begin{align*}

\tikzset{x=1em, y=2.1ex}
\InputIfFileExists{global-merge.tikz}{}{\input{./tikz/global-merge.tikz}}
\tikzset{x=1em, y=1.5ex}
\quad = \quad 
\tikzset{x=1em, y=2.1ex}
\InputIfFileExists{global-merge-1.tikz}{}{\input{./tikz/global-merge-1.tikz}}
\tikzset{x=1em, y=1.5ex}

\end{align*}
\end{lemma}
\begin{proof}
See~\cite[Lemma 4.9]{piedeleu2023finite} (\textsf{co-cpy}). It is a simple  structural induction. For the base cases, all the generators of matrix-diagrams satisfy the equality of the lemma, by axioms \textsf{(B5)}, \textsf{(B7)},\textsf{(B8)}, and \textsf{(C1)}. The inductive cases are straightforward.
\end{proof}
\unrolling*
\begin{proof}
\begin{align*}

\tikzset{x=1em, y=2.1ex}
\InputIfFileExists{d-star.tikz}{}{\input{./tikz/d-star.tikz}}
\tikzset{x=1em, y=1.5ex}
 & \myeq{B7} 
\tikzset{x=1em, y=2.1ex}
\InputIfFileExists{d-star-1.tikz}{}{\input{./tikz/d-star-1.tikz}}
\tikzset{x=1em, y=1.5ex}
\myeq{A1} 
\tikzset{x=1em, y=2.1ex}
\InputIfFileExists{d-star-2.tikz}{}{\input{./tikz/d-star-2.tikz}}
\tikzset{x=1em, y=1.5ex}

\\
& \myeq{SMC} 
\tikzset{x=1em, y=2.1ex}
\InputIfFileExists{d-star-3.tikz}{}{\input{./tikz/d-star-3.tikz}}
\tikzset{x=1em, y=1.5ex}
\myeq{A1} 
\tikzset{x=1em, y=2.1ex}
\InputIfFileExists{d-star-4.tikz}{}{\input{./tikz/d-star-4.tikz}}
\tikzset{x=1em, y=1.5ex}

\\
&\myeq{Lemma~\ref{lem:matrix-cocopy}} 
\tikzset{x=1em, y=2.1ex}
\InputIfFileExists{d-star-5.tikz}{}{\input{./tikz/d-star-5.tikz}}
\tikzset{x=1em, y=1.5ex}
 \myeq{A1} 
\tikzset{x=1em, y=2.1ex}
\InputIfFileExists{d-star-unroll.tikz}{}{\input{./tikz/d-star-unroll.tikz}}
\tikzset{x=1em, y=1.5ex}

\end{align*}
\end{proof}
\globalcocopy*
\begin{proof}
First, by \Cref{thm:representation}, we can find a matrix-diagram, $a\from \objr^{\ell+m}\to \objr^{\ell+m}$ and a relation-diagram $o\from \objr^{\ell+m}\to \objr$ such that
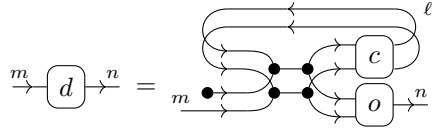
\[\dbox{d}{m}{n} \; = \;
\tikzset{x=1em, y=2.1ex}
\InputIfFileExists{automata-rep.tikz}{}{\input{./tikz/automata-rep.tikz}}
\tikzset{x=1em, y=1.5ex}
\] 
We will show that 
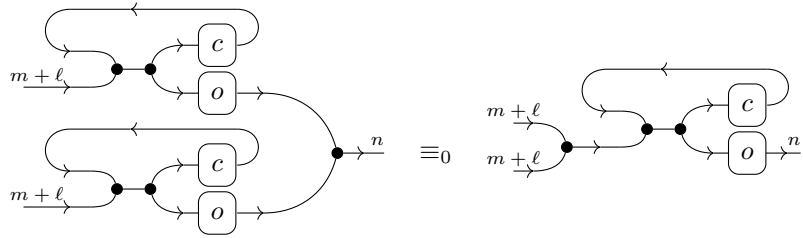
\[
\tikzset{x=1em, y=2.1ex}
\InputIfFileExists{c-star-o-merge.tikz}{}{\input{./tikz/c-star-o-merge.tikz}}
\tikzset{x=1em, y=1.5ex}
\quad\disteq{0}\quad
\tikzset{x=1em, y=2.1ex}
\InputIfFileExists{merge-c-star-o.tikz}{}{\input{./tikz/merge-c-star-o.tikz}}
\tikzset{x=1em, y=1.5ex}
\]
from which the statement of the lemma immediately follows, by pre-composing with $
\tikzset{x=1em, y=2.1ex}
\begin{tikzpicture}
	\begin{pgfonlayer}{nodelayer}
		\node [style=black] (37) at (-0.5, 0) {};
		\node [style=none] (43) at (0.25, 0) {};
		\node [style=none] (44) at (0.75, 0) {};
	\end{pgfonlayer}
	\begin{pgfonlayer}{edgelayer}
		\draw [->] (37) to (43.center);
		\draw (44.center) to (43.center);
	\end{pgfonlayer}
\end{tikzpicture}
}
\tikzset{x=1em, y=1.5ex}
$.

Since $c\from m+\ell\to m+\ell$ is a guarded matrix-diagram, so is
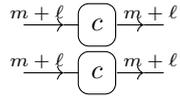
\[
\tikzset{x=1em, y=2.1ex}
\InputIfFileExists{cxc.tikz}{}{\input{./tikz/cxc.tikz}}
\tikzset{x=1em, y=1.5ex}
\]
Therefore, by Lemma~\ref{lem:guarded-precompose}, we get
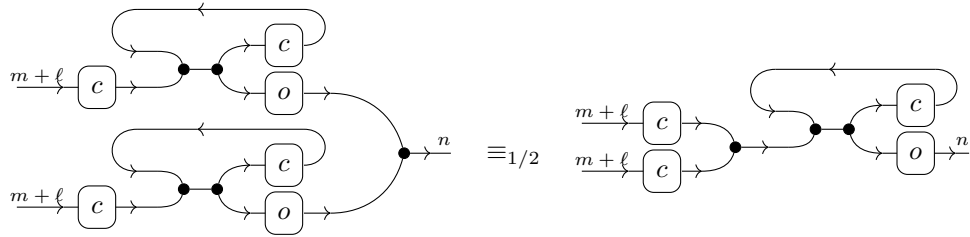
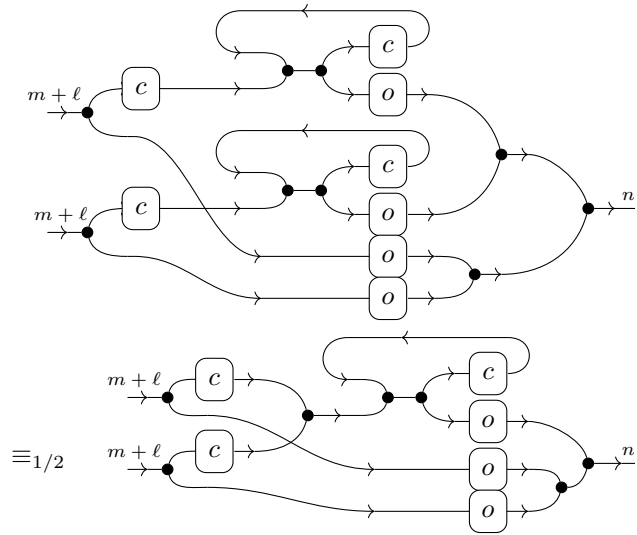
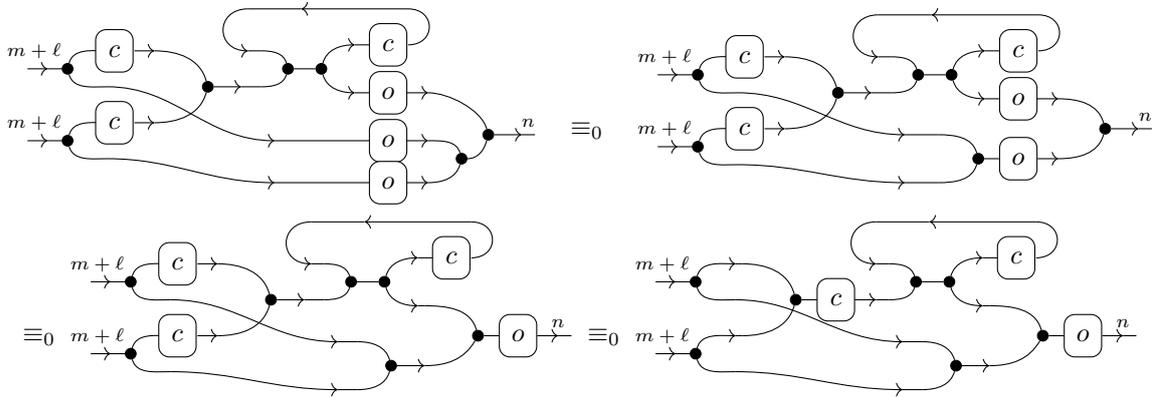
\[
\tikzset{x=1em, y=2.1ex}
\InputIfFileExists{prefix-c-star-o-merge.tikz}{}{\input{./tikz/prefix-c-star-o-merge.tikz}}
\tikzset{x=1em, y=1.5ex}
 \quad\disteq{1/2}\quad
\tikzset{x=1em, y=2.1ex}
\InputIfFileExists{prefix-merge-c-star-o.tikz}{}{\input{./tikz/prefix-merge-c-star-o.tikz}}
\tikzset{x=1em, y=1.5ex}
 \]
and thus, 
\begin{align*}

\tikzset{x=1em, y=2.1ex}
\InputIfFileExists{unroll-c-star-o-merge.tikz}{}{\input{./tikz/unroll-c-star-o-merge.tikz}}
\tikzset{x=1em, y=1.5ex}
\\
\qquad \disteq{1/2}\;\quad
\tikzset{x=1em, y=2.1ex}
\InputIfFileExists{unroll-merge-c-star-o.tikz}{}{\input{./tikz/unroll-merge-c-star-o.tikz}}
\tikzset{x=1em, y=1.5ex}

\end{align*}
We also have
\begin{equation*}
\begin{gathered}

\tikzset{x=1em, y=2.1ex}
\InputIfFileExists{unroll-merge-c-star-o.tikz}{}{\input{./tikz/unroll-merge-c-star-o.tikz}}
\tikzset{x=1em, y=1.5ex}
\quad\disteq{0}\quad
\tikzset{x=1em, y=2.1ex}
\InputIfFileExists{unroll-merge-c-star-o-1.tikz}{}{\input{./tikz/unroll-merge-c-star-o-1.tikz}}
\tikzset{x=1em, y=1.5ex}

\\
 \disteq{0} \;
\tikzset{x=1em, y=2.1ex}
\InputIfFileExists{unroll-merge-c-star-o-2.tikz}{}{\input{./tikz/unroll-merge-c-star-o-2.tikz}}
\tikzset{x=1em, y=1.5ex}

\; \disteq{0} \;
\tikzset{x=1em, y=2.1ex}
\InputIfFileExists{unroll-merge-c-star-o-3.tikz}{}{\input{./tikz/unroll-merge-c-star-o-3.tikz}}
\tikzset{x=1em, y=1.5ex}

\end{gathered}
\end{equation*}
where the last step uses Lemma~\ref{lem:matrix-cocopy} and \textsf{(Refl)} to merge the two occurrences of the matrix-diagram $c$. Resuming, we get
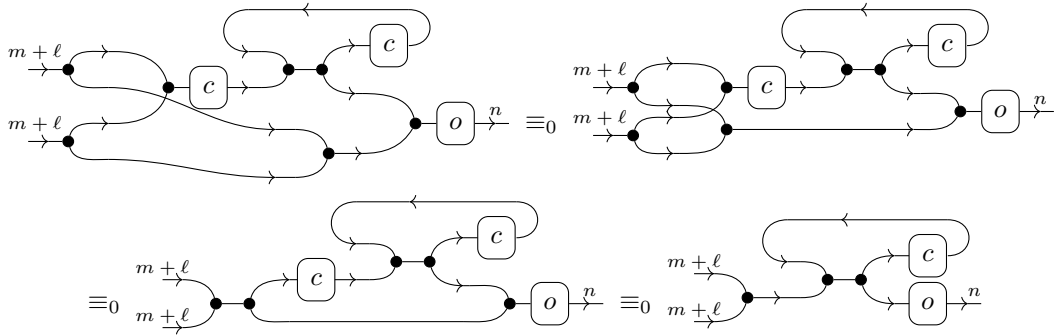
\begin{equation*}
\begin{gathered}

\tikzset{x=1em, y=2.1ex}
\InputIfFileExists{unroll-merge-c-star-o-3.tikz}{}{\input{./tikz/unroll-merge-c-star-o-3.tikz}}
\tikzset{x=1em, y=1.5ex}
\; \disteq{0} \;
\tikzset{x=1em, y=2.1ex}
\InputIfFileExists{unroll-merge-c-star-o-4.tikz}{}{\input{./tikz/unroll-merge-c-star-o-4.tikz}}
\tikzset{x=1em, y=1.5ex}
 \\
\; \disteq{0} \;
\tikzset{x=1em, y=2.1ex}
\InputIfFileExists{unroll-merge-c-star-o-5.tikz}{}{\input{./tikz/unroll-merge-c-star-o-5.tikz}}
\tikzset{x=1em, y=1.5ex}
\;
\disteq{0} \;
\tikzset{x=1em, y=2.1ex}
\InputIfFileExists{merge-c-star-o.tikz}{}{\input{./tikz/merge-c-star-o.tikz}}
\tikzset{x=1em, y=1.5ex}

\end{gathered}
\end{equation*}
We can show in the same way that
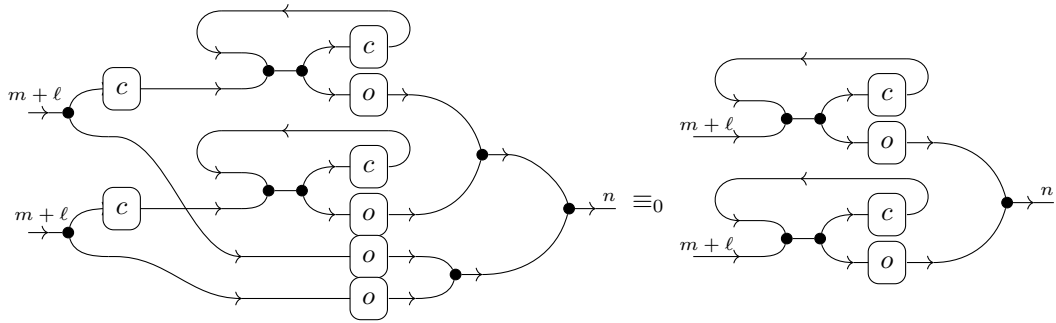
\[
\tikzset{x=1em, y=2.1ex}
\InputIfFileExists{unroll-c-star-o-merge.tikz}{}{\input{./tikz/unroll-c-star-o-merge.tikz}}
\tikzset{x=1em, y=1.5ex}
\;\disteq{0}\;  
\tikzset{x=1em, y=2.1ex}
\InputIfFileExists{c-star-o-merge.tikz}{}{\input{./tikz/c-star-o-merge.tikz}}
\tikzset{x=1em, y=1.5ex}
\]
Thus, we have shown that 
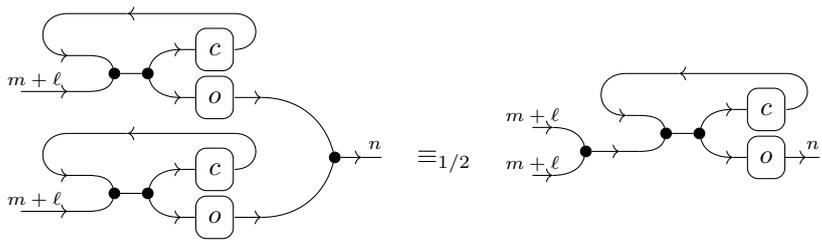
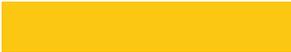
\[
\tikzset{x=1em, y=2.1ex}
\InputIfFileExists{c-star-o-merge.tikz}{}{\input{./tikz/c-star-o-merge.tikz}}
\tikzset{x=1em, y=1.5ex}
\quad\disteq{1/2}\quad
\tikzset{x=1em, y=2.1ex}
\InputIfFileExists{merge-c-star-o.tikz}{}{\input{./tikz/merge-c-star-o.tikz}}
\tikzset{x=1em, y=1.5ex}
\]
In the same way, we can show that 
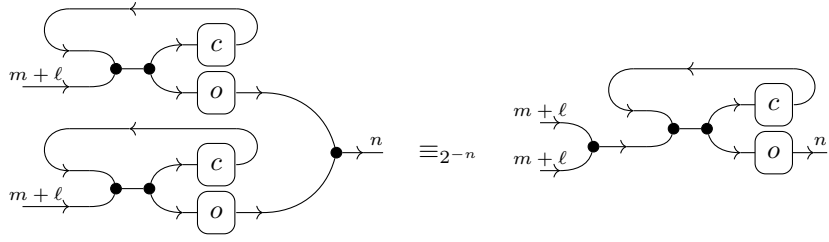
\[
\tikzset{x=1em, y=2.1ex}
\InputIfFileExists{c-star-o-merge.tikz}{}{\input{./tikz/c-star-o-merge.tikz}}
\tikzset{x=1em, y=1.5ex}
\quad\disteq{2^{-n}}\quad
\tikzset{x=1em, y=2.1ex}
\InputIfFileExists{merge-c-star-o.tikz}{}{\input{./tikz/merge-c-star-o.tikz}}
\tikzset{x=1em, y=1.5ex}
\]
for any $n\in\N$ and thus, by the continuity axiom \textsf{(Cont)}, we conclude that
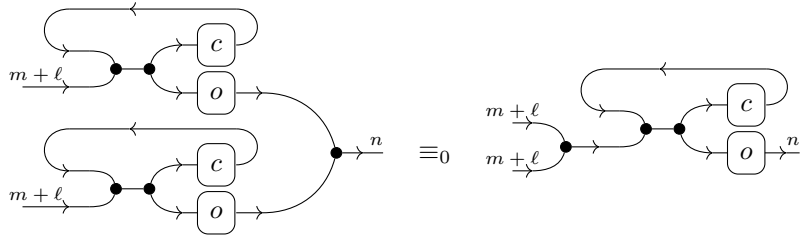
\[
\tikzset{x=1em, y=2.1ex}
\InputIfFileExists{c-star-o-merge.tikz}{}{\input{./tikz/c-star-o-merge.tikz}}
\tikzset{x=1em, y=1.5ex}
\quad\disteq{0}\quad
\tikzset{x=1em, y=2.1ex}
\InputIfFileExists{merge-c-star-o.tikz}{}{\input{./tikz/merge-c-star-o.tikz}}
\tikzset{x=1em, y=1.5ex}
\]
as we wanted to show.
\end{proof}
\begin{lemma}\label{lem:merge-repr}
	Let $e,f\from\objr\to\objr^n$, such that $\sem{e} = N(s)$ and $\sem{f} = N(d)$, where $s, t \in \RegBeh(1,n)$. We have that
	\begin{align*}
		\sem{
\tikzset{x=1em, y=2.1ex}
\InputIfFileExists{merge-cxd.tikz}{}{\input{./tikz/merge-cxd.tikz}}
\tikzset{x=1em, y=1.5ex}
} &= N \left( \langle s, t\rangle \right)
	\end{align*}
\end{lemma}
\begin{proof}
	\begin{align*}
		\sem{
\tikzset{x=1em, y=2.1ex}
\InputIfFileExists{merge-cxd.tikz}{}{\input{./tikz/merge-cxd.tikz}}
\tikzset{x=1em, y=1.5ex}
} &=(N(s)\oplus N(f)); N(\nabla_1) \\
		&= N((s \oplus t) ; \nabla_1 ) \tag{Functoriality of $N$}\\
		&= N( \langle s,t[v_2/v_1] \rangle ; \langle \id_1, \id_1\rangle )\\
		&=N( \langle s , t \rangle ) \\
	\end{align*}
\end{proof}
\begin{lemma}\label{lem:distance_on_merge}
	Let $e_1, e_2, f_1,f_2 \colon \objr \to \objr^n$. We have that
	\begin{align*}
		d^{N(2),N(n)}\left(\sem{
\tikzset{x=1em, y=2.1ex}
\InputIfFileExists{merge-cxd-1.tikz}{}{\input{./tikz/merge-cxd-1.tikz}}
\tikzset{x=1em, y=1.5ex}
}, \sem{
\tikzset{x=1em, y=2.1ex}
\InputIfFileExists{merge-cxd-2.tikz}{}{\input{./tikz/merge-cxd-2.tikz}}
\tikzset{x=1em, y=1.5ex}
}\right)\\ \quad= \max \left\{d^{N(1),N(n)}\left(\sem{\dbox{e_1}{}{n}},\sem{\dbox{e_2}{}{n}} \right), d^{N(1),N(n)} \left(\sem{\dbox{f_1}{}{n}}, \sem{\dbox{f_2}{}{n}} \right)\right\}
	\end{align*}
\end{lemma}
\begin{proof}
	Since $e_1, e_2, f_1,f_2$ are left-to-right diagrams, we can safely assume that there exist $s_1, s_2, t_1, t_2 \in \RegBeh(1,n)$ such that $\sem{c_1} = N(s_1)$,$\sem{c_2} = N(s_2)$,$\sem{d_1} = N(t_1)$. We have the following
	\begin{align*}
		&d^{N(2),N(n)}\left(\sem{
\tikzset{x=1em, y=2.1ex}
\InputIfFileExists{merge-cxd-1.tikz}{}{\input{./tikz/merge-cxd-1.tikz}}
\tikzset{x=1em, y=1.5ex}
}, \sem{
\tikzset{x=1em, y=2.1ex}
\InputIfFileExists{merge-cxd-2.tikz}{}{\input{./tikz/merge-cxd-2.tikz}}
\tikzset{x=1em, y=1.5ex}
}\right)\\
		&\qquad = d^{N(2),N(n)}(N(\langle e_1, f_1 \rangle), N\langle e_2, f_2 \rangle)\tag{\Cref{lem:merge-repr}}\\
		&\qquad = d^{2,n}(\langle e_1, f_1 \rangle, \langle e_2, f_2 \rangle)\tag{\Cref{cor:int_isometry}}\\
		&\qquad= \max \{\mathsf{bd}_{\overline \partial} (e_1, e_2), \mathsf{bd}_{\overline \partial} (f_1, f_2)\}\tag{Def. in distance of $\RegBeh$}\\
		&\qquad = \max \{d^{1,n}(e_1, e_2), d^{1,n}(f_1, f_2)\}\\
		&\qquad = \max \{d^{N(1), N(n)}(N(e_1), N(e_2)), d^{N(1), N(n)}(N(f_1), N(f_2))\}\tag{\Cref{lem:merge-repr}}\\
		&\qquad = \max \left\{d^{N(1),N(n)}\left(\sem{\dbox{e_1}{}{n}},\sem{\dbox{e_2}{}{n}} \right), d^{N(1),N(n)} \left(\sem{\dbox{f_1}{}{n}}, \sem{\dbox{f_2}{}{n}} \right)\right\}
	\end{align*}
\end{proof}

Let $F = \{f_i \colon \objr \to \objr^n\}_{i \in I}$ be an indexed collection of string diagrams. Given a finite indexed collection $A = \{f_{i_1}, \dots, f_{i_k}\}_{k \in K} \subseteq F$ of string diagrams from the set $F$, we define its convolution to be the string diagram $R_A \colon \objr \to \objr^n$ given by
	$$
	
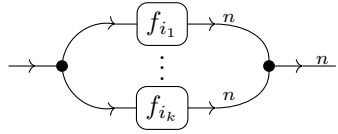
\tikzset{x=1em, y=2.1ex}
\InputIfFileExists{set-repr.tikz}{}{\input{./tikz/set-repr.tikz}}
\tikzset{x=1em, y=1.5ex}

	$$ 
	If two finite indexed collections $A$ and $B$ of elements of $F$ are equivalent modulo ACI (Associativity, Commutativity, Idempotence) then their convolutions are at distance zero from each other. Hence, given a finite subset of $F$ we can unambiguously talk about its convolution.
\begin{lemma}
\label{lem:convolution-aci}
Any two finite indexed collections $A = \{f_{i_1}, \dots, f_{i_k}\}_{k \in K}$ and $B = \{g_{i_1}, \dots, g_{i_k}\}_{k \in K}$ that are equal as sets, then their convolution are at distance zero from each other.
\end{lemma}
\begin{proof}
It is enough to show that the convolution of two string diagrams is an associative, commutative and idempotent operation (up to $\disteq{0}$).
\begin{description}
\item[Associativity:] For any $f_1,f_2$, and $f_3$, we have
\begin{align*}

\tikzset{x=1em, y=2.1ex}
\InputIfFileExists{convolution-assoc.tikz}{}{\input{./tikz/convolution-assoc.tikz}}
\tikzset{x=1em, y=1.5ex}
 &= 
\tikzset{x=1em, y=2.1ex}
\InputIfFileExists{convolution-assoc-1.tikz}{}{\input{./tikz/convolution-assoc-1.tikz}}
\tikzset{x=1em, y=1.5ex}

\\
& = 
\tikzset{x=1em, y=2.1ex}
\InputIfFileExists{convolution-assoc-2.tikz}{}{\input{./tikz/convolution-assoc-2.tikz}}
\tikzset{x=1em, y=1.5ex}

\end{align*}
\item[Commutativity:]
\begin{align*}

\tikzset{x=1em, y=2.1ex}
\InputIfFileExists{convolution-commut.tikz}{}{\input{./tikz/convolution-commut.tikz}}
\tikzset{x=1em, y=1.5ex}
 &= 
\tikzset{x=1em, y=2.1ex}
\InputIfFileExists{convolution-commut-1.tikz}{}{\input{./tikz/convolution-commut-1.tikz}}
\tikzset{x=1em, y=1.5ex}

\\
& = 
\tikzset{x=1em, y=2.1ex}
\InputIfFileExists{convolution-commut-2.tikz}{}{\input{./tikz/convolution-commut-2.tikz}}
\tikzset{x=1em, y=1.5ex}

\\
& = 
\tikzset{x=1em, y=2.1ex}
\InputIfFileExists{convolution-commut-3.tikz}{}{\input{./tikz/convolution-commut-3.tikz}}
\tikzset{x=1em, y=1.5ex}

\end{align*}
\item[Idempotence:]
\begin{align*}

\tikzset{x=1em, y=2.1ex}
\InputIfFileExists{convolution-idemp.tikz}{}{\input{./tikz/convolution-idemp.tikz}}
\tikzset{x=1em, y=1.5ex}
 &\disteq{0} 
\tikzset{x=1em, y=2.1ex}
\InputIfFileExists{convolution-idemp-1.tikz}{}{\input{./tikz/convolution-idemp-1.tikz}}
\tikzset{x=1em, y=1.5ex}
 \qquad & (\text{Lemma}~\ref{thm:co-copy-delete})
\\
& = \dbox{f_1}{}{n} &
\end{align*}
\end{description}
\end{proof}
\sumconvolutions*
\begin{proof}
First, since $c$ and $d$ are left-to-right, there exists expressions $s,t\in\RegBeh(1,n)$, such that $\sem{c} = N(s)$ and $\sem{d} = N(t)$. Then,
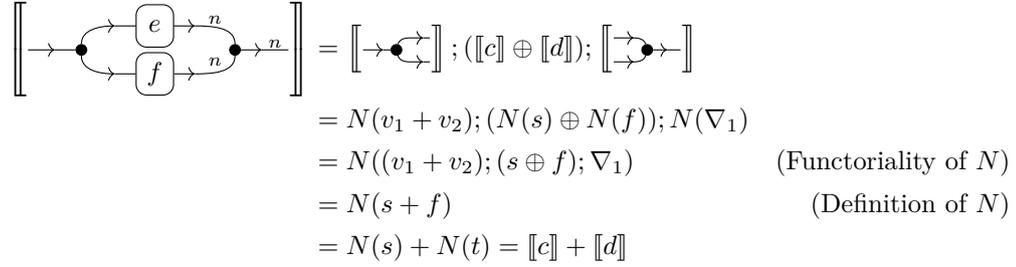
\begin{align*}
\sem{
\tikzset{x=1em, y=2.1ex}
\InputIfFileExists{convolution-cxd.tikz}{}{\input{./tikz/convolution-cxd.tikz}}
\tikzset{x=1em, y=1.5ex}
} & = \sem{\Bcomult}; (\sem{c}\oplus \sem{d});\sem{\Bmult}
\\
& = N(v_1+v_2) ; (N(s)\oplus N(f)) ;  N(\nabla_1) & 
\\
& = N((v_1+v_2); (s\oplus f); \nabla_1) & (\text{Functoriality of $N$})
\\
& = N(s+f) & (\text{Definition of $N$})
\\
& = N(s) + N(t) = \sem{c} + \sem{d} &
\end{align*}

\end{proof}
\begin{lemma}
\label{lem:diagram-to-eq-system}
For any diagram $f\from \objr^m\to \objr^n$, if
$$
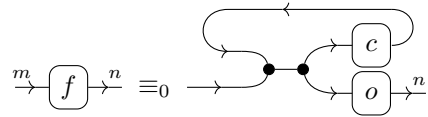
\dbox{f}{m}{n} \; \disteq{0} \;
\tikzset{x=1em, y=2.1ex}
\InputIfFileExists{automata-rep-all-states.tikz}{}{\input{./tikz/automata-rep-all-states.tikz}}
\tikzset{x=1em, y=1.5ex}

$$ 
for some guarded matrix-diagram $c\from \objr^{\ell+m}\to \objr^{\ell+m}$ and a relation-diagram $o\from \objr^{\ell+m}\to \objr^n$,
then, for all $i\in\{1,\dots,m\}$ we can find $\{a_j, f_{i_j}\}_{1\leq j\leq k}$, and $\{v_{q_j}\}_{1\leq j\leq \ell}$ such that
$$
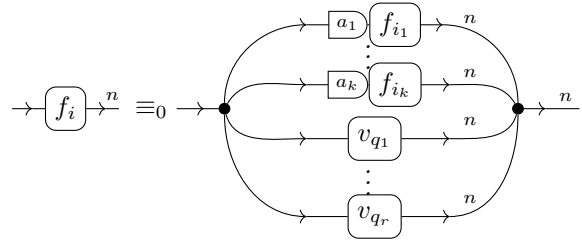
\dbox{f_i}{}{n} \;\disteq{0} 
\tikzset{x=1em, y=2.1ex}
\InputIfFileExists{state-eq.tikz}{}{\input{./tikz/state-eq.tikz}}
\tikzset{x=1em, y=1.5ex}

$$
where $f_i$, $1\leq i\leq m$ and $v_{q_j}$, $1\leq i\leq \ell$  are defined as above.
\end{lemma}
\begin{proof}
First,  by unrolling (Lemma~\ref{lem:unroll}),we have
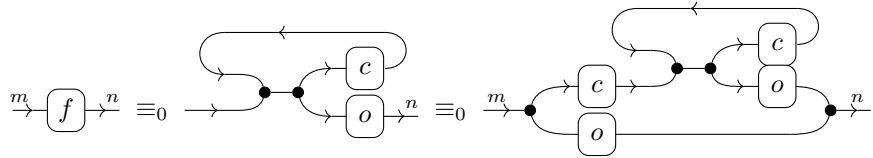
\begin{align*}
\dbox{f}{m}{n} \; \disteq{0} \;
\tikzset{x=1em, y=2.1ex}
\InputIfFileExists{automata-rep-all-states.tikz}{}{\input{./tikz/automata-rep-all-states.tikz}}
\tikzset{x=1em, y=1.5ex}
\;\disteq{0}  \;
\tikzset{x=1em, y=2.1ex}
\InputIfFileExists{automata-rep-all-states-1.tikz}{}{\input{./tikz/automata-rep-all-states-1.tikz}}
\tikzset{x=1em, y=1.5ex}

\end{align*}
Thus, for any $i\in\{1,\dots,m\}$, say $i=m$ for simplicity, we get
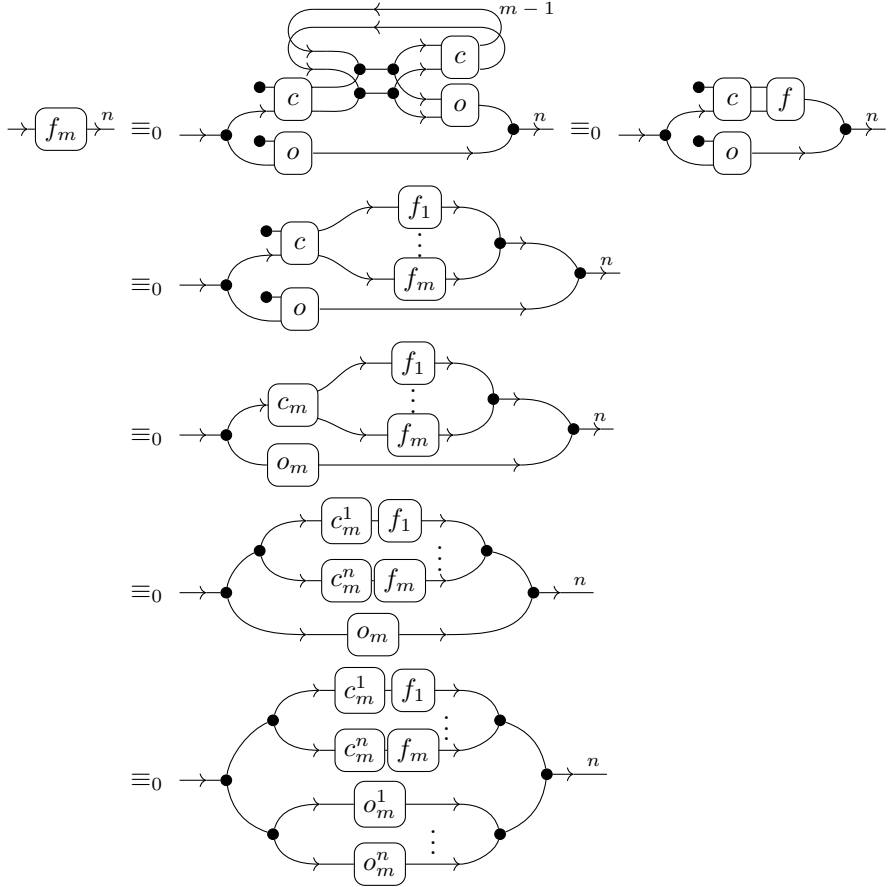
\begin{align*}
\dbox{f_m}{}{n} \;&\disteq{0} \;
\tikzset{x=1em, y=2.1ex}
\InputIfFileExists{automata-rep-last-state-1.tikz}{}{\input{./tikz/automata-rep-last-state-1.tikz}}
\tikzset{x=1em, y=1.5ex}
  \;\disteq{0} \;
\tikzset{x=1em, y=2.1ex}
\InputIfFileExists{automata-rep-last-state-2.tikz}{}{\input{./tikz/automata-rep-last-state-2.tikz}}
\tikzset{x=1em, y=1.5ex}
\\
& \disteq{0} \;
\tikzset{x=1em, y=2.1ex}
\InputIfFileExists{automata-rep-last-state-3.tikz}{}{\input{./tikz/automata-rep-last-state-3.tikz}}
\tikzset{x=1em, y=1.5ex}

\\
& \disteq{0} \;
\tikzset{x=1em, y=2.1ex}
\InputIfFileExists{automata-rep-last-state-4.tikz}{}{\input{./tikz/automata-rep-last-state-4.tikz}}
\tikzset{x=1em, y=1.5ex}

\\
& \disteq{0} \;
\tikzset{x=1em, y=2.1ex}
\InputIfFileExists{automata-rep-last-state-5.tikz}{}{\input{./tikz/automata-rep-last-state-5.tikz}}
\tikzset{x=1em, y=1.5ex}

\\
& \disteq{0} \;
\tikzset{x=1em, y=2.1ex}
\InputIfFileExists{automata-rep-last-state-6.tikz}{}{\input{./tikz/automata-rep-last-state-6.tikz}}
\tikzset{x=1em, y=1.5ex}

\end{align*}
where $c^i_m$ is either $\scalar{a}$ for some $a\in \Sigma$, when there is an $a$-transition connecting its only input wire to some $f_j$, or $\Bcounit \;\;\Bunit$ otherwise (recall that $c$ is guarded), and $o^i_m$ is either an identity, when the only input of $o_m$ is connected to some output wire, or $\Bcounit \;\;\Bunit$ otherwise. Since all $f_j$ connected to some $\Bcounit \;\;\Bunit$ can be removed (using co-deleting), we get the equality we wanted.
\end{proof}

\fundamental*
\begin{proof}
	Follows from \Cref{thm:representation} and \Cref{lem:diagram-to-eq-system}.
\end{proof}

\begin{lemma}
\label{lem:fixpoint}
For any $f\from \objr^m\to \objr^n$ and $f_i$, $1\leq i\leq m$ defined as above, for all $i\in\{1,\dots,m\}$, we have that
\[\sem{f_i}= \sum_{j=1}^k a_{j}.\sem{f_{i_j}} + \sum_{j=1}^{\ell} \sem{v_{q_k}} \quad \text{ if }\quad \dbox{f_i}{}{n} \;\disteq{0} 
\tikzset{x=1em, y=2.1ex}
\InputIfFileExists{state-eq.tikz}{}{\input{./tikz/state-eq.tikz}}
\tikzset{x=1em, y=1.5ex}
\]
where, for $1\leq j\leq \ell$,
each $v_{q_j}\from \objr\to \objr^n$ is a diagram encoding the output variables, as defined in \Cref{lem:fundamental}.
\end{lemma}
\begin{proof}
This is a consequence of Lemma~\ref{lem:sum-convolution}.
\end{proof}

\begin{remark}\label{rem:repr_transitions}
	Let $F$ be a set of string diagrams of the type $\objr \to \objr^n$. The set $\Sigma \times F + V_n$ is isomorphic to the set
	$$
	G = \left\{
\tikzset{x=1em, y=2.1ex}
\InputIfFileExists{prefix-transition.tikz}{}{\input{./tikz/prefix-transition.tikz}}
\tikzset{x=1em, y=1.5ex}
 \mid a \in \Sigma, f_j \in F \right\} \cup \left\{\dbox{v_s}{}{n} \mid 1 \leq s \leq n\right\}
	$$
\end{remark}
\begin{lemma}\label{lem:transition_approximation}
	Let $F$ be a set of string diagrams of the type $\objr \to \objr^n$ that is equipped with a $1$-bounded pseudometric $d_F \colon F \times F \to [0,1]$. Assume that for all $f_i, f_k \in F$, $\varepsilon \in \Qp$, such that $d_F(f_i, f_k)\leq \varepsilon$, we have that $\dbox{f_i}{}{n} \disteq{\varepsilon} \dbox{f_k}{}{n}$ is derivable. For all $g_u, g_v \in G$, with $G$ defined as above and all $\varepsilon \in \Qp$, such that $d_F^\uparrow(g_u, g_v) \leq \varepsilon$, we have that $\dbox{g_u}{}{n} \disteq{\varepsilon} \dbox{g_v}{}{n}$ is derivable.
\end{lemma}
\begin{proof}
Let $\varepsilon \geq d_F^\uparrow(g_u, g_v)$. First, consider the case, when $\dbox{g_u}{}{n} = 
\tikzset{x=1em, y=2.1ex}
\InputIfFileExists{prefix-transition-1.tikz}{}{\input{./tikz/prefix-transition-1.tikz}}
\tikzset{x=1em, y=1.5ex}
$ and $\dbox{g_v}{}{n} = 
\tikzset{x=1em, y=2.1ex}
\InputIfFileExists{prefix-transition-2.tikz}{}{\input{./tikz/prefix-transition-2.tikz}}
\tikzset{x=1em, y=1.5ex}
$. We have that $d_F^\uparrow(g_u, g_v)= \frac{1}{2} d_F(f_i, f_k)$ and hence $2 \varepsilon \geq d_F(f_i, f_k)$. By the assumption, we know that $\dbox{f_i}{}{n} \disteq{2\varepsilon} \dbox{f_k}{}{n}$ is derivable. Using \textsf{(Pref)}, we can derive $
\tikzset{x=1em, y=2.1ex}
\InputIfFileExists{prefix-transition-1.tikz}{}{\input{./tikz/prefix-transition-1.tikz}}
\tikzset{x=1em, y=1.5ex}
 \disteq{\varepsilon} 
\tikzset{x=1em, y=2.1ex}
\InputIfFileExists{prefix-transition-2.tikz}{}{\input{./tikz/prefix-transition-2.tikz}}
\tikzset{x=1em, y=1.5ex}
$, which is the same as $\dbox{g_u}{}{n} \disteq{\varepsilon} \dbox{g_v}{}{n}$. In all the remaining cases, $d_F^\uparrow$ behaves like a discrete pseudometric, hence there are two remaining subcases. In the situation when $\dbox{g_u}{}{n} = \dbox{g_v}{}{n}$, we have that $d_F^\uparrow(g_u, g_v)=0$ and hence we can derive $\dbox{g_u}{}{n} \disteq{\varepsilon} \dbox{g_v}{}{n}$ by first applying \textsf{(Refl)} and then \textsf{(Max)}. Othwerise, when $\dbox{g_u}{}{n} \neq \dbox{g_v}{}{n}$, we have that $d_F^\uparrow(g_u, g_v)=1$ and hence we can derive $\dbox{g_u}{}{n} \disteq{\varepsilon} \dbox{g_v}{}{n}$ by first applying \textsf{(Top)} and then \textsf{(Max)}.
\end{proof}
\begin{lemma}\label{lem:hausdorff_approx}
	Let $F$ be a finite set of string diagrams of the type $\objr \to \objr^n$ that is equipped with a $1$-bounded pseudometric $d_F \colon F \times F \to [0,1]$. Assume that for all $f_i, f_k \in F$, $\varepsilon \in \Qp$, such that $d_F(f_i, f_k)\leq \varepsilon$, we have that $\dbox{f_i}{}{n} \disteq{\varepsilon} \dbox{f_k}{}{n}$ is derivable. Then, for all $A,B \subseteq F$ and all $\varepsilon \in \Qp$, such that $\mathcal{H}(d_F)(A,B) \leq \varepsilon$, we have that $\dbox{R_A}{}{n} \disteq{\varepsilon} \dbox{R_B}{}{n}$ is also derivable.
\end{lemma}
\begin{proof} Pick an arbitrary $\varepsilon \in \Qp$, such that $\mathcal{H}(d_F)(A,B) \leq \varepsilon$.
	If $A=B=\emptyset$, then by the usage of \textsf{(Refl)} and \textsf{(Max)} we are done. Similarly, when only one of $A$ and $B$ is empty, that we can immediately obtain the desired result using \textsf{(Top)} and \textsf{(Max)} rules. From now on, we can safely assume that $A$ and $B$ are nonempty. Recall the characterisation of Hausdorff distance from \Cref{rem:hausdorff_duality}. One can easily observe that in the case when $A$ and $B$ are nonempty, the set $\Gamma(A,B)$ of relational couplings between $A$ and $B$ is nonempty and hence 
	$$
	\mathcal{H}(d_F)(A,B) = \min \left\{\sup_{(f_i,f_k) \in R} d_F(f_i,f_k) \mid R \in \Gamma(A,B)\right\} \leq \varepsilon
	$$
	There must exist some optimal coupling $R_{\min} \in \Gamma(A,B)$, which witnesses the above minimum. Hence, we have that $
	\sup_{(f_i, f_k) \in R_{\min}} d_F(f_i, f_k) \leq \varepsilon
	$, which in turn implies that $d_F(f_i, f_k) \leq \varepsilon$ for all $(f_i, f_k) \in R_{\min}$. Using the assumption, we know that for all pairs $(f_i, f_k) \in R_{\min}$, we have that
	$$
	\dbox{f_i}{1}{n} \disteq{\varepsilon} \dbox{f_k}{1}{n}
	$$
	For the sake of simplicity, assume that $R_{\min} = \{(f_{i_1}, f_{k_1}), \dots, (f_{i_j}, f_{k_j})\}$. Using the $\mathsf{(Tens)}$ rule we can stack in parallel all these pairs and obtain:
	$$
	
\tikzset{x=1em, y=2.1ex}
\InputIfFileExists{par-stack-hausdorff-l.tikz}{}{\input{./tikz/par-stack-hausdorff-l.tikz}}
\tikzset{x=1em, y=1.5ex}
 \disteq{\varepsilon} 
\tikzset{x=1em, y=2.1ex}
\InputIfFileExists{par-stack-hausdorff-r.tikz}{}{\input{./tikz/par-stack-hausdorff-r.tikz}}
\tikzset{x=1em, y=1.5ex}
 
	$$
	Using the \textsf{(Seq)} rule, we can derive that
	$$
	
\tikzset{x=1em, y=2.1ex}
\InputIfFileExists{repr-hausdorff-l.tikz}{}{\input{./tikz/repr-hausdorff-l.tikz}}
\tikzset{x=1em, y=1.5ex}
 \disteq{\varepsilon} 
\tikzset{x=1em, y=2.1ex}
\InputIfFileExists{repr-hausdorff-r.tikz}{}{\input{./tikz/repr-hausdorff-r.tikz}}
\tikzset{x=1em, y=1.5ex}
 
	$$
	From the definition of relational couplings, we have that $\pi_1(R_{\min}) = A$ and $\pi_2(R_{\min}) = B$ and hence the diagrams above are convolutions of the sets $A$ and $B$ respectively. This allows us to conclude that $\dbox{R_A}{1}{n} \disteq{\varepsilon} \dbox{R_B}{1}{n}$ is derivable.
\end{proof}

\approximation*
\begin{proof}
	Pick an arbitrary $f_g,f_h \in Q_f$.
	By induction on $p$. When $p=0$, $\Phi_{\beta}^{(p)}$ is a discrete pseudometric on the set $Q_f$ and hence for all $\varepsilon \geq \Phi_{\beta}^{(p)}(f_g, f_h)$, we can derive $\dbox{f_g}{}{n} \disteq{\varepsilon} \dbox{f_h}{}{n}$ using \textsf{(Refl)}, \textsf{(Top)} and \textsf{(Max)} rules, similarly to the proof of \Cref{lem:transition_approximation}.	
	For the induction step, when $p = p' + 1$. Recall that $\Phi_\beta^{p'+1}(f_g, f_h) = \mathcal{H}\left({\Phi_\beta^{(p')}}^\uparrow\right)(\beta(f_g), \beta(f_h))$. Pick an arbitrary $\varepsilon \geq \mathcal{H}\left({\Phi_\beta^{(p')}}^\uparrow\right)(\beta(f_g), \beta(f_h))$. In order to derive that $\dbox{f_g}{}{n} \disteq{\varepsilon} \dbox{f_h}{}{n}$, we will rely on \Cref{lem:hausdorff_approx}. In order to use it, we need to be able to derive approximations to the distance given by ${\Phi_\beta^{(p')}}^\uparrow$ on the string diagrams representing the elements of the set $\Sigma \times Q_f + V_n$ (see \Cref{rem:repr_transitions}). For this we will use \Cref{lem:transition_approximation}, which requires that for all $f_{g'}, f_{h'} \in Q_f$, $\varepsilon' \geq \Phi^{(p')}_\beta$ one can derive that $\dbox{f_{g'} }{}{n} \disteq{\varepsilon'}\dbox{f_{h'}}{}{n}$. This in turn is guaranteed by the induction hypothesis, which completes the proof.
\end{proof}
\begin{restatable}{corollary}{extractedhomom}\label{lem:extracted_homom}
	A function mapping each state $f_i \in Q_d$ to $\sem{f_i}$ is a prechart homomorphism from $(Q_d, \beta)$ to $\Omega$ 
\end{restatable}
\begin{proof}
	Immediately follows from \Cref{lem:fixpoint} and the definition on transition structure on ${\Expr}/{\sim}$ (given by \Cref{lem:quotient_chart}). Essentially, homomorphisms are maps that preserve and reflect prechart transitions~\cite[Example~2.1]{Rutten:2000:Universal} and $(Q_d, \beta)$ is precisely defined to satisfy this.
\end{proof}
\extracteddist*
\begin{proof}
	Follows from \Cref{lem:extracted_homom} and \Cref{thm:beh_dist}.
\end{proof}

\onemcompleteness*
\begin{proof}
	Let $\dbox{f_g}{}{n} \disteq{\varepsilon} \dbox{f_h}{}{n}$ be valid, that is $\mathsf{bd}_{\overline \partial}(\sem{f_g}, \sem{f_h}) \leq \varepsilon$.
	We will rely on \textsf{(Cont)} rule. In order to deduce that $\dbox{f_g}{}{n} \disteq{\varepsilon} \dbox{f_h}{}{n}$ we need to show that for all $\varepsilon' > \varepsilon$, we have that $\dbox{f_g}{}{n} \disteq{\varepsilon'} \dbox{f_h}{}{n}$ is derivable. Since $\varepsilon' > \varepsilon$, we have that $\mathsf{bd}_{\overline \partial}(\sem{f_g}, \sem{f_h}) < \varepsilon'$. Because of \Cref{cor:dist-chart-diag} and \Cref{cor:kleene_locally_finite}, we have that
	$$
		\inf_{p \in \N}\left\{\Phi^{(p)}_\beta(f_g, f_h)\right\} < \varepsilon'
	$$
	We will argue that there exists $p \in \N$, such that $\Phi^{(p)}_\beta(f_g,f_h) < \varepsilon'$. For the sake of an argument, assume that for all $p \in \N$, we have that $\Phi^{(p)}_\beta(f_g,f_h) \geq \varepsilon'$. This would make $\varepsilon'$ into the lower bound of the $\omega$-cochain $\left\{\Phi^{(p)}_\beta(f_g, f_h)\right\}_{p \in \N}$ and hence $\varepsilon' \leq 	\inf_{p \in \N}\left\{\Phi^{(p)}_\beta(f_g, f_h)\right\} < \varepsilon'$, which leads to contradiction. Combining that argument with \Cref{lem:approximation} allows us to conclude that $\dbox{f_g}{}{n} \disteq{\varepsilon'} \dbox{f_h}{}{n}$ is derivable, which completes the proof. 
\end{proof}
\completeness*
\begin{proof}
	Assume that $\dbox{f}{m}{n} \disteq{\varepsilon} \dbox{g}{m}{n}$ is valid.
	Recall that because of \Cref{lem:injections}, we have that
	\begin{gather*}
		\dbox{f}{m}{n} \disteq{0} 
\tikzset{x=1em, y=2.1ex}
\InputIfFileExists{d-decomposition.tikz}{}{\input{./tikz/d-decomposition.tikz}}
\tikzset{x=1em, y=1.5ex}
 \qquad \text{and}\qquad \dbox{g}{m}{n} \disteq{0} 
\tikzset{x=1em, y=2.1ex}
\InputIfFileExists{d-decomposition-r.tikz}{}{\input{./tikz/d-decomposition-r.tikz}}
\tikzset{x=1em, y=1.5ex}
  
	\end{gather*}
	Assume that $\sem{f_i} = N(s_i)$ and $\sem{g_i}=N(t_i)$ for $1 \leq i \leq m$. Because of \Cref{lem:distance_on_merge}, we have that $d^{1,n}(s_i, t_i) \leq \varepsilon$.
	We will consider the following diagram
	$$
\tikzset{x=1em, y=2.1ex}
\InputIfFileExists{comp-par-lr.tikz}{}{\input{./tikz/comp-par-lr.tikz}}
\tikzset{x=1em, y=1.5ex}
 \disteq{0} 
\tikzset{x=1em, y=2.1ex}
\InputIfFileExists{comp-par-states.tikz}{}{\input{./tikz/comp-par-states.tikz}}
\tikzset{x=1em, y=1.5ex}
$$
	Using \Cref{lem:fundamental}, we can show that each of the $\objr \to \objr^{m}$ subdiagrams is in the form allowing to use \Cref{lem:one-to-m-completeness}. In turn, that lemma allows to derive any valid equations between the subdiagrams mentioned above. In particular, we can derive that $\diagbox{f_i}{}{m} \disteq{\varepsilon} \diagbox{g_i}{}{m}$ for all $1 \leq i \leq m$. We can use \textsf{(Tens)} rule to derive
	$$
		
\tikzset{x=1em, y=2.1ex}
\InputIfFileExists{par-stack-completeness-l.tikz}{}{\input{./tikz/par-stack-completeness-l.tikz}}
\tikzset{x=1em, y=1.5ex}
 \disteq{\varepsilon} 
\tikzset{x=1em, y=2.1ex}
\InputIfFileExists{par-stack-completeness-r.tikz}{}{\input{./tikz/par-stack-completeness-r.tikz}}
\tikzset{x=1em, y=1.5ex}

	$$
	We can then apply \textsf{(Seq)} to postcompose co-copying to the diagrams above to obtain
	$$
	 
\tikzset{x=1em, y=2.1ex}
\InputIfFileExists{d-decomposition.tikz}{}{\input{./tikz/d-decomposition.tikz}}
\tikzset{x=1em, y=1.5ex}
 \disteq{\varepsilon}  
\tikzset{x=1em, y=2.1ex}
\InputIfFileExists{d-decomposition-r.tikz}{}{\input{./tikz/d-decomposition-r.tikz}}
\tikzset{x=1em, y=1.5ex}

	$$
	By previous reasoning and \textsf{(Triang)} rule this is the same as
	$$
	\diagbox{f}{m}{n} \disteq{\varepsilon} \diagbox{g}{m}{n}
	$$
	which completes the proof.
\end{proof}

\compact*
\begin{proof}
The lemma holds in any compact closed category and relies on the ability to bend wires using $
\tikzset{x=1em, y=2.1ex}
\InputIfFileExists{cap-down.tikz}{}{\input{./tikz/cap-down.tikz}}
\tikzset{x=1em, y=1.5ex}
$ and $
\tikzset{x=1em, y=2.1ex}
\InputIfFileExists{cup-down.tikz}{}{\input{./tikz/cup-down.tikz}}
\tikzset{x=1em, y=1.5ex}
$. Explicitly, given a diagram of the first form, we can obtain one of the second form as follows:
\begin{equation*}

\tikzset{x=1em, y=2.1ex}
\InputIfFileExists{wrong-way-left.tikz}{}{\input{./tikz/wrong-way-left.tikz}}
\tikzset{x=1em, y=1.5ex}
\quad \mapsto\quad 
\tikzset{x=1em, y=2.1ex}
\InputIfFileExists{bent-wires.tikz}{}{\input{./tikz/bent-wires.tikz}}
\tikzset{x=1em, y=1.5ex}

\end{equation*}
The inverse mapping is given by the same wiring with the opposite direction. That they are inverse transformations follows immediately from the defining axioms of compact closed categories (A1-A2). 
\begin{equation*}
 
\tikzset{x=1em, y=2.1ex}
\InputIfFileExists{bent-wires.tikz}{}{\input{./tikz/bent-wires.tikz}}
\tikzset{x=1em, y=1.5ex}
 \quad\mapsto\quad  
\tikzset{x=1em, y=2.1ex}
\InputIfFileExists{unbent-wires.tikz}{}{\input{./tikz/unbent-wires.tikz}}
\tikzset{x=1em, y=1.5ex}
 \myeq{A1} 
\tikzset{x=1em, y=2.1ex}
\InputIfFileExists{wrong-way-left.tikz}{}{\input{./tikz/wrong-way-left.tikz}}
\tikzset{x=1em, y=1.5ex}

\end{equation*}
The other bijection is constructed analogously. 
\end{proof}
Intuitively, Lemma~\ref{lem:compact} tells us that we can always bend incoming wires to the left and outgoing wires to the right to obtain a $\objr^m\to\objr^n$ diagram from any given diagram.
Let $S(f)\from\objr^m\to\objr^n$ be the diagram obtained by applying the bijections of Lemma~\ref{lem:compact} to a diagram $f\from v\to w$ until all the objects occurring in its domain and codomain are $\objr$. 
\begin{lemma}\label{lem:right-to-right}
Given two diagrams $f,g\from v\to w$, $S(f)=S(g)$ iff $f=g$. 
\end{lemma}
\begin{proof}
The idea is that, if $S(f)=S(g)$, we can always show that $f=g$ using a similar derivation, by simply applying the transformation of Lemma~\ref{lem:compact} before using the derivation that $S(f)=S(g)$, and then recover the original orientation of the wires by bending them back into their original place afterwards; and the same idea applies to show that $f=g$ implies $S(f)=S(g)$. 
\end{proof}
\generalcompleteness*
\begin{proof}
	It is not too hard to see that a mapping $\sem{f} \to \sem{S(f)}$ is an isometry. Namely, given a pair $f,g \colon v \to w$ we obtain $S(f)$ and $S(g)$ by postcomposing 
\tikzset{x=1em, y=2.1ex}
\InputIfFileExists{cap-down.tikz}{}{\input{./tikz/cap-down.tikz}}
\tikzset{x=1em, y=1.5ex}
 (or precomposing 
\tikzset{x=1em, y=2.1ex}
\InputIfFileExists{cup-down.tikz}{}{\input{./tikz/cup-down.tikz}}
\tikzset{x=1em, y=1.5ex}
) which because of \Cref{cor:sem_enriched} preserves distances. Hence, if $\diagbox{f}{v}{w} \disteq{\varepsilon} \diagbox{g}{v}{w}$ is valid, so is $\dbox{S(f)}{m}{n} \disteq{\varepsilon} \dbox{S(g)}{m}{n}$. The rest follows as a consequence of \Cref{lem:mncompleteness}.
\end{proof}

\end{document}